\documentclass[11pt]{article}

\usepackage[T1]{fontenc}

\usepackage[a4paper, total={7in, 8in}]{geometry}

\usepackage{amssymb,amsmath}
\usepackage{color}
\usepackage{graphicx,graphics}
\usepackage{mathtools}
\usepackage{bbold}
\usepackage[english]{babel}
\usepackage[utf8]{inputenc}
\usepackage{epsfig,url}
\usepackage{bbm,theorem}
\usepackage{enumerate}
\usepackage{relsize}
\usepackage{physics}
\usepackage[square,sort,comma,numbers]{natbib}
\usepackage{tikz}
\usetikzlibrary{decorations.pathreplacing,decorations.markings,snakes}

\usepackage{hyperref}

\usepackage{mathrsfs}

\usepackage{xcolor}

\usepackage{verbatim}

\newtheorem{theorem}{Theorem}[section]
\newtheorem{definition}[theorem]{Definition}

\newtheorem{lemma}[theorem]{Lemma}
\newtheorem{proposition}[theorem]{Proposition}
\newtheorem{assumption}[theorem]{Assumption}

\newtheorem{notation}[theorem]{Notation}

\def\w#1{\mathop{:}\nolimits\!#1\!\mathop{:}\nolimits}

{\theorembodyfont{\upshape}
\newtheorem{remark}[theorem]{Remark}
\newtheorem{example}[theorem]{Example}
}
\numberwithin{equation}{section}
\numberwithin{theorem}{section}

\newcommand{\qed}{\hfill$\Box$}
\newenvironment{proof}{\begin{trivlist}\item[]{\em Proof:}\/}{%
\qed\end{trivlist}}
\newenvironment{proofof}[1]{%
\begin{trivlist}\item[]{\em Proof of #1}\ }{\qed\end{trivlist}}

\newcommand{\Nsphere}{S^{N-1}(\sqrt{N})}

\newcommand{\E}{{\mathbb E}}

\newcommand{\R}{{\mathbb R}}
\newcommand{\C}{{\mathbb C\hspace{0.05 ex}}}
\newcommand{\N}{{\mathbb N}}

\newcommand{\cf}[1]{{\mathbbm 1}_{\{#1\}}}

\renewcommand{\epsilon}{\varepsilon}

\renewcommand{\abs}[1]{\left| #1 \right|}

\newcommand{\rme}{{\rm e}}
\newcommand{\rmd}{{\rm d}}

\renewcommand{\norm}[1]{\Vert #1\Vert}

\newcommand{\set}[1]{\{#1\}}

\newcommand{\mean}[1]{\langle #1\rangle}

\newcommand{\vep}{\varepsilon}
\newcommand{\defem}[1]{{\em #1\/}}

\usepackage{mathrsfs}

\newcommand{\wick}[1]{\mathlarger{:}#1\mathlarger{:}}

\newcommand{\inner}[2]{\langle #1, #2 \rangle}

\newcommand{\len}{\mathrm{len}}

\newcommand{\email}[1]{E-mail: \tt #1}
\newcommand{\emailjani}{\email{jani.lukkarinen@helsinki.fi}}
\newcommand{\emailaleksis}{\email{aleksis.vuoksenmaa@helsinki.fi}}
\newcommand{\UHaddress}{\em University of Helsinki, Department of Mathematics 
and Statistics\\
\em P.O. Box 68, FI-00014 Helsingin yliopisto, Finland}

\newcounter{jlisti}

\begin{document}

\title{Generation of chaos in the cumulant hierarchy of the stochastic Kac model}
\author{Jani Lukkarinen\thanks{\emailjani}, %
Aleksis Vuoksenmaa\thanks{\emailaleksis}
\\[1em]%
$\,^*,\,^\dag$\UHaddress}
\date{\today}

\maketitle

\abstract{We study the time-evolution of cumulants of velocities and kinetic energies in the stochastic Kac model for velocity exchange of $N$ particles, with the aim of quantifying how fast these degrees of freedom become chaotic in a time scale in which the collision rate for each particle is order one.   Chaos here is understood in the sense of the original Sto\ss zahlansatz, as an almost complete independence of the particle velocities which we measure by the magnitude of their cumulants up to a finite, but arbitrary order.
Known spectral gap results imply that typical initial densities converge to uniform distribution on the constant energy sphere at a time which has order of $N$ expected collisions.
We prove that the finite order cumulants converge to their small stationary values much faster, already at a time scale of order one  collisions, and we state how this convergence can be controlled via solutions to the kinetic equation of this model.
The proof relies on stability analysis of the closed, but nonlinear, hierarchy of energy cumulants around the fixed point formed by their values in the stationary spherical distribution. It provides the first example of an application of the cumulant hierarchy method to control the properties of a microscopic model related to kinetic theory.}

\section{Introduction}

The stochastic Kac model for velocity exchange was introduced in 1956 by Mark Kac \cite{kac_foundations_1956}, as a toy model to study the emergence and accuracy of kinetic theory as a description of large scale properties of a particle system. This is a stochastic mean-field model and, as such, amenable to various techniques, coming both from physics and probability theory.
We will define the model in full detail and discuss its known properties in Section \ref{sec:Kacdefinition}. We start by presenting its relation to the kinetic theory of a rarefied gas of particles which serves as one important motivation for the present work.

In the original work by Boltzmann, the collision operator in the kinetic equation was derived assuming
that the collisions between particles are well separated in time and space, so that only collisions between two particles need to be considered in detail.
More precisely, it is assumed that the particles otherwise move freely, apart from brief moments in time in which a pair of incoming particle velocities $(v_i,v_j)$, $i,j$ the labels for the colliding particles, are replaced by outgoing velocities
$(v'_i,v'_j)$, according to the rules of elastic collisions between these particles.

The main assumption about these collisions was that the incoming velocities $v_i,v_j$ are \defem{independent} random variables.
This assumption is now called
the \defem{molecular chaos hypothesis}, or \defem{Sto\ss zahlansatz}. Albeit perhaps intuitively believable, there is an inherent problem in the above strict molecular chaos hypothesis: the outgoing velocities $(v'_i,v'_j)$ cannot be statistically independent if the incoming velocities are that.  Hence, if these particles can have a recollision later in time, the molecular chaos hypothesis is in doubt, and although the recollision probability is expected to be small in these systems, it is not zero for the original rarefied gas model.

This has lead to the following \defem{propagation of chaos} problem:
if in the initial state the particle velocities are nearly independent, will they be so also later in time?  If yes, for how long and can we quantify how ``nearly''?

The stochastic Kac model allows to focus on one part of the original problem by ignoring all spatial structure
and most correlations introduced by the collisions in the original model. More precisely, each particle has only velocity as its dynamical variable, we assume that collision times are determined by a random Poisson distributed clock, and the pair of colliding particles are chosen randomly.  In a collision, the outgoing velocities are otherwise completely random but we do require that the collision conserves kinetic energy, i.e., we require that $(v'_i)^2+ (v'_j)^2=v_i^2+ v_j^2$.  As a consequence, the total kinetic energy of particles, $\sum_{i=1}^Nv_i^2$ is still a constant of motion in the stochastic Kac model.

It is quite remarkable, as proven already by Kac,
that if the system is started from a state which is close to a product state for velocities but which has a deterministic total energy fixed to $N$, then the
evolution does remain close to a product state and the velocity marginal distribution evolves according to kinetic equation which is similar to the
spatially homogeneous Boltzmann equation, at least up to times of order $N$.

It took several decades to reach similar conclusions for the original rarefied gas system, in an analogous Boltzmann--Grad scaling limit \cite{lenard_entropy_1973} --- a review of the current status for rarefied gas models can be found from \cite{gallagher_newton_2014}.  Significant progress has also been made for weakly nonlinear wave equations, such as for the weakly nonlinear Schr\"odinger equation \cite{deng_full_2023}, \cite{deng_propagation_2024}.

The main method to study propagation of chaos in the above deterministic examples comes from hierarchical evolution equations for moments, such as the BBGKY hierarchy for classical particle system. This is a delicate task, as almost independence manifests itself as an almost factorization property for moments which is difficult to propagate. These technical difficulties are also reflected in the length of such studies, often running into more than hundred pages.

Already for a while now, our group has advocated the idea that cumulants would serve as an easier quantifier of statistical independence than moments.
Namely, the cumulant of a finite collection of random variables is zero already if any one of the variables is independent from the rest.  This allows to use smallness of joint cumulants to quantify also those cases in which the random variables are almost, but not completely, independent.
It was also observed in \cite{lukkarinen_wick_2016} that
a regularization of powers of random variables known as Wick polynomials could be a useful tool in derivation of evolution equations for cumulants.

Indeed, what we call cumulants here have already been used in statistical physics, as well, under the names of Ursell functions or connected correlation functions.  Decay properties of cumulants of initial data have also been an important technical tool in studies of moment hierarchies of dynamical systems with random initial data \cite{erdos_linear_2000}, \cite{lukkarinen_weakly_2011}, and lately also for control of fluctuations in the rarefied gas case above \cite{gallagher_newton_2014}.

Main motivation for the present work has been to check how far can we push the cumulant hierarchy techniques in the study of chaoticity in the stochastic Kac model. We expect that the ideas and techniques presented here will also facilitate the control of cumulant hierarchies in deterministic models with random initial data. 
However, additional challenges will be posed by mechanisms which generate dependence between the relevant random variables. Such mechanisms are provided by practically all Hamiltonian dynamics. Examples include recollisions in a rarefied gas, and constructive interference by wave-like evolution, such as the nonlinear Schrödinger equation.

Due to the strongly stochastic mean-field nature of the present dynamics, we are able to prove much more than just propagation of chaos. Namely, also \defem{generation of chaos} holds in this system: even if there are significant correlations between the energies initially in the system, they will decrease with time.  Our main result is summarized in Theorems \ref{thm:generation-of-alpha-chaotic-bounds} and  \ref{thm:convergence-to-stationary-non-chaotic} in which we conclude that such correlations decrease exponentially fast and converge to their small stationary values, at a rate which is often uniformly order one in the number of particles, $N$, and never more than proportional to $\ln N$.

As a further generalization, we can extend the earlier proofs of propagation of chaos to larger class of initial data, including deterministic data which are not absolutely continuous with respect to the stationary measure.
This is done without resorting to scaling limits but we instead provide error estimates for any sufficiently large systems and times.  In particular, we derive in Theorem \ref{thm:kinetic-accuracy} explicit estimates for how well the solutions to the Boltzmann--Kac kinetic evolution equation describe the evolution of energy cumulants, after the state has become chaotic.

The proof is based on finding suitable norms and iterative structures in the cumulant hierarchy of the system.  The system has a unique stationary state, and also the cumulant hierarchy has a unique fixed  point.  We find a formulation of the cumulant hierarchy which will allow performing a standard stable manifold inspection around this fixed point, with an explicit control over the exponential convergence to the fixed point values.  As a by-product, we thus also obtain quantitative estimates for how fast the cumulants equilibrate in the stochastic Kac model.

Since all these results are obtained by using cumulant hierarchy analysis, derived with the help of Wick polynomials, we can conclude that these techniques can facilitate the analysis of the propagation and generation of chaos, at least in the stochastic toy model considered here.  It should also be stressed that, not only we can relax some of the assumptions, but also the derivation of the these results takes less than half of the usual number of pages associated with such kinetic theory computations.

The paper is organized as follows: Firstly, we will introduce the stochastic Kac model and comment on some of the known results concerning it. In section \ref{sec:Kacdefinition}, we define the velocity and energy cumulants and write down their evolution time-evolution, as specified by the Kac master equation. We then explain how the symmetry of the underlying measure on \(S^{N-1}(\sqrt{N})\) is reflected on the level of joint cumulants. We define a family of norms on the spaces of joint cumulants. These norms aim to describe the degree of chaoticity of states of the model. Theorems \ref{thm:generation-of-alpha-chaotic-bounds} and \ref{thm:convergence-to-stationary-non-chaotic} describe the time-evolution of the joint cumulants and the behavior of their norm in these spaces. In section \ref{sec:initdata}, we describe different kinds of initial data and establish bounds for the joint energy cumulants when the underlying probability measure is symmetric. In section \ref{sec:Nonrep}, we solve the time-evolution of an important subclass of energy cumulants consisting joint cumulants with no repeated particle indices. Section \ref{sec:Repeated} contains the proofs of Theorems \ref{thm:generation-of-alpha-chaotic-bounds} and \ref{thm:convergence-to-stationary-non-chaotic}, and in the final Section \ref{sec:KacBEaccuracy} we discuss the connection to the Boltzmann--Kac kinetic equation, and prove Theorem \ref{thm:kinetic-accuracy} about its accuracy.

\subsection{Stochastic Kac model}\label{sec:Kacdefinition}

In this section, we recall the definition of  the stochastic model introduced by Kac in his 1956 article \cite{kac_foundations_1956}, and
discuss earlier results about the properties of its evolution. We will follow the usual conventions and choose the parameters of this model based on computational convenience instead of the greatest physical generality.  These choices will also be explained in this section.

Take a number \(N \geq 2\) of particles, which we label by \([N] \coloneqq \{1,\dots,N\}\). At each moment in time, the particle with label \(i \in [N]\) is associated with one-dimensional velocity \(v_i \in \R\), and the collective velocities of all the particles form a velocity vector \(v = (v_1,\dots ,v_N) \in \R^{N}\).
Each velocity vector constitutes a configuration of the system.
Originally, Kac consider the physical case with three-dimensional velocities, $v_i\in \R^3$, but this would lead to more complex collision rules and indexing of the configuration vectors.  Here, as in most mathematical work on this model, we focus on the one-dimensional case.
 As noted earlier, we do not track the positions of the particles and accordingly they are assumed to have no influence on the time-evolution of the velocities.

The configurations are updated as follows. Assume that the velocities are at some time given by \(v \in \R^N\).
As we move to the next instant, we update the velocity vector by picking, uniformly at random, two particles \(i\) and \(j\), with $j\ne i$, and a ``collision'' angle \(\theta \in (-\pi, \pi]\), all independent from each other. We then update the velocity vector to
\[v^* \coloneqq R_{i,j}(\theta)v \in \R^{N}.\]

	Here the matrix \(R_{i,j}(\theta)\) is obtained as a permutation of the following \(N \times N\)-matrix 
	\begin{align}
	R_{1,2}(\theta) &\coloneqq \begin{pmatrix}
	\cos(\theta) & \sin(\theta) \\
	-\sin(\theta) & \cos(\theta)
	\end{pmatrix} \oplus I_{N-2}\,,
	\end{align}
	where $I_n$ denotes the $n$-dimensional identity operator.
	The non-zero elements of the matrix \(R_{i,j}(\theta)\) are given by
	\begin{align*}
		(R_{i,j}(\theta))_{\ell,\ell} &= 1, \quad \ell \neq i,j \\
		(R_{i,j}(\theta))_{i,i} &= \cos(\theta) \\
		(R_{i,j}(\theta))_{i,j} &= \sin(\theta) \\
		(R_{i,j}(\theta))_{j,i} &= -\sin(\theta) \\
		(R_{i,j}(\theta))_{j,j} &= \cos(\theta)\, .
	\end{align*}
	For any fixed \(i,j\) and \(\theta\), the matrix \(R_{i,j}(\theta)\) is orthogonal. Thus the Euclidean norm of the velocity vectors is preserved,
	\begin{align}
		\abs{R_{i,j}(\theta)v} = \abs{v}.
	\end{align} 
	Therefore, we can consider the dynamics as taking place on a (hyper-)sphere \(S^{N-1}(r)\),
	with the radius \(r\) specified by the Euclidean norm of the initial configuration.

	The above specification of the dynamics on the configuration space lifts to the space observables 
	\(\phi \in C_c(S^{N-1}(r))\) in terms of the transition operator, which specifies the conditional expectation of \(\phi(v^*)\) based on the previous configuration:
	\[(Q_N \phi)(v) \coloneqq \frac{1}{N(N-1)}\sum_{i,j=1}^{N} \cf{i\neq j}\int_{-\pi}^{\pi} \frac{\rmd \theta}{2\pi}\phi(R_{i,j}(\theta)v).\]
	Kac's stochastic model is now obtained by assuming that the above collision times are given  by a Poisson process.  As usual,
	we want to consider time scales at which the average rate of collisions experienced by one particle remains order one, i.e., it is bounded from below and above by some uniform constants for all relevant \(N\).  If the Poisson process has rate one, this can be achieved by speeding up time by a factor of \(N\).  For the sake of convenience, we follow the usual convention for time-scales in this model and choose the scale in which the average collision rate per particle is equal to one.

	Since each \(R_{i,j}(\theta)\) is orthogonal and leaves the sphere \(S^{N-1}(r)\) invariant, the operator \(Q_N\)  is a self-adjoint operator on the Hilbert space \(L^2(S^{N-1}(r),\nu_{N,r})\), where \(\nu_{N,r}\) is the uniform probability measure on \(S^{N-1}(r)\); note that $\nu_{N,r}$ is also invariant under rotations.
	Since we are interested in the properties of the system in the thermodynamic limit, we will pick \(r = \sqrt{N}\), and \(\nu_N\) stands for the corresponding uniform measure \(\nu_{N,\sqrt{N}}\).  This choice yields configurations for which the kinetic energy per particle,
	\(\frac{1}{N}\sum_{i=1}^N v_i^2 \), is  almost surely equal to one; note that we choose each particle to have the same mass, equal to two, so that $\frac{m}{2}=1$.

    Implementing the above choices, we can now
    state the evolution equation for probability densities $f_t^N(v)$, $v\in \Nsphere$, $t\ge 0$.
    Namely, if \(f^N_0 \in L^2(\Nsphere
    ,\nu_{N})\) is non-negative and integrates to one, we can consider initial data for the stochastic Kac process given by the measure $f^N_0\nu_{N}$.  Then, the distribution of $v(t)$, $t>0$, is given by $f^N_t\nu_{N}$ where also
    \(f^N_t \in L^2(\Nsphere
    ,\nu_{N})\) and it can be solved from the  following Cauchy problem:
	\begin{align}
		\begin{cases}
		\dv{t} f_t^N &= N(Q_N - I)f_t^N \\
		f_t^{N}\vert_{t=0} &= f_0^N. 
		\end{cases}
		\label{eq:Kac-stochastic-model}
	\end{align}
    More details about the derivation of this equation can be found from \cite{carlen_entropy_2010,mischler_kacs_2013} and the references therein.

The equation  \eqref{eq:Kac-stochastic-model} is called the  \emph{Kac Master Equation}, and it has been studied extensively.
	In conjunction with its introduction in \cite{kac_foundations_1956}, Kac showed that the evolution propagates approximate tensorization, meaning that if the marginals \(\Pi^k[f_0\nu_N]\), for a fixed number of variables $k$, of the initial measure converge to tensor products of a one-particle measure, then the time-evolved measure satisfies the same property. Moreover, the limit $g_t$ of the first marginals of the measure was then shown to satisfy the homogeneous Boltzmann--Kac equation
	\begin{align}
		\label{eq:boltzmann-kac}
		\partial_t g_t(v) &= C[g_t,g_t](v),
	\end{align}
	with the bilinear collision operator \(C[f,g]\) given by
	\begin{align}
		C[f,g](v) = 2\int_{-\pi}^{\pi} \int_{\R} \left(f(\cos(\theta)v - \sin(\theta)w)g(\sin(\theta)v + \cos(\theta)w) - f(v)g(w)\right)\rmd w \frac{\rmd \theta}{2\pi}\,.
	\end{align}
	Kac's stochastic equation thus provides an early example of a stochastic system from which a kinetic equation could be derived in a mathematically rigorous way.

	The pioneering work  at the turn of the millennium established that the operator \(N(Q_N-I)\) has an \(N\)-independent gap in its spectrum. Janvresse \cite{janvresse_spectral_2001} used martingale methods to prove that there is \emph{some} spectral gap.  Carlen et al.\ \cite{carlen_determination_2003} pinpointed the size of the gap uniform in \(N\) to be \(\frac{1}{2}\). The same precise value for the gap was reached by Maslen \cite{maslen_eigenvalues_2003} with the help of different techniques around the same time.
	This shows that the density converges to the uniform density exponentially fast, at a rate \(\frac{1}{2}\),
	\begin{align}
		\norm{f_t^N - 1}_{L^2} \leq \rme^{-\frac{t}{2}}\norm{f_0^N -1}_{L^2}\,.
	\end{align}
	Thus, in the context of the Hilbert space \(L^2(\Nsphere)\), the generator of the dynamics is well understood and the semigroup is contractive, with the solutions converging towards the unique stationary solution exponentially fast.  This also yields an  estimate of convergence to equilibrium for any observable $\phi\in L^2(\Nsphere)$:
	\[
	 |\mean{\phi(v(t))}- \mean{\phi}_{\nu_N}|
	 = |\mean{(f_t^N - 1)\phi}_{\nu_N}|
	 \le \norm{f_t^N - 1} \norm{\phi}
	 \le \rme^{-\frac{t}{2}}\norm{f_0^N -1}
	  \norm{\phi}\,.
	\]

	However, for a large subclass of chaotic initial data, like \(f^{N}_0\) obtained by conditioning \(f^{\otimes N}\) onto the hypersphere and with \(f\) satisfying certain moment and concentration properties, the initial \(L^2(\Nsphere)\) distance from the uniform density is exponential in \(N\), simply because the size of the initial data grows exponentially in \(N\).  Namely,
\begin{align}
\norm{f^N_0 - 1}_{L^2} \geq \norm{f^N_0}_{L^2} - 1 \approx C^{N}\,,
\end{align}
where the constant $C$ often is greater than, and not close to, one; this will be discussed in more detail in Sec.\ \ref{sec:intialKac}.
It follows that even though the decay rate of the \(L^2\)-norm is exponential, it is still necessary to wait for times that are \emph{linear} in \(N\) in order for the above spectral gap estimate to imply that all observables $\phi\in L^2(\Nsphere)$ have expectations close to their equilibrium values.

To study if such slow convergence holds for typical thermodynamic quantities, much of the later work has focused on the (relative) \emph{entropy production} of the system. Via Pinsker's inequality, relative entropy gives bounds on the total variation distance to the uniform measure on \(\Nsphere\), and thus estimates on the entropy production provide a way of studying the convergence of the time-evolved measure on \(\Nsphere\) towards the uniform measure. Carlen et al.\ \cite{carlen_entropy_2010} proved that there is no \(N\)-uniform lower bound on the entropy production. This result was improved by Einav \cite{einav_villanis_2011}, who proved that the rate of entropy production is almost as bad as \(\frac{1}{N}\), which means that one has to wait for times that are \emph{almost} linear in the number of particles -- a result that gives relaxation times similar to the ones established by the spectral gap results. Around the same time, Hauray and Mischler \cite{hauray_kacs_2014} provided quantitative estimates for the chaos of a probability measure on submanifolds of \(\R^N\) defined by conservation laws, while Mischler and Mouhot \cite{mischler_kacs_2013} proved quantitative propagation of chaos for related models. Both works rely on delicate analysis of the relative entropy functional.

Coupling methods have been proven useful in estimating the mixing time of the closely related Kac's walk. In \cite{pillai_kacs_2017}, Pillai and Smith used a sequence of couplings for copies of the Kac walk and with this technique managed to prove an \(O(N\log N)\) mixing time for the Kac walk, which after rescaling the time to the same order as in the Kac model corresponds to times of order \(\log(N)\).
 
The main aim of this work is to show that if only a subset of finitely many particles is considered, the relaxation of the joint cumulants of their kinetic energies is much faster, not worse than proportional to \(\log(N)\). Furthermore, if we consider the joint cumulant involving energies of at least two particles, this will quickly become \(o_N(1)\) as \(N\to \infty\).
Therefore, this result also implies that these random variables become ``almost independent'' at logarithmic time scales.

	\subsection*{Acknowledgements}
	
	We are most grateful to Cl\'ement Mouhot for pointing to us this simplified model to try out and develop the cumulant hierarchy techniques.
	We also thank  Gerardo Barrera Vargas, Cristian Giardin\'a,
	Stefan Gro\ss{}kinsky,
	Kalle Koskinen, Jorge Kurchan,
	Herbert Spohn, Gigliola Staffilani, and
Minh Binh Tran, for helpful discussions and additional references.

  The work has been supported by the Academy of Finland, via an Academy project (project No. 339228) and the {\em Finnish Centre of Excellence in Randomness and Structures\/} (project No. 346306).
		JL would also like to thank the Isaac Newton Institute for Mathematical
Sciences, Cambridge, UK, for support and hospitality during the
programme \defem{Frontiers in kinetic theory: connecting microscopic to
macroscopic} (KineCon 2022) where partial work on this paper was
undertaken. This work was also supported by a grant from the Simons
Foundation.	
	
	\section{Cumulants in the Kac model}
	\label{sect:Cumulants-in-the-Kac}
	
For the sake of completeness, we
have summarized the definition of joint cumulants and the Wick polynomial regularizations in
	the Appendix \ref{sec:cumulantsandWick}.

Since cumulants are related to expectations of certain monomials of the random variables, instead of the master equation, we start here from the evolution equation of generic observables.   The stochastic Kac model corresponds to a Feller process whose generator is given by the adjoint of the operator in the Kac master equation.  More precisely, assume $F_0^N$ is a Radon probability measure on $\Nsphere$
which we use to determine the initial data for the stochastic Kac process above.  Let $F_t^N$ denote the resulting distribution for $v(t)$.  Then for any bounded continuous observable $\phi\in C_b(\Nsphere)$ we have
\[
\dv{t}  \mean{\phi}_{F_t^N} = \mean{N(Q_N - I)^* \phi}_{F_t^N}
 = \mean{N(Q_N - I) \phi}_{F_t^N}\,,
\]
using the self-adjointness of the operator $Q_N$.
In other words, for any \(\phi \in C_{b}(\Nsphere)\), we have
	\begin{align}
	\dv{t}\int_{\Nsphere} \phi(v) F_t^N(\rmd v) = \int_{\Nsphere} \left(N(Q_N-I)\phi\right)(v)F_t^N(\rmd v)\,.
	\label{eq:backward-kolmogorov}
	\end{align}
The connection between the Feller generator and the adjoint of the operator in the master equation on $L^2(\Nsphere,\nu_{N})$ is also apparent in this formulation,
although it should be stressed that this formulation makes sense even if the measure \(F_t^N\) does not admit a density with respect to the uniform measure on \(\Nsphere\).

One simplifying feature of the above formulation of Kac model is that it has exponential moments and a moment generating function which is entire.
We will consider both the velocity random variables $v_i$ and the corresponding kinetic energies $e_i:= (v_i)^2$.  Since for any $v\in \Nsphere$
\[
 \sum_{i=1}^N e_i = |v|^2 = N\,,
\]
it follows that $0\le e_i(t)\le N$ and $|v_i(t)|\le \sqrt{N}$ almost surely for any $t\ge 0$,
and thus the functions
\[
 \phi_1(v;\xi) := \rme^{\xi\cdot v}\,,
 \qquad \phi_2(v;\xi) := \rme^{\sum_{i=1}^N\xi_i (v_i)^2}\,,
\]
are bounded and continuous on $\Nsphere$
for any $\xi \in \C^N$.  In addition, we can conclude using Morera's theorem that their expectations
$g_j(\xi,t) :=\mean{\phi_j(v;\xi)}_{F_t^N}$, $j=1,2$,
are entire functions in $\xi \in \C^N$ with $g_j(0,t)=1$.
In addition, since for all $z\in \C$ we have
$|\rme^z-1|\le |z| \rme^{|z|}$, it follows that
\[
 |\phi_j(v(t);\xi)-1|\le \frac{1}{2}\,,
\]
almost surely
for all $\xi \in B(0,\vep)\subset \C^N$, all $t\ge 0$, and both $j=1,2$, at least if $\vep = \frac{1}{4 N}$.
For such $\xi$, we also have
\begin{align}\label{eq:momgenbound}
 |g_j(\xi,t)-1|=|\mean{\phi_j(v(t);\xi)-1}_{F_t^N}|
 \le \frac{1}{2}\,,
\end{align}
and thus also $|g_j(\xi,t)|\ge 1-|1-g_j(\xi,t)|\ge \frac{1}{2}$.

The above results imply that the generating functions of cumulants, $h_j(\xi,t):=\ln g_j(\xi,t)$,
are analytic in the neighborhood $\xi \in B(0,\vep)$ for any $t$.  We also find that for such $\xi$
\[
 \partial_t h_j(\xi,t) = \frac{\partial_t g_j(\xi,t)}{g_j(\xi,t)}
 = \rme^{-h_j(\xi,t)} \mean{N(Q_N - I)\phi_j(\cdot;\xi)}_{F_t^N}
\]
is a continuous function in $t$.  From now on,
given any integrable observable $\Phi:\R^N\to \C$,
we use the shorthand notation $\mean{\Phi(v(t))}$ or $\mean{\Phi(\cdot)}_t$ for the expectation
$\mean{\Phi(\cdot)}_{F_t^N}$. When it is clear from the context, we even use \(\mean{\cdot}\).
Integrating the time-derivative thus results in the formula
\[
 h_j(\xi,t)
 = h_j(\xi,0)
+ \int_0^t\rmd s\, \rme^{-h_j(\xi,s)} \mean{N(Q_N - I)\phi_j(\cdot;\xi)}_{s}\,,
\]
for all $\xi\in B(0,\vep)$.

By analyticity, we can use the Cauchy's integral formula to express all partial derivatives in $\xi$ as iterated contour integrals.  Then, the above bounds allow application of Fubini's theorem to exchange their order with the time-integral and the expectation in the above formula.
This yields the following identity satisfied by cumulants in our setup: Given any sequence of indices $I_j\in [N]$, $j=1,2,\ldots,n$, we have
for velocity cumulants
$\kappa[v(t)_I]  = \partial_{\xi}^I h_1(0,t)=
\prod_{j=1}^n \partial_{\xi_{I_j}}(h_1(\xi,t))|_{\xi=0}$ and thus
\[
 \kappa[v(t)_I] = \kappa[v(0)_I]
 + \int_0^t\rmd s\, \mean{\partial_{\xi}^I\left( \rme^{-h_j(\xi,s)}N(Q_N - I)\phi_1(\cdot;\xi)\right)_{\xi=0}}_{s}\,.
\]

To obtain a simplified hierarchy for the cumulants, we next recall the generating functions of Wick polynomials of the above collections of random variables,
$G_j(Y;\xi,t) = \rme^{\xi \cdot Y-h_j(\xi,t)}$. Then,
for any $v$,
$\rme^{-h_1(\xi,t)}= G_1(v;\xi,t) \phi_1(v;-\xi)$
and with $e_i=v_i^2$, $i=1,2,\ldots,N$, also
$\rme^{-h_2(\xi,t)}= G_2(e;\xi,t) \phi_2(v;-\xi)$.
We will use the notations $\w{v_I}_t$ and $\w{e_I}_t$ for the Wick polynomials $\partial_{\xi}^I G_1(v;\xi,t)|_{\xi=0}$
and $\partial_{\xi}^I G_2(e;\xi,t)|_{\xi=0}$, respectively.  Summarizing, for the velocity cumulants we have
\begin{align}\label{eq:integratedvhierarchy}
 \kappa[v(t)_I] = \kappa[v(0)_I]
 + \int_0^t\rmd s\, \mean{\partial_{\xi}^I\left(
 G_1(v(s);\xi,s) \phi_1(v(s);-\xi) N(Q_N - I)\phi_1(v(s);\xi)\right)_{\xi=0}}\,,
\end{align}
and for the energy cumulants
\begin{align}\label{eq:integratedehierarchy}
 \kappa[e(t)_I] = \kappa[e(0)_I]
 + \int_0^t\rmd s\,
 \mean{ \partial_{\xi}^I\left(G_2(e(s);\xi,s) \phi_2(v(s);-\xi) N(Q_N - I)\phi_2(v(s);\xi)\right)_{\xi=0}}\,.
\end{align}
We will next show that these results allow derivation of evolution equations for the cumulant hierarchy of velocities and of energies.

	\subsection{Time-evolution of velocity and energy cumulants}

Let us first inspect the cumulants of velocities.
We recall the explicit form of the operator $Q_N$
and apply this in (\ref{eq:integratedvhierarchy})
which results in the identity
\[
 \phi_1(v;-\xi) N(Q_N - I)\phi_1(v;\xi)
  =
  \frac{1}{N-1}\sum_{i,j=1}^N \cf{i\neq j}\int_{-\pi}^{\pi} \frac{\rmd \theta}{2\pi}
   \left(\rme^{(R_{i,j}(\theta)-I)v\cdot \xi}- 1\right)\,, \quad v\in \R^N\,.
\]
To evaluate its derivatives,
we use the Leibniz rule which in the subsequence notation can be written as
\[
 \partial_\xi^I(A(\xi) B(\xi) )
 = \sum_{J \subseteq I}
  \partial_\xi^{I\setminus J} A  \, \partial_\xi^J B\,.
\]
Differentiating the result in time yields
	\begin{align}
	\dv{t}\kappa[v(t)_I]
	&= \frac{1}{N-1}\sum_{i,j=1}^N \cf{i\neq j}\int_{-\pi}^{\pi} \frac{\rmd \theta}{2\pi} \sum_{\emptyset \neq J \subseteq I}\mean{\wick{v_{I\setminus J}}_t \partial_{\xi}^J \left(\rme^{(R_{i,j}(\theta)-I)v\cdot \xi}- 1\right)\vert_{\xi=0}}_t \,.
	\end{align}
Here, the term corresponding to an empty sequence $J$ evaluates to zero when $\xi=0$ and one could evaluate the remaining derivative to obtain a polynomial of $v$ which then expands using (\ref{eq:truncated-moments-to-cumulants}) into a nonlinear term in cumulants of $v$.

Instead of velocity cumulants we focus here on the evolution of the energy cumulants.
We follow the same strategy as for the velocities, where in the first step differences in energies appear instead of differences of velocities.
Explicitly, we note that $\phi_2(v;-\xi) \phi_2(v';\xi)$, where $v'$ is the velocity vector after a ``collision'', involves an exponential of a sum over the following terms:
	 \begin{align}
	 (R_{i,j}(\theta)v)_i^2 - v_i^2 &= P_\theta(v_i,v_j) \nonumber \\
	 (R_{i,j}(\theta)v)_j^2 - v_j^2 &= Q_\theta(v_j,v_i) \nonumber \\
	 (R_{i,j}(\theta)v)_k^2 - v_k^2  &= 0, \quad k \neq i,j	\,.
	 \end{align}
	 Here the \emph{energy collision polynomials} \(P_\theta\) and \(Q_\theta\) are given by
		\begin{align}
 			P_\theta(v_i,v_j)&=-\sin(\theta)^2v_i^2 
 			+ 2\cos(\theta)\sin(\theta)v_iv_j + \sin(\theta)^2 v_j^2 \nonumber \\
 			Q_\theta(v_j,v_i)&= -\sin(\theta)^2v_j^2 - 2\cos(\theta)\sin(\theta)v_iv_j + \sin(\theta)^2 v_i^2 \,.
 		\label{eq:energy-collision-polynomials-def}
		\end{align}
		They satisfy the following antisymmetry relation
		\begin{align}
		Q_\theta(v_j,v_i) = (-1)P_\theta(v_i,v_j).
		\end{align}
	The cumulant generating function of the energy variables thus evolves according to
	\begin{align}
	\dv{t} \kappa_t(e_I)
		&=\frac{1}{N-1}\sum_{i,j=1}^N \cf{i\neq j} 
		\sum_{\emptyset \ne J \subseteq I} \int_{-\pi}^{\pi} \frac{\rmd \theta}{2\pi}
		\mean{\wick{e_{I\setminus J}}_t \partial_\xi^{J}\left(\rme^{\xi_i P(v_i,v_j) + \xi_j Q(v_j,v_i)}-1 \right)\vert_{\xi=0}}_t.
		\label{eq:time-evolution-of-energy-cumulants}
 	\end{align}
 	To derive the main results, we will expand and simplify the remaining expectation.  In particular, it will turn out to be a term involving only energy cumulants, so that also the energy cumulant hierarchy is ``closed'' and can be used independently from the velocity cumulants.  This will be done in more detail in Sections \ref{sec:Nonrep} and \ref{sec:Repeated}.
	 
	\subsection{Symmetric measures and partition classifiers of their cumulants}

Our main result will be to control the generation of chaos in the above Kac model in the sense of asymptotic near-independence of energies of different particles.  Similar results could then be derived for the particle velocities, but for the sake of brevity we do not go into this argument in detail here.  The control will be given as an explicit bound for the difference between the energy cumulant and its value at equilibrium, i.e., for the uniform distribution on the energy sphere.

Although the bound will be given for fixed system, in particular, for fixed number of particles $N$, we aim at bounds which are valid for all large enough $N$.  To stress this, we will carefully track the dependence of the constants in the upper bounds on $N$ and on the initial data.
Although we do not consider limiting properties of suitably prepared sequences \((F_0^N)_{N\geq 2}\) of initial probability measures, as is common in derivations of kinetic equations, the bounds here could be employed to make also such scaling limit statements.

Given a number of particles $N\ge 2$, the initial data is taken to be essentially arbitrary, apart from one technically simplifying assumption: we assume that $F_0^N$ is a \defem{symmetric} Radon probability measure on $\Nsphere$.
Here, and in the following, we identify such measures with their extensions to $\R^N$, i.e., with the Borel measure $\hat{F}_0^N$ defined by setting $\hat{F}_0^N(B)= F_0^N(B\cap \Nsphere)$ for any Borel set $B\subset \R^N$; note that then the support of $\hat{F}_0^N$ is contained in  $\Nsphere$ and the restriction to the Borel $\sigma$-algebra of $\Nsphere$ yields back $F^N_0$.
Symmetry then refers to invariance of the extension $\mu$ under permutations of labels of the particles, see Definition \ref{def:symmetric-measure}.

Although this assumption hides some features of the dynamics, it is not a substantial restriction.
For example, given an arbitrary initial Radon probability measure but only asking for properties of expectation values of label permutation invariant observables, we can replace the initial measure by its symmetric projection without changing any of the expectation values of interest.

\begin{definition}[Symmetric probability measure on \(\R^N\)]\label{def:symmetric-measure}
	Suppose \(F^N\) is a Radon probability measure on $\R^N$.   The measure $F^N$ is called symmetric, if for any fixed permutation \(\sigma \colon [N] \to [N]\), its pushforward  by the coordinate permutation function \(\iota_\sigma \colon E^N \to E^N\) with \[\iota(x_1,\dots, x_N) = (x_{\sigma(1)},\dots, x_{\sigma(N)})\] agrees with \(F^N\):
	\begin{align}
	F^N = \iota_{\sigma,\ast}[F^N].
	\end{align}
	
	Equivalently, this can by formulated via test functions. The measure \(F^N\) is symmetric, in case for every test function \(\phi \in C_c(\R^N)\), we have
	\begin{align}
	\int_{\R^N} \phi(v)F^N(\rmd v) = \int_{\R^N} (\phi \circ \iota_{\sigma})(v)F^N(\rmd v),
	\end{align}
	for all \(\sigma \in \mathfrak{S}_N\).
	\end{definition}
	The symmetry condition encapsulates the notion that the individual constituents of the system are indistinguishable from one another. In other words, they are mutually \emph{exchangeable}. Symmetric probability measures have rich mathematical structure, as discussed e.g.\ in \cite{hewitt_symmetric_1955}.

		The assumption that the system starts from a symmetric initial state corresponds to requiring that the initial probability measures  \((F_0^N)_{N\geq 2}\), each defined on the corresponding hypersphere \(\Nsphere\), are symmetric in the sense that each measure \(\hat{F}_0^N(A) = F_0^N(A \cap \Nsphere)\) is a symmetric measure according to Definition \ref{def:symmetric-measure}. This symmetry assumption will permeate the cumulant hierarchy and result in additional symmetries in the joint cumulants. In Sec.~\ref{sec:Nonrep}, we shall explain how this becomes visible in the joint cumulants of the energy variables, and how we can quantify the chaoticity of the joint cumulants. This paves way to the precise statement of the main results of the article.

		The symmetry of the probability measure is reflected in the joint cumulants of the energy variables as an invariance with respect to permutations of the labels of the random variables. Indeed, if \(I = \{(1,i_1),\dots, (n, i_n)\}\) is a sequence of particle labels, then for every permutation \(\sigma \in \mathfrak{S}_N\), we have 
	\begin{align}
		\kappa[e(0)_I] = \kappa[e(0)_{\sigma^{-1}(I)}],	
	\end{align}
	with \(\sigma^{-1}(I) = \{(1,\sigma^{-1}(i_1)),\dots (n,\sigma^{-1}(i_n))\}\).\footnote{It is important to keep in mind that this is distinct from the usual permutation symmetry of the arguments of joint cumulants. The label permutation invariance means that, for instance, \(\kappa[e_1(0),e_2(0),e_2(0)] = \kappa[e_3(0), e_2(0), e_2(0)]\), while \(\kappa_[e_1(0),e_2(0),e_2(0)] = \kappa[e_2(0),e_1(0),e_2(0)]\) is something that is true even if \(F_0\) is not a symmetric measure.}
	
	The time-evolution in \eqref{eq:time-evolution-of-energy-cumulants} is invariant under label permutations, so this property of label permutation invariance carries over to the time-evolved cumulants \(\kappa[e(t)_I]\). Therefore, the value of \(\kappa[e(t)_I]\) is determined simply by how many \emph{different} labels there are in the sequence \(I = \{(1,i_1),\dots, (n,i_n)\}\), and how many times each such label is repeated in the sequence -- no matter what the specific label is. This motivates the following definition of \emph{partition classifiers}, which we will use to index the cumulants after the label permutation symmetry has been factored out.

	\begin{definition}[Partition classifier]
	An \(n\)-tuple \(r = (r_1,\dots, r_n) \in \N_{0}^{n}\) is called a \emph{partition classifier} (or \emph{classifier}) \emph{of order \(n\)}, in case the following two conditions are met: (i) The components of the vector sum to \(n\), i.e. \(\sum_{k=1}^{n}r_k =n\) and (ii) the components have been arranged in a decreasing order, so that \(r_1\geq r_2 \geq \dots \geq r_n \geq 0\).
	\label{def:partition-classifier}
	\end{definition}
		
	Given a partition classifier \(r\), we reserve the notation \(\len(r)\) for the number of its nonzero components. The collection of all partition classifiers of order \(n\) is denoted by \(\mathscr{C}_n\),
    and we often drop the trailing zeroes when denoting elements $r\in \mathscr{C}_n$; note that $n$ is uniquely determined by the sum of the (non-zero) elements in $r$.  The classifier corresponding to the label sequence with no repeated labels has a special role in what follows, and we will denote it by \(1_n \coloneqq (1,1,\dots, 1) \in \N_0^n\). Finally, \(\mathscr{C}'_n \coloneqq \mathscr{C}_n \setminus \{1_n\}\) denotes the set of all particle classifiers that correspond to a label sequence with at least one repeated label.
	
	\begin{remark}
	Each partition classifier of order \(n\) corresponds bijectively to a number theoretic partition of the natural number \(n\) -- hence the name. The number of different partition classifiers therefore satisfies \(\abs{\mathscr{C}_n} \sim \frac{1}{4n\sqrt{3}}\exp\left(\pi \sqrt{\frac{2n}{3}}\right)\) as \(n \to \infty\).
	\end{remark}

	Any partition classifier defines a multi-index of order \(n\), with the additional special property that the components are decreasing. We therefore need to note how to use multi-indices to label joint cumulants. Given any multi-index \(r \in \N_0^n\), we may construct one possible compatible label sequence in the following way. First, we let \(\Delta(r) = (0, r_1, r_1+r_2,\dots, \sum_{i=1}^{n}r_i)\) be a sequence of increments. After this, we define a sequence \(I_r\) by setting
	\[(I_r)_{i} = \sum_{j=1}^n j\cf{i \in (\Delta(r)_j, \Delta(r)_{j+1}]}\,, \quad i=1,2,\ldots,n\,.\]
	For example, with $r=(3,1,1,0,0)\in \mathscr{C}'_5$ we thus have $I_r=(1,1,1,2,3)$ which corresponds to the set $\{(1,1),(2,1),(3,1),(4,2),(5,3)\}$ in the notation used in the beginning of this section.
	
	With this in mind, we define cumulants corresponding to partition classifiers or any other multi-indices in the following way.
	\begin{notation}
	Let \(r \in \N_0^n\) be a multi-index. The joint cumulant corresponding to this multi-index is defined as
	\begin{align}
		\kappa[e_r] \coloneqq \kappa[e_{I_r}].
	\end{align}
	\end{notation}
	
	If \(r\) is a partition classifier of order \(n\) and \(s \in \N_0^n\) is a multi-index, we say that \(r \sim s\), in case there exists a permutation \(\sigma \in \mathfrak{S}_n\) such that \(r = \sigma(s)\) 	
	
	In the following computations, we will encounter situations where certain repetitions of a random variable are removed when computing a Wick polynomial or cumulant corresponding to a collection of random variables (determined by a sequence, multi-index or a partition classifier, dependent on the context). This motivates the following notation:
	\begin{notation}[Removing and adding multi-indices]
	 Suppose that \(r \in \N_0^n\) is a multi-index (possibly a partition classifier). If \(r_i \geq \ell\), we denote by
	\begin{align}
	r-(\ell \times i)
	\end{align}
	the multi-index whose components are given by \((r-(\ell \times i))_k = r_k\) for \(k \neq i\) and \((r-(\ell \times i))_i = r_i - \ell \). Note that even if \(r\) is a partition classifier, this is usually not a partition classifier.
	
	Similarly, given any \(r \in \N_0^n\), 
	\begin{align}
	r + (\ell \times i) 
	\end{align}
	is the multi-index obtained by increasing the \(i\)th component of \(r\) by \(\ell\) while keeping other components the same.
	\end{notation}
	
	The following notation will become relevant later, when we try to isolate the main contribution in the time-evolution of cumulants.
	\begin{definition}
	Let \(r \in \mathscr{C}_n\) be a partition classifier. We define \(\texttt{break}_\ell(r)\) to be the set that consist of those \(q \in \mathscr{C}_n\), for which there exist \(\ell_1,\ell_2 \in [n]\), \(\ell_1 \neq \ell_2\), such that \(r_\ell = q_{\ell_1} + q_{\ell_2}\); and in addition for which the sequence \(r'\) obtained from \(r\) by removing the components \(\ell,n\)  and shifting, is identical to the sequence \(q'\) obtained from \(q\) by removing the components \(\ell_1,\ell_2\) and shifting.
	
	The set \(\texttt{break}(r)\) is the union \(\cup_{\ell=1}^n \texttt{break}_\ell(r)\).
	
	\end{definition}

\begin{example}
Let \(r = (4,0,0,0)\). Now \((3,1,0,0), (2,2,0,0) \in \texttt{break}_1(r)\), whereas \((2,1,1,0) \notin \texttt{break}_1(r)\). Also, \((3,1,0,0)\notin \texttt{break}_2(r)\).
\end{example}

	\begin{definition}
	Fix \(n \in \N\). Define a relation \(\succ\) on \(\mathscr{C}_n\) by setting \(s \succ r\), if and only if \(r \in \texttt{break}(s)\). 	
	
	Then, set \(\mathscr{C}_{n,1} = \{1_n\} = \{(1,\dots, 1)\}\), and define iteratively using the above relation for any $i$ with $1<i\le n$,
\begin{align}
\mathscr{C}_{n,i}=\{s \in \mathscr{C}_n \colon \text{ there exists } r \in \mathscr{C}_{n,i-1} \text{ such that } s\succ r\}.
	\end{align}
	\label{def:partition-classifier-layering}
\end{definition}

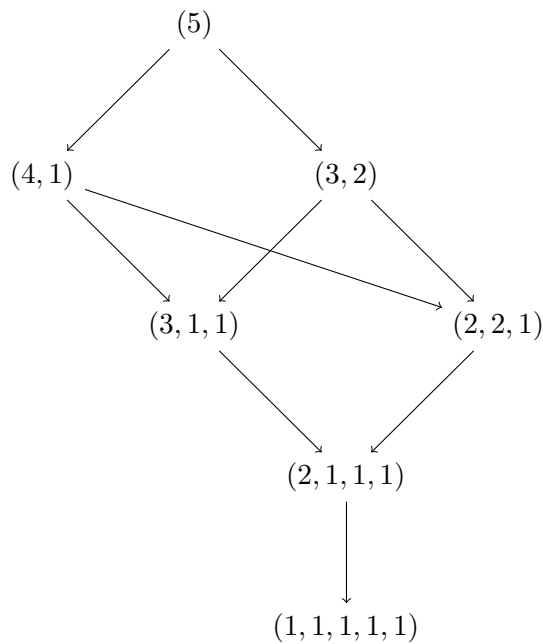
\begin{figure}[ht!]
		\centering
		\begin{tikzpicture}
			
			\node (l1e1) at (0,0) {\((1,1,1,1,1)\)};
			
			\node (l2e1) at (0, 2)  {\((2,1,1,1)\)};
			
			\node (l3e1) at (-2,4) {\((3,1,1)\)};
			
			\node (l3e2) at (2,4) {\((2,2, 1)\)};
			
			\node (l4e1) at (-4,6)  {\((4,1)\)};
			
			\node (l4e2) at (0, 6) {\((3,2)\)};
			
			\node (l5e1) at (-2, 8) {\((5)\)};
			
			\draw[<-] (l1e1) to node [swap, left] {} (l2e1);
			
			\draw[<-] (l2e1) to node [swap, left] {} (l3e1);
			
			\draw[<-] (l2e1) to node [swap, left] {} (l3e2);
			
			\draw[<-] (l3e1) to node [swap, left] {} (l4e1);
			
			\draw[<-] (l3e1) to node [swap, left] {} (l4e2);
			
			\draw[<-] (l3e2) to node [swap, left] {} (l4e2);
			
			\draw[<-] (l4e1) to node [swap, left] {} (l5e1);
			
			\draw[<-] (l4e2) to node [swap, left] {} (l5e1);
			
			\draw[<-] (l3e2) to node [swap, left] {} (l4e1);
		\end{tikzpicture}
		
		\caption{Relations between partition classifiers in \(\mathscr{C}_5\). In the graph, the nodes are in correspondence with elements in \(\mathscr{C}_5\), and there is a directed edge \(s\to r\) in the graph precisely when \(r \in \text{break}(s)\).}
		\label{fig:partition-classifier-layering}
	\end{figure}

\begin{remark}
The above relation can be extended to be a strict partial order on \(\mathscr{C}_n\), but this fact will not be used in what follows. Clearly \(\{\mathscr{C}_{n,i}\}_{i=1}^n\) partitions \(\mathscr{C}_n\), and \(s \in \mathscr{C}_{n,k}\) precisely when it is a partition classifier that has exactly \(n-(k-1)\) non-zero components. Figure \ref{fig:partition-classifier-layering} illustrates these sets and their relation to \(\mathscr{C}_{n}\), when \(n=5\)
\end{remark}

	\subsection{Main Results: controlling chaos and equilibriation via the cumulant hierarchy}
	
 	We will quantify the chaoticity of the energy cumulants by studying their norm in the following normed spaces, which we will call the \(\alpha\)-chaos spaces.
 	
 	\begin{definition}[\(\alpha\)-chaos space]
	Let \(\alpha \in [0,1]\). We define the normed space \((X_{n,N}^{\alpha}, \norm{\cdot}_{\alpha,n,N})\) to consists of vectors \((\kappa_r)_{r\in \mathscr{C}_{n}'} \subset \R^{\mathscr{C}_n'}\), together with the weighted supremum norm
	\begin{align}
		\norm{\kappa}_{\alpha} = \norm{\kappa}_{\alpha,n,N} = \sup_{r \in \mathscr{C}_n'} \abs{\kappa_r}(N-1)^{\alpha(\len(r)-1)}.
	\end{align}
	\label{def:spaces-Xa}
\end{definition}
  We will mainly use these norms with $0<\alpha<1$, but the extremal cases will appear in some of our estimates.  The norms are increasing functions of $\alpha$: For any $r\in \mathscr{C}_{n}'$
  and $0\le \alpha\le \alpha'\le 1$ we have
  $\norm{\kappa}_{\alpha,n,N}\le \norm{\kappa}_{\alpha',n,N}$.

\begin{remark} We have excluded the non-repeated label sequence \(1_n = (1,1,\dots, 1)\) from the definition of these \(\alpha\)-chaos spaces. The reason for this is that the completely non-repeated cumulants will be treated separately, and thus will show up as a source term to the cumulant hierarchy. However, the conservation law implies that if
\begin{align}
\sup_{r\in \mathscr{C}_n'} \abs{\kappa[e_r]}(N-1)^{\alpha(\len(r)-1)} \leq C,
\end{align}
then
\begin{align}
\sup_{r\in \mathscr{C}_n} \abs{\kappa[e_r]}(N-1)^{\alpha(\len(r)-1)} \leq 4C
\end{align}
whenever \(n \leq N/2\) (see the proof of Proposition \ref{prop:non-repeated-stronger-claim-for-finite-order}).

Therefore, assumptions pertaining to the joint cumulants indexed by \(\mathscr{C}_n'\) will imply similar assumptions for all joint cumulants indexed by \(\mathscr{C}_n\) with a slightly larger constant.
\end{remark}

	If the above norm of a vector \((\kappa_r)_{r\in \mathscr{C}_n'}\) is bounded by a constant \(C_n\), it follows that the individual components are bounded by
	\begin{align}
		\abs{\kappa_r} \leq \frac{C_n}{(N-1)^{\alpha(\len(r)-1)}}.
	\end{align}
	Since the elements of \(X_{n,N}^\alpha\) will be given by the \(n\)th order energy cumulants with some repetition, so that \(\kappa_r = \kappa^{n,N}[e_r]\), it follows that such uniform bounds imply
	\begin{align}
	\limsup_{N\to \infty}\abs{\kappa^{n,N}[e_r]} \leq \cf{\len(r)=1}C_n,
	\end{align}
	corresponding to the fact that only cumulants involving one particle can survive to the many-particle limit of any sequence of probability measures whose cumulants have a norm uniformly bounded in $N$.
	
	When studying the time-evolution of the energy cumulants in the Kac model, we will try to find good bounds for their \(X^{\alpha}_{n,N}\) norms. To this end, it should be noted that the operator norm of an operator \(A \colon X_{n,N}^\alpha \to X_{n,N}^{\alpha}\) satisfies
	\begin{align}
	\norm{A}_{\alpha} \leq \sup_{r} \sum_{r' \in \mathscr{C}_n'} \abs{A_{r,r'}}(N-1)^{\alpha(\len(r)-\len(r'))},
	\end{align}
	a fact that we will need in order to control the linear part of the time-evolution.

	We are now ready to state the main theorem of this article. Our results concern the generation of chaos and convergence towards the equilibrium of the energy cumulant hierarchy. The initial data is allowed to have different ``orders of chaoticity'', and this will affect the corresponding results.

	\begin{assumption}[Chaos bounds]
	Let \(c\ge 0\), \(N\ge 2\), and $n^*\ge 1$ be given. Define \(\gamma_n = c(n-1)\) for \(n \geq 1\).  We say that $B>0$ is a constant
	for the
	\emph{chaos bound} of type \(\gamma_n\) up to order $n^*$
	if  for every \(n \in [n^*]\) the joint cumulants are bounded in the \(\alpha\)-norm as
	\begin{align}
	\norm{\kappa^{n,N}_0}_\alpha \leq B^{n-1}(n-1)! N^{\gamma_n}.
	\end{align}
	A sequence of initial data \((F_0^N)_{N \geq 2}\) on \(\Nsphere\) is said to satisfy the chaos bound of type \(\gamma_n\) with a constant $B$ if the same $B$ can be used for every $N$.
	\label{assumption:chaos-bounds}
	\end{assumption}
	
	\begin{remark}
	Of course, for any  fixed $N$, $c$ and $n^*$ such a constant $B$ can be found.  However, $B$ might need be very large, in particular, it could increase as a power law in $N$.  For the results below to be useful, $B$ should be order one even if $N$ is large.  This can always be achieved for chaos bounds with a constant $c=1$  in
	the Assumption \ref{assumption:chaos-bounds} by increasing $N$ if needed, at least when considering joint energy cumulants in a symmetric state up to any finite order. Indeed, Proposition \ref{prop:generic-bounds-for-initial-data} shows that the energy joint cumulants will always satisfy a bound
	\begin{align}
	\norm{\kappa^{n,N}}_\alpha \leq 4^{n-1} (n-1)! N^{n-1}, \quad n \leq N/2 + 1\,.
	\end{align}
	\end{remark}

	\begin{theorem}[Generation of chaotic bounds]
	\label{thm:generation-of-alpha-chaotic-bounds}	
	Let \(c\ge 0\) and \(\alpha \in (0,1)\). Let \(n^* \in \N\) be a maximal order of cumulants. Then there exists a \(N_0 = N_0(n^*, \alpha, c)\ge 2\) such that the following result holds.
	
	Consider some fixed \(N\ge N_0\) and some symmetric initial data \(F_0^N\) on \(\Nsphere\). Denote the corresponding joint cumulants at order \(n \in [n^*]\) and at time \(t \geq 0\) by \(\kappa^n_t[e_r] = \kappa[e_{I_r}(t)]\). Let \(B \geq 1\) be a constant such that the initial values of the cumulants satisfy
	\begin{align}
	\max\limits_{1 \leq n \leq n^*}\max\limits_{r \in \mathscr{C}_n}\left(\abs{\kappa_0^n[e_r]}N^{\alpha (\len(r)-1)-c(n-1)}/(n-1)!\right)^{1/n} \leq B.
	\end{align}
	In other words, assume that the initial data satisfies assumption \ref{assumption:chaos-bounds}. Then there exists a constant \(C\), depending only on \(B\), such that for all
	\(n \leq n^*\), the time-evolved cumulants satisfy
		\begin{align}
		\norm{\kappa_t^{n,N}}_\alpha \leq C^{n^2} n! \left(N^{\gamma_n}\rme^{-\frac{1}{4}t} + 1\right).
		\label{ineq:generation-of-alpha-chaos-highly-chaotic-1}
	\end{align}
\end{theorem}
 We should point out that we have not tried to optimize the bounds in $n$, the order of the cumulants.  Most likely their asymptotic increase is much less than $C^{n^2} n!$, perhaps even as good as $C^{n-1} n!$.

\begin{remark}
	Under conditions that guarantee that odd joint cumulants involving velocities of particles vanish, one could use Möbius inversion techniques to recover chaotic bounds for the joint cumulants of the \emph{velocities} of the particles. For example,  Theorem 2 in Bauer et al.\ \cite{bauer_bernoulli_2024} can be used to express \(\kappa^{2n,N}_t(v_I)\) in terms of certain energy cumulants \(\kappa_{t}^{n,N}(e_J)\), and it remains to check what kinds of energy cumulants appear on the right hand side of the identity. Some combinatorial losses seem inevitable, however.
\end{remark}

\begin{theorem}[Convergence to stationarity]
\label{thm:convergence-to-stationary-non-chaotic}
Let \(c\ge 0\) and \(\alpha \in (0,1)\). Let \(\gamma_n = c(n-1)\) and fix \(n^* \in \N\) as the maximal order of cumulants. Then there exists  \(N_0 = N_0(n^*, \alpha, c)\ge 2\) such that the following result holds.
	
	Consider some fixed \(N\ge N_0\) and some symmetric initial data \(F_0^N\) on \(S^{N-1}(\sqrt{N})\). Denote the corresponding joint cumulants at order \(n \in [n^*]\) and at time \(t \geq 0\) by \(\kappa^n_t[e_r] = \kappa[e_{I_r}(t)]\) and the stationary cumulants at order \(n\), those corresponding to the uniform probability distribution on \(S^{N-1}(\sqrt{N})\), by \(\bar{\kappa}^n[e_r] = \bar{\kappa}^n[e_{I_r}]\). Let \(B \geq 1\) be a constant such that the initial values of the cumulants satisfy
	\begin{align}
	\max\limits_{1 \leq n \leq n^*}\max\limits_{r \in \mathscr{C}_n}\left(\abs{\kappa_0^n[e_r]}N^{\alpha (\len(r)-1)-c(n-1)}/(n-1)!\right)^{1/n} \leq B.
	\end{align}
Then there exists a constant \(C\), depending only on \(B\), such that the time-evolved energy cumulants satisfy the following bound for all \(n \leq n^*\) and $t\ge 0$
\begin{align}
\norm{\kappa_t^{n,N}-\bar{\kappa}^{n,N}}_\alpha \leq n! C^{n^2}N^{\gamma_n}\rme^{-\frac{t}{4}}\,.
\end{align}
\end{theorem}

	We have excluded the completely non-repeated cumulants from the \(\alpha\)-chaos spaces, since we will treat them as source terms. Fortunately for us, their time-evolution solves a closed hierarchy, so we can first describe their time-evolution and then take it as known when proving results pertaining to the time-evolution of the other cumulants. The following stronger results holds for the non-repeated energy cumulants.

	\newcommand{\kappanr}[2]{\kappa_{#1}^{#2,\text{nr}}}

	\begin{proposition}[Generation of chaos for non-repeated cumulants]
	Let \(\alpha \in (0,1)\) and choose some
	\(c\ge 0\) and a maximal order of cumulants \(n^*\in \N\).  Then there is \(N_0 = N_0(n^*,\alpha,c)\ge 2\) such that the following result holds.

	Consider some fixed $N\ge N_0$ and some symmetric initial data.  Denote the corresponding non-repeated cumulants at order $n\in[N]$ and time $t\ge 0$ by
	$\kappanr{t}{n}:=\kappa[e_1(t),e_2(t),\ldots,e_n(t)]$.  Let $B\ge 1$ be a constant such that the initial values of the non-repeated cumulants satisfy
	\begin{align}\label{eq:defBnrconst}
	 \max_{1\le n \le n^*}
	 \left(|\kappanr{0}{n}|N^{(\alpha-c)(n-1)}/(n-1)!\right)^{\frac{1}{n}}\le B\,.
	\end{align}

    Then there is a constant $C\ge 1$, which depends only on $B$, such that
    the time-evolved non-repeated cumulants satisfy
	\begin{align}\label{eq:nrgengoal}
	N^{\alpha(n-1)}\abs{\kappanr{t}{n}} \leq C^n (n-1)! (\rme^{-\frac{n}{4}t}N^{c(n-1)} + 1)\,,
	\end{align}
	for all \(n \in [1,n^*]\) and $t\ge 0$.
	\label{prop:non-repeated-stronger-claim-for-finite-order}
	\end{proposition}

 Although the result appears to concern propagation of chaos, it also yields a form of generation of chaos if $\alpha<c$; hence, the name given to the Proposition.
 The proof of the proposition shows that it holds for example for
 the choices $C=2B$, $N_0=1+(8n_*^2)^{\frac{1}{1-\alpha}}$.
 There is also a corresponding statement concerning the convergence to the stationary state for the non-repeated energy cumulants. This is the content of Proposition \ref{proposition:convergence-to-equilibrium-non-repeated-cumulants}, which will be formulated in the section \ref{sec:nonrepequil}.
 
 \medskip

   Having fixed the number of particles \(N\) and a symmetric initial data \(F_0^N\) we obtain the measures
	\(F_t^N\), $t\ge 0$, from the Kac process.  Suppose we first wait a time $t_0\ge 0$ and then use the first marginal of $F_{t_0}^N$ to define an initial measure $\mu_0$ for the Boltzmann--Kac equation for which $T=t-t_0\ge 0$ serves as the time parameter; we will discuss the well-posedness of this problem in Sec.\ \ref{sec:KacBEaccuracy}.
	We can then use these measures to define the symmetric product measures $\tilde{F}^N_T := \otimes_{i=1}^N \mu_T$ on $\R^N$, and ask the question: How close are the energy cumulants of $\tilde{F}^N_T$ to those of $F^N_{t_0+T}$?
 The following result shows that then up to any finite order and for all large enough systems,
 the two cumulants remain very close to each other, if the Boltzmann equation is started with initial data from a chaotic state, i.e.,
 with suitable large enough $t_0$.  As shown in the theorem, any $t_0$ for which
 $N^{c(n^*-1)}\rme^{-\frac{1}{4}t_0}\le 1$, is sufficient for cumulants of order $n\le n^*$.
    If $c=0$, we can use $t_0=0$; hence the name ``chaotic initial data'' for these.  If $c>0$, we can set $t_0= 4 c (n^*-1) \ln N$ which is $O(\ln N)$ as $N\to \infty$.

Without going into details, let us also point out that these estimates imply concrete upper bounds also
for expectation values of much larger class of observables $g$.  For example, if $g$ is any function which can be approximated by polynomials of energies of a fixed number of particles,
we would still have $\E_{F^N_{t_0 +T}}[g]\approx \E_{\tilde{F}^N_{T}}[g]$ if $N$ is large enough, with an explicit bound for the error obtained by applying moments-to-cumulants formula to the expectation of the polynomial approximation.

 \begin{theorem}[Accuracy of the ``Boltzmann--Kac'' hierarchy]
		\label{thm:kinetic-accuracy}		
Let \(c\ge 0\) and \(\alpha \in (0,1)\). Let \(n^* \in \N\) be a maximal order of cumulants.
Consider \(N\ge N_0\) and some symmetric initial data \(F_0^N\) on \(\Nsphere\), and  suppose \(N_0\ge 2\) and
\(B \geq 1\) are constants
for which Theorem
\ref{thm:generation-of-alpha-chaotic-bounds} holds.   Denote the corresponding joint cumulants at order \(n \in [n^*]\) and at time \(t \geq 0\) by \(\kappa^n_t[e_r] = \kappa[e_{I_r}(t)]\).

Pick some $t_0\ge 0$ for which $N^{c(n^*-1)}\rme^{-\frac{1}{4}t_0}\le 1$. Let $\mu_0$ be the first first marginal of $F_{t_0}^N$,
and let $\mu_T$ denote the corresponding weak solution to the Boltzmann--Kac equation. Denote the cumulants of
the product measure $\otimes_{i=1}^N \mu_T$ on $\R^N$ by $\tilde{\kappa}^{n,N}_T[e_r]$.
Then, for all \(T\ge 0\), we have
		\begin{align}
		 |\kappa^{n,N}_{t_0+T}[e_r] - \tilde{\kappa}^{n,N}_T[e_r]| \le 2 (N-1)^{-\alpha} C^{n^2} n!
 = O(N^{-\alpha}).
 		\end{align}
	\end{theorem}

\section{Properties of initial data}
	\label{sec:initdata}
	
	In this section, we discuss the properties of the possible initial data for the Kac model: Chaotically bounded initial data with all moments, chaotically bounded initial data with diverging moments, non-chaotic initial data, and deterministic initial data.
	
	\subsection{Kac's chaotic initial data}
	\label{sec:intialKac}

	Suppose $f:\R\to [0,\infty)$ is a
	continuous function. Consider the standard entropy function $H:[0,\infty) \to \R$ defined by
	$H(0)=0$ and $H(r) := r \ln r$, for $r>0$.
  Since $H''(r)=\frac{1}{r}>0$, $H$ is a convex function and we may apply Jensen's inequality
  \[
   \int_\R\mu(\rmd y) H(f(y)) \ge
   H\!\left(\int_\R\mu(\rmd y) f(y)\right)
  \]
  to any Borel probability measure $\mu$ on $\R$.
  In particular, if $\int_\R\mu(\rmd y) f(y)=1$, this implies \[\int_\R\mu(\rmd y) f(y) \ln f(y) \ge 0.\]

  Consider then a Kac probability measure generated by the above function $f$ for some fixed $N\ge 2$.  More precisely, define
  $g(v) := \prod_{j=1}^N f(v_j)$ which is a continuous non-negative function on $\R^N$ and hence bounded on the compact set $\Nsphere$. We add an assumption that $f$ and $N$ are such that the restriction of $g$ to the sphere is non-zero; this is true for instance if $f(1)>0$.
  Then $$Z_N := \int_{\Nsphere} \nu_N(\rmd v) g(v) >0,$$ and we can consider initial data given by the probability measure
  $\mu(\rmd v) := Z_N^{-1} g(v) \nu_N(\rmd v) $ on $\Nsphere$.

  Denote $\tilde{f}^N(v) := f(v) Z_N^{-1/N}$
  and $f_0^N(v) := \prod_{j=1}^N \tilde{f}^N(v_j)$
  using which $\mu(\rmd v) =f_0^N(v)  \nu_N(\rmd v) $.  Then, by Jensen's inequality,
  \begin{align*}
   &
 \norm{f_0^N}_2^2 = \int_{\Nsphere} \nu_N(\rmd v) \prod_{j=1}^N (\tilde{f}^N(v_j))^2
   = \lim_{\vep\to 0^+}
   \int_{S^{N-1}} \mu(\rmd v)
   \rme^{\sum_{j=1}^N \ln(\vep + \tilde{f}^N(v_j))}
\\ & \quad
\ge
   \lim_{\vep\to 0^+}
   \exp\left(\sum_{j=1}^N \int_{\Nsphere} \mu(\rmd v) \ln(\vep + \tilde{f}^N(v_j))\right) \\
   & \quad =  \prod_{j=1}^N  \exp\left(\int_{\Nsphere} \mu(\rmd v) \ln(\tilde{f}^N(v_j))\right)=\rme^{N \alpha}\,,
   \end{align*}
where we have used the symmetry of $\mu$ to conclude that the constant $\alpha$ does not depend on $j$,
\[
\alpha:= \int_{\Nsphere} \mu(\rmd v) \ln(\tilde{f}^N(v_1))
= \int_{\Nsphere} \nu_N(\rmd v) \tilde{f}^N(v_1)\ln(\tilde{f}^N(v_1)) \prod_{j=2}^N \tilde{f}^N(v_j)\,.
\]

We can then normalize the remaining measure into a probability measure by using \[c'_N:=\int_{\Nsphere} \nu_N(\rmd v)  \prod_{j=2}^N \tilde{f}^N(v_j)>0.\]  Then $\alpha = c'_N \mean{H(\tilde{f}^N(v_1))}$.  Since $\mean{\tilde{f}^N(v_1)}= 1/c'_N$ we can also consider the function $h := c'_N\tilde{f}^N$ and conclude that $\alpha = \mean{H(h(v_1))} +\ln (1/c'_N)$ with $\mean{h}=1$, which implies that
$\alpha\ge \ln (1/c'_N)$. For large $N$, we would expect $c'_N\approx 1$ and thus the constant term should converge to zero. A more careful check using the known saturation conditions for the Jensen's inequality shows that indeed $\mean{H(h(v_1))}=0$ only if $h$, and hence $f$, is a constant function.
Thus for non-constant $f$, we would expect to find $C>1$, independent of $N$, such that $\norm{f_0^N}_2 \ge C^N$ for all large enough $N$.

Despite the large $L^2$-norm, fixed order $k$ marginals of this sequence of measures converge to
a probability measure on $\R^k$ whose density is,  up to a normalization, equal to $ \prod_{j=1}^N  f(v_j)$.  Thus all non-trivial joint cumulants
of $v$, and of $e$, converge to zero as $N\to \infty$.
In particular, the finite order cumulants are uniformly bounded for all large enough $N$.
Hence, they are a lot smaller than in the generic case considered in the next section, and thus already quite ``chaotic'', although we have not tried to estimate the $\alpha$-norm of these cumulants to quantify this yet.

 	\subsection{Cumulants of generic symmetric measures on \(\Nsphere\)}

		\begin{proposition}
		\label{prop:generic-bounds-for-initial-data}
		Let \(F^N\) be a symmetric measure on \(\Nsphere\) with $N\ge 2$.
		The joint energy cumulants of order \(n \leq \frac{N}{2}+1\) satisfy a bound
		\begin{align}\label{eq:gensymmkappabound}
		|\kappa[e_J]| \leq 4^{n-1} (n-1)! N^{n-\len(J)}\,,\qquad|J|=n\,.
		\end{align}
		In particular, for any \(n \leq \frac{N}{2}+1\), $J\in \mathscr{C}'_{n}$ and
		$\alpha \in [0,1]$,
		\begin{align}
		\norm{\kappa[e_J]}_\alpha \leq 4^{n-1} (n-1)! N^{n-1}\,.
		\end{align}
			\end{proposition}

		\begin{proof}
		Denote by \(\E[\phi]\) the expectation of \(\phi\) with respect to the measure \(F^N\).
				
		If $n=1$, the first order cumulant equals expectation value which by symmetry is equal to
		$\E[e_1]= \frac{1}{N}\sum_{i=1}^N \E[e_i]=1$, since $\sum_{i=1}^Ne_i=N$ almost surely.  Hence, the stated bound holds for $n=1$.

	We next establish an estimate
	\begin{align}
	0\le \E[e_1e_2\cdots e_n] \leq 2^{n-1} \,, \quad 1 \leq n \leq \frac{N}{2}+1\,.
	\end{align}
	The estimate clearly holds for $n=1$, and we claim that for other $n$ it satisfies
	\begin{align}\label{eq:nonrepeest}
	0\le \E[e_1e_2\cdots e_n] \leq 2 \E[e_1e_2\cdots e_{n-1}]\,.
	\end{align}
	Combining these results proves the claim by induction.

	To prove (\ref{eq:nonrepeest}), we first note that
	\begin{align}
	& N \E[e_1\cdots e_{n-1}] = \sum_{i=1}^N \E[e_1e_2\cdots e_{n-1}e_i] = (N-(n-1))\E[e_1 \cdots e_{n}] + (n-1)\E[e_1^2\cdots e_{n-1}]
	\nonumber \\ & \quad
	\ge (N-(n-1))\E[e_1 \cdots e_{n}]
	\end{align}
	since the measure is symmetric and we have
	$e_i\ge 0$ and
	$\sum_{i=1}^Ne_i=N$ almost surely.
	Therefore,
	\begin{align}
	\E[e_1\cdots e_n] \le 2 \E[e_1\cdots e_{n-1}]\,,
	\end{align}
	using the assumption \(n \leq \frac{N}{2}+1\).
	Hence, (\ref{eq:nonrepeest}) holds.

	Let us then assume $2\le n\le \frac{N}{2}+1$.
	By symmetry, we only need to prove the bound
	(\ref{eq:gensymmkappabound}) for $J=J_r$ with
	$r\in \mathscr{C}_n$.  Denote $k:=\len(J)$ for which $1\le k\le n$, $\sum_{\ell=1}^k r_\ell=n$, and $r_1\ge r_2\ge \cdots \ge r_k\ge 1$. Then, for each $\ell\in [k]$ we may use the trivial energy bound $e_\ell \le N$, which holds almost surely, and the estimate in (\ref{eq:nonrepeest}) to conclude that
	\begin{align}
	0\le \E[e^J] \leq N^{\sum_{\ell=1}^{k}(r_\ell-1)}\E[e_1\cdots e_k] \leq N^{n-k}2^{k-1}\,.
	\end{align}

	We next use the general iteration relation satisfied by cumulants\footnote{The formula follows straightforwardly by differentiating the generating function or from the full moments to cumulants formula, see Appendix \ref{sec:cumulantsandWick}.}:
	\[\E[e^J] = \sum_{(1,J_1) \in I \subseteq J}\E[e^{I\setminus J}] \kappa[e_J].\]  Combined with the already established energy moment bounds this yields an estimate
	\begin{align}
	\abs{\kappa[e_J]} & \leq \abs{\E[e^J]} + \sum_{(1,J_1) \in I \subsetneq J}\abs{\E[e^{J\setminus I}]}\abs{\kappa[e_I]}
	\nonumber \\
	&\leq N^{n-k}2^{k-1}
	+ \sum_{(1,J_1) \in I \subsetneq J}\abs{\kappa[e_I]}
	N^{n'-k'}2^{k'-1}|_{n'=|J\setminus I|,\, k'=\len(J\setminus I)}
\,.
	\end{align}
  In this sum, $1\le |I|,n'\le n-1$ and $1\le k' \le k,n'$.  In addition, $n'+|I|=n$ and $k'+\len(I)\ge k$.

  Let us then denote $A_{m,k} := \max_{I:|I|= m, \len(I)=k} |\kappa[e_I]| N^{1-m}$ for $1\le m\le \frac{N}{2}+1$ and $1\le k\le m$.  This yields a sequence which, by the above bound, satisfies
\begin{align}
	A_{n,k} & \leq (2/N)^{k-1}
	+ \sum_{(1,J_1) \in I \subsetneq J}
	A_{n-n',len(I)}
	N^{-k'} 2^{k'-1}|_{n'=|J\setminus I|,\, k'=\len(J\setminus I)}
\,.
	\end{align}
Here, $n'=|J\setminus I|$ satisfies $1\le n'\le n-1$ and for each $n'$ there are at most $\binom{n-1}{n'}$ terms with $|J\setminus I|=n'$ in the sum.
Therefore, $B_n := \max_{1\le m\le n,\,1\le k\le m} \left(A_{m,k} \frac{(N/2)^{k-1} }{ (m-1)! }\right)$, satisfy an inequality
\begin{align}
	B_{n} & \leq \frac{1}{(n-1)!}
	+ B_{n-1} \frac{1}{2}\sum_{n'=1}^{n-1} \binom{n-1}{n'} \frac{(n-n'-1)!}{(n-1)!}
\,,
	\end{align}
for all $2\le n \le  \frac{N}{2}+1$.
Since $B_{1}=A_{1,1} =1$ we obtain
$B_2 \le 1 + \frac{1}{2} < 2$.
For $3\le n \le  \frac{N}{2}+1$, it follows that
\begin{align}
	B_{n} & \leq \frac{1}{2} \left( 1
	+ B_{n-1} (\rme -1 )\right)
\,.
	\end{align}
Therefore, by induction, we can conclude that
\[
 B_n \le 2^{n-1} \,, \qquad 1\le n \le   \frac{N}{2}+1\,.
\]
Hence, for all such $n$ and whenever $|J|=n$ and $\len(J)=k$, we have
\[
 |\kappa[e_J]| \le N^{n-1} A_{n,k} \le
 N^{n-k} 2^{k-1} (n-1)! B_{n}\le
 N^{n-k} 2^{n+k-2} (n-1)! \,.
\]
Since $k\le n$, (\ref{eq:gensymmkappabound}) also holds.

The second statement is an obvious corollary for
$\alpha=1$ and thus it holds for all $\alpha\in[0,1]$ since the norm is increasing in $\alpha$.
		\end{proof}
	
	\subsection{Extreme case: Symmetrized deterministic initial data}
	
	In this subsection, we construct an example of initial data on \(\Nsphere\) which is symmetric but highly non-chaotic, in the sense that it saturates the power of $N$ in the generic bound in  (\ref{eq:gensymmkappabound}).

 Pick \(\bar{v} \in \Nsphere\). Define a measure
	\begin{align}
		\mu \coloneqq \frac{1}{N!}\sum_{\sigma \in \mathfrak{S}_N} f_{\sigma,\ast}[\delta_{\bar{v}}],
	\end{align}
	on \(\Nsphere\) with Borel \(\sigma\)-algebra. The measure \(f_{\sigma, \ast}[\delta_{\bar{v}}]\) is obtained from the Dirac measure at \(\bar{v}\) by pushing it forward by a component permutation map \(f_{\sigma} \colon \Nsphere \to \Nsphere\):
	\begin{align}
	f_{\sigma,\ast}[\delta_{\bar{v}}](B) = \delta_{\bar{v}}(f_\sigma^{-1}(B))
	\end{align}
	with
	\begin{align}
		f_{\sigma}^{-1}(v_1,\dots, v_N) = (v_{\sigma^{-1}(i)})_{i=1}^N.
	\end{align}
	Then for any continuous function
	$\phi\in C(\Nsphere) = C_c(\Nsphere)$, we have
	\begin{align}\label{eq:symmdirac}
	 \int_{\Nsphere} \mu (\rmd  v) \phi(v)
	 = \frac{1}{N!}\sum_{\sigma \in \mathfrak{S}_N}
	 \phi(\bar{v}_{\sigma(i)})_{i=1}^N\,,
	\end{align}
    which also could serve as a definition of the measure $\mu$, by the Riesz--Markov--Kakutani representation theorem.
    It is clearly a symmetric measure.

   The expectation of energy with respect to this measure are fixed by symmetry to be $1$.
   Since the energy observable is continuous, we can also easily compute all moments from
   (\ref{eq:symmdirac}).
   For example, the non-repeating
   second moments are given by
	\begin{align}
	 \E_\mu(e_1e_2)
	&= \int_{\Nsphere} e_1(v)e_2(v) \mu(\rmd v)
	= \frac{1}{N!}\sum_{\sigma \in \mathfrak{S}_N} \bar{v}_{\sigma(1)}^2\bar{v}_{\sigma(2)}^2
	= \frac{1}{N(N-1)}\sum_{i,j=1; j\ne i}^N \bar{v}_{i}^2\bar{v}_{j}^2\,,
	\end{align}
   and the second cumulant, which is equal to covariance, is
   \[
    \kappa_\mu[e_1,e_2] =
    \frac{1}{N(N-1)}\sum_{i,j=1; j\ne i}^N \bar{v}_{i}^2\bar{v}_{j}^2- 1
    = \frac{1}{N}\sum_{i=1}^N \bar{v}_{i}^2
    \frac{1-\bar{v}_{i}^2}{N-1}
    \,.
   \]
   Similarly, we can compute the variance to be
   \[
    \kappa_\mu[e_1,e_1] = \frac{1}{N!}\sum_{\sigma \in \mathfrak{S}_N} \bar{v}_{\sigma(1)}^4 -1
    = \frac{1}{N} \sum_{i=1}^N \bar{v}_{i}^2
    \left(\bar{v}_{i}^2 -1\right)\,.
   \]

Let us consider a simple example of extreme initial data starting from a case where all the energy lies in the first particle: set \(\bar{v} = (\sqrt{N},0,0,\dots, 0)\), and consider the corresponding symmetrized measure $\mu$.
By the above formula, its covariance is given by $\kappa_\mu[e_1,e_2] = -1$ and variance by $\kappa_\mu[e_1,e_1]=N-1$. Thus the upper bound in  (\ref{eq:gensymmkappabound}) is saturated for large $N$ and $n=2$, apart from the prefactor $4^{n-1}$.

The moment generating function for energy is explicitly given by
\[
 \E_\mu[\rme^{e\cdot \xi}] = \frac{1}{N} \sum_{i=1}^N \rme^{N \xi_i}
 = 1+ \frac{1}{N} \sum_{i=1}^N \left( \rme^{N \xi_i}-1\right) = 1+ X(\xi)\,,
\]
and the cumulant generating function is thus
\[
 g_c(\xi) = \ln \E_\mu[\rme^{e\cdot \xi}]
 = \sum_{m=1}^\infty \frac{(-1)^{m-1}}{m}
 X(\xi)^m\,.
\]
Here $X(0)=0$, and we obtain a formula for cumulants with $J=J_r$, $r\in \mathscr{C}_n$,
\[
 \kappa[e_J] = \sum_{m=1}^n
 \frac{(-1)^{m-1}}{m} \sum_{I}
 \prod_{\ell=1}^m \partial^{I_\ell}_\xi X(0)\,.
\]
The sum $\sum_I$ here denotes a sum over all partitions of the sequence $J$ into a sequence
containing $m$ subsequences, $I=(I_\ell)_{\ell=1,\ldots,m}$.  If $\len(I_\ell) > 1$, then $ \partial^{I_\ell}_\xi X(0)=0$, and if $\len(I_\ell) = 1$ and $p_\ell=|I_\ell|$, then
$ \partial^{I_\ell}_\xi X(0)=N^{p_\ell-1}$.
Denoting $k=\len(J)$, we may thus conclude that the non-zero terms in this sum must have $m\ge k$, and then   \[
 \prod_{\ell=1}^m \partial^{I_\ell}_\xi X(0)
 = \prod_{\ell=1}^m N^{p_\ell-1} = N^{n-m}\,.
\]
Hence, the dominant power of $N$ occurs at $m=k$, in which case there are $k!$ choices for $I$ which yield a non-zero contribution.  We find that
\[
  \kappa[e_J] = (-1)^{k-1} (k-1)! N^{n-k}+ O(N^{n-k-1})\,,
\]
which shows that the power law of $N$ in the bound in  (\ref{eq:gensymmkappabound}) is saturated for large $N$ also for any finite $n>2$.

\section{Non-repeated cumulants}\label{sec:Nonrep}
	In this and the following section, the notation \(\mean{\cdot}\) stands for expectation taken with respect to \(F_t^N\). Likewise the \(\wick{e_I}\) denotes the Wick polynomial with respect to \(F_t^N\).
		
	\subsection{Generation of chaos for non-repeated energy cumulants}	
		
	\begin{proofof}{Proposition \ref{prop:non-repeated-stronger-claim-for-finite-order}:}
	
    Let us begin with a brief outline of the  proof.  Using the conservation law for the \(e_i\)-variables, the time-evolution of \(\kappanr{t}{n} = \kappa_t[e_1,\dots,e_n]\) can be shown to satisfy a closed hierarchy, where the nonlinear term in the time-evolution of \(\kappanr{t}{n}\) involves cumulants of strictly lower order, and can therefore be treated as a source term.
	Following this idea, we make an ansatz for the upper bound of the norm of the time-evolved cumulant on the basis of the behavior of the first few cumulants of this form (orders \(n=1, 2\)). Using the Duhamel formula on the hierarchy, we can then show that this upper bound propagates.
	
	First, we thus recall that the first cumulant is equal to the first moment which is fixed by the energy conservation law: we always have $\kappa_t[e_{1_1}]= \mean{e_1}=1$. 

	Thus we only need to consider values $n\ge 2$.
	We begin from the formula (\ref{eq:time-evolution-of-energy-cumulants}).  We first aim at simplifying its right hand side, using the fact that the cumulant is non-repeated.
	Since $t$ is fixed, we do not denote dependence on it in the computations below.
	Since the remaining derivative over indices determined by a subsequence  $J$ of $1_n$
	is zero if $J$ contains any other index apart from $i$ and $j$, there are only three choices of $J$ which yield a non-zero term.  Explicitly, we find that
	\begin{align}
	\dv{t} \kappa_t[e_{1_n}] &= \frac{1}{N-1}\sum_{i,j=1}^N \cf{i\neq j} \sum_{k=1}^n \cf{k=i} \int_{-\pi}^{\pi} \frac{\rmd \theta}{2\pi}\mean{\wick{e_{1_n-(1\times k)}}P_\theta(v_i,v_j)} \nonumber \\
	&+ \frac{1}{N-1}\sum_{i,j=1}^N \cf{i\neq j} \sum_{k=1}^n \cf{k=j}\int_{-\pi}^{\pi} \frac{\rmd \theta}{2\pi} \mean{\wick{e_{1_n-(1\times k)}}Q_\theta(v_j,v_i)} \nonumber \\
	&+ \frac{1}{N-1}\sum_{i,j=1}^n \cf{i\neq j} \int_{-\pi}^{\pi} \frac{\rmd \theta}{2\pi} \mean{\wick{e_{1_n-((1\times i) + (1\times j))}}P_\theta(v_i,v_j)Q_\theta(v_j,v_i)}
	\label{eq:time-evolution-non-repeated-cumulants-wick-form}	
	\end{align}
	where the collision polynomials \(P_\theta(v_i,v_j)\) and \(Q_{\theta}(v_j,v_i)\) are
	\begin{align}
	P_\theta(v_i,v_j) &= -\sin(\theta)^2 v_i^2 +2\cos(\theta)\sin(\theta)v_iv_j + \sin(\theta)^2 v_j^2\nonumber \\
	&=-\sin(\theta)^2 e_i + 2\cos(\theta)\sin(\theta) v_iv_j + \sin(\theta)^2 e_j
	\end{align}
	and
	\begin{align}
	Q_{\theta}(v_j,v_i)= - \sin(\theta)^2 e_j - 2\cos(\theta)\sin(\theta) v_iv_j + \sin(\theta)^2 e_i = -P_\theta(v_i,v_j).
	\end{align}

	In the integral over the random angle $\theta$, any odd function of $\theta$ evaluates to zero.  In particular, this is the case for the trigonometric term \(2\cos(\theta)\sin(\theta)\),
	and thus the first term in the sum may be simplified to
	\begin{align}
	&\frac{1}{N-1}\sum_{i,j=1}^N \cf{i\neq j} \sum_{k=1}^n \cf{k=i} \int_{-\pi}^{\pi} \frac{\rmd \theta}{2\pi}\mean{\wick{e_{1_n-(1\times k)}}P_\theta(v_i,v_j)} \nonumber \\
	&= \frac{1}{N-1}\sum_{i=1}^n \sum_{j=1}^N \cf{i\neq j} \int_{-\pi}^{\pi} \frac{\rmd \theta}{2\pi}\mean{\wick{e_{1_n-(1\times i)}}P_\theta(v_i,v_j)} \nonumber \\
	&= \frac{1}{N-1}\sum_{i=1}^n \sum_{j=1}^N \cf{i\neq j} \int_{-\pi}^{\pi} \frac{\rmd \theta}{2\pi} \left((-\sin(\theta)^2)\mean{\wick{e_{1_n-(1\times i)}}e_i} + \sin(\theta)^2\mean{\wick{e_{1_n-(1\times i)}}e_j}\right)\,.
	\label{eq:non-repeated-one-particle-1}
	\end{align}
	In the remaining sum, we
	 distinguish between two different cases.
	If $n<j\le N$, we can use the symmetry to conclude that $\mean{\wick{e_{1_n-(1\times i)}}e_j}
	= \kappa_t[e_{1_n}]$ since we can swap the labels for $i$ and $j$.  Since also
	$\mean{\wick{e_{1_n-(1\times i)}}e_i}=\kappa_t[e_{1_n}]$, these values of $j$ will not contribute to the sum.

	If $1\le j \le n$, $j\ne i$, we will get a cumulant involving a repetition of the index $j$.  We use the symmetry to reorder the label sequence so that $i$ is moved to the end (position $n$) and $j$ is moved to the beginning (position $1$).
	This yields $\mean{\wick{e_{1_n-(1\times i)}}e_j}
	= \kappa_t[e_1,e_{1_{n-1}}]$.

We recall that $\int_{-\pi}^{\pi} \frac{\rmd \theta}{2\pi} \sin(\theta)^2 = \frac{1}{2}$ which allows to conclude that the  right hand side of (\ref{eq:non-repeated-one-particle-1}) is equal to
\begin{align}
	&
	-\frac{n(n-1)}{2(N-1)}\left(
	\kappa_t[e_{1_n}]-\kappa_t[e_1,e_{1_{n-1}}]\right)
	\,.
	\label{eq:non-repeated-one-particle-1'}
	\end{align}
An analogous argument can be applied to compute the second term, yielding
	\begin{align}
	&\frac{1}{N-1}\sum_{i,j=1}^N \cf{i\neq j} \sum_{k=1}^n \cf{k=j} \int_{-\pi}^{\pi} \frac{\rmd \theta}{2\pi}\mean{\wick{e_{1_n-(1\times k)}}Q_\theta(v_i,v_j)}  \nonumber \\
	&=\frac{1}{N-1}\sum_{j=1}^n \sum_{i=1}^N \cf{i\neq j}\int_{-\pi}^{\pi} \frac{\rmd \theta}{2\pi} \left((-\sin(\theta)^2)\mean{\wick{e_{1_n-(1\times j)}}e_j} + \sin(\theta)^2\mean{\wick{e_{1_n-(1\times j)}}e_i}\right) \nonumber \\
	&= -\frac{n(n-1)}{2(N-1)}\left(
	\kappa_t[e_{1_n}]-\kappa_t[e_1,e_{1_{n-1}}]\right)
	\,.
	\label{eq:non-repeated-one-particle-2}
	\end{align}
	
	To study the final, third term in \eqref{eq:time-evolution-non-repeated-cumulants-wick-form}, we first note that the even part of $P_\theta(v_i,v_j)Q_\theta(v_j,v_i)$ may be simplified to $2 \sin(\theta)^4 e_i e_j-  \sin(\theta)^4 (e_i^2 + e_j^2) - 4
	\cos(\theta)^2\sin(\theta)^2 e_i e_j$.
    Therefore, after relabeling, we can rewrite the term as
	\begin{align}
	&\frac{n(n-1)}{N-1}\int_{-\pi}^{\pi}\frac{\rmd \theta}{2\pi}\left(-2\sin(\theta)^4\mean{\wick{e_{1_n-((1\times 1)+(1\times n))}}e_{1}^2}\right) \nonumber \\
	&+ \frac{n(n-1)}{N-1}\int_{-\pi}^{\pi}\frac{\rmd \theta}{2\pi}\left((2\sin(\theta)^4 - 4\cos(\theta)^2 \sin(\theta)^2 )\mean{\wick{e_{1_n-((1\times (n-1))+(1\times n))}}e_{n-1}e_n}\right)\,.
	\label{eq:non-repeated-one-particle-31}
	\end{align}
    The  truncated moments to cumulants formula
    now yields terms with one or two clusters.  Explicitly,
    \begin{align}
     &     \mean{\wick{e_{1_n-((1\times 1)+(1\times n))}}e_{1}^2}
  = \kappa[e_1,e_{1_{n-1}}]
   +\sum_{J\subseteq 1_n-((1\times 1)+(1\times n))}
  \kappa[e_1,e_{J}] \kappa[e_1,e_{J^c}]
  \nonumber \\ & \quad
  =
  \kappa[e_1,e_{1_{n-1}}] + \sum_{k=0}^{n-2}
  \binom{n-2}{k} \kappa[e_{1_{k+1}}]\kappa[e_{1_{n-1-k}}]\,,
  \end{align}
   and
    \begin{align}
     &  \mean{\wick{e_{1_n-((1\times (n-1))+(1\times n))}}e_{n-1}e_n}
  =
  \kappa[e_{1_{n}}] + \sum_{k=0}^{n-2}
  \binom{n-2}{k} \kappa[e_{1_{k+1}}]\kappa[e_{1_{n-1-k}}]\,.
  \end{align}
  As shown in Appendix \ref{sec:trigonometric-integrals},
 the remaining integrals over $\theta$ evaluate to
	\[
	 \int_{-\pi}^{\pi} \frac{\rmd \theta}{2\pi} \sin(\theta)^4 = \frac{3}{8}\,,\qquad
	 \int_{-\pi}^{\pi} \frac{\rmd \theta}{2\pi} 4 \cos(\theta)^2\sin(\theta)^2 = \frac{1}{2}\,,
	\]
and thus the third term is equal to
	\begin{align}
	&- \frac{3}{4}\frac{n(n-1)}{N-1}
	 \kappa[e_1,e_{1_{n-1}}]
	 + \frac{1}{4}\frac{n(n-1)}{N-1}
	   \kappa[e_{1_{n}}]
	   +
	 \mathcal{N}_{<n}(t)\,,
	\label{eq:non-repeated-one-particle-3}
	\end{align}
where the nonlinear term \(\mathcal{N}_{<n}(t)\)
depends only on lower cumulants of order lower than $n$.  It can be written as
	\begin{align}\label{eq:defNlower}
	\mathcal{N}_{<n}(t) &= -\frac{n(n-1)}{4(N-1)}\sum_{m=1}^{n-1}A_{m,n}\kappanr{t}{m}\kappanr{t}{n-m},
	\end{align}
	where the combinatorial terms \(A_{m,n}\) are given by
	\begin{align}
	A_{m,n} = \frac{2(n-2)!}{(m-1)!(n-m-1)!}.
	\end{align}

Summing up the three results from \eqref{eq:non-repeated-one-particle-1'}, \eqref{eq:non-repeated-one-particle-2} and \eqref{eq:non-repeated-one-particle-3}, we thus find
	\begin{align}
	\dv{t}\kappanr{t}{n} &= -\frac{3n(n-1)}{4(N-1)}\kappanr{t}{n} + \frac{n(n-1)}{4(N-1)}\kappa_t[e_1,e_1,e_2,\dots, e_{n-1}]+ \mathcal{N}_{<n}(t).
	\label{eq:non-repeated-evolution}	
	\end{align}
As indicated in the introduction, this evolution equation is linear for the order $n$ cumulants and has a non-linear source term involving the lower order cumulants.  However, we can simplify the linear part further, by using the conservation law.

Namely, by the permutation symmetry of \(F_t^N\), we have
	\begin{align}
\kappa_t[N, e_2,\dots, e_n] &= \sum_{i=1}^N \kappa_t[e_i,e_2,\dots, e_n] \nonumber \\
	&= (n-1)\kappa_t[e_1,e_1,\dots, e_{n-1}] + (N-(n-1))\kappanr{t}{n}.
	\end{align}
	Any joint cumulant involving a constant random variable and some other random variables vanishes. In particular, \(\kappa_t[N, e_2,\dots, e_n] = 0\)
	above.   Therefore, we have the following relationship between the completely non-repeated energy cumulants and energy cumulants with one repeated particle label:
	\begin{align}
	\kappa_t[e_1,e_1,\dots, e_{n-1}] =
	- \frac{N-(n-1)}{n-1}\kappanr{t}{n}\,.
	\end{align}
	
	Plugging this into \eqref{eq:non-repeated-evolution}, we obtain
	\begin{align}
	\dv{t}\kappanr{t}{n} &= -D_{n,N}\kappanr{t}{n}
	+\mathcal{N}_{<n}(t), \quad 2\leq n  \leq N\,,
	\label{eq:non-repeated-evolution'}
	\end{align}
	where the dissipation constant is given by
	\begin{align}
	D_{n,N} =
	\frac{n}{4} + \frac{2n^2 - n}{4(N-1)}\,, \quad 2 \leq n \leq N
	\label{eq:constant_CnN}
	\end{align}
   and we arbitrarily also set $D_{1,N}=0$ for further use.

	To solve the system of equations \eqref{eq:non-repeated-evolution'}, we define a family of new dynamic variables by setting
	\begin{align}
	h^\alpha_n(t) = h_n^{\alpha,N}(t) \coloneqq (-1)^{n-1}\frac{(N-1)^{\alpha(n-1)}}{(n-1)!}\kappanr{t}{n}.
	\end{align}
	By assumption on $B$, for any $n\le n^*$,
	\begin{align}
	N^{\alpha(n-1)}\abs{\kappa_0(e_{1_n})} \leq B^{n-1} (n-1)! N^{c (n-1)}\,,
	\end{align}
  and, therefore, also
  \[
   \abs{h_n^{\alpha}(0)} \le B^{n-1} N^{\gamma_n}\,,
  \]
  where $\gamma_n:=c(n-1)$.
  Our goal is to prove  (\ref{eq:nrgengoal}) which is equivalent to the statement
  \begin{align}\label{eq:nrgoalinhvar}
   \abs{h^\alpha_n(t)} \leq C^{n-1} (\rme^{-\frac{n}{4}t}N^{\gamma_n} + 1)\,,
  \end{align}
  in the new variables.

	From \eqref{eq:non-repeated-evolution'} and (\ref{eq:defNlower}), it follows that these variables satisfy
	\begin{align}
	\dv{t}h^\alpha_n(t) = -D_{n,N} h^\alpha_n(t) + \frac{n}{2(N-1)^{1-\alpha}} \sum_{m=1}^{n-1}h^\alpha_m(t)h^\alpha_{n-m}(t).
	\end{align}
		This results implies the Duhamel formula
	\[
	 h^\alpha_n(t) = \rme^{-D_{n,N} t} h^\alpha_n(0)
	 + \frac{n}{2(N-1)^{1-\alpha}} \sum_{m=1}^{n-1}
	\rmd s\,.
	\]
	Let \(\epsilon = \epsilon(n,N,\alpha) = \frac{n^2}{2(N-1)^{1-\alpha}}\) and \(\epsilon' = \frac{n}{2(N-1)^{1-\alpha}}\) Then,
	\begin{align}
	\abs{h^\alpha_n(t)} \leq \rme^{-D_{n,N}t}\abs{h^\alpha_n(0)} &+ 2\epsilon' \rme^{-D_{n,N}t} \int_{0}^t \rme^{D_{n,N}s}\abs{h_1^\alpha(s)}\abs{h_{n-1}^\alpha(s)} \\
	& + \epsilon' \rme^{-D_{n,N}t} \sum_{m=2}^{n-2}\int_{0}^{t} \rme^{D_{n,N}s} \abs{h^\alpha_m(s)}\abs{h^\alpha_{n-m}(s)} \rmd s.
	\label{ineq:h_n-duhamel-estimate}
	\end{align}

    If \(n=1\), we have $h^\alpha_1(t)=\kappanr{t}{1}=1$ and thus
    (\ref{eq:nrgoalinhvar}) holds for any choice of $C\ge 1$ and all $t\ge 0$.  Consider then $n\in [2,n^*]$.
    We next make an induction assumption that $C$ has been chosen so that it holds for all orders up to  $n-1$ and any $t\ge 0$,
    and our goal is to show that the result is then true also at $n$.

    Trying to bound the second term in \eqref{ineq:h_n-duhamel-estimate} there are two essentially different cases to consider: (i) the edge terms, where \(m=1\) or \(m=n-1\), and (ii) the bulk term, where \(m \in [2,n-2]\).  Also, the term is symmetric under $m\mapsto n-m$, so it
	it suffices to consider $m=1$ in case (i) and $2\le m\le \frac{n}{2}$ in case (ii).

	If $m=1$, the induction assumption, $h^\alpha_1=1$, and $D_{n,N}\ge \frac{n}{4}$ imply that
	\begin{align*}
	&\epsilon' \rme^{-D_{n,N}t} \int_{0}^{t} \rme^{D_{n,N}s} \abs{h^\alpha_1(s)}\abs{h^\alpha_{n-1}(s)} \rmd s
	\\
	& \le
	\epsilon' \int_{0}^t \rme^{-D_{n,N}(t-s)} C^{n-2} (1+\rme^{-\frac{n-1}{4}s}N^{\gamma_{n-1}} )
	 \rmd s \\
	& \le
	\epsilon' C^{n-2}\int_{0}^t \rme^{-\frac{n}{4}(t-s)}  (1+\rme^{-\frac{n-1}{4}s}N^{\gamma_{n-1}} )
	 \rmd s \\
	&\leq 4 \epsilon' C^{n-2} \left(\frac{1}{n}
	+
	\rme^{-\frac{n-1}{4} t}
 N^{\gamma_{n-1}}\right)
	 \,.
	\end{align*}
Here, if \(t \geq 4c\frac{n-2}{n-1} \log(N) + 4\frac{\log(n)}{n-1}\), then
$\rme^{-\frac{n-1}{4} t} N^{\gamma_{n-1}}\le \rme^{-c(n-2)\log(N)}N^{\gamma_{n-1}}\rme^{-\log(n)}\leq 1/n$,
and, otherwise, $\rme^{\frac{t}{4}}\le  N^{c\frac{n-2}{n-1}}n^{1/(n-1)}$. Therefore,
\begin{align*}
\frac{\rme^{-\frac{n-1}{4}t}N^{\gamma_{n-1}}}{\rme^{-\frac{n}{4}t}N^{\gamma_n}} \leq N^{c\frac{n-2}{n-1}-c}n^{1/(n-1)} = N^{-c/(n-1)}n^{1/(n-1)}.
\end{align*}
Therefore, the contribution in case (i) is bounded by
\[
 4\epsilon' C^{n-2} N^{-c/(n-1)}n^{1/(n-1)}N^{\gamma_n}\rme^{-\frac{n}{4}t} +  8\frac{\epsilon'}{n} C^{n-2}\,.
\]

If case (ii), we can assume $2\le m\le \frac{n}{2}$.
As above, we then obtain an estimate
\begin{align*}
	&\epsilon' \rme^{-D_{n,N}t} \int_{0}^{t} \rme^{D_{n,N}s} \abs{h^\alpha_m(s)}\abs{h^\alpha_{n-m}(s)} \rmd s
	\\
	& \le
	\epsilon' C^{n-2}\int_{0}^t \rme^{-\frac{n}{4}(t-s)}  \left(1+
	\rme^{-\frac{m}{4}s}N^{\gamma_{m}}
	+	\rme^{-\frac{n-m}{4}s}N^{\gamma_{n-m}}
	+ 	\rme^{-\frac{n}{4}s}
	N^{\gamma_{m}}N^{\gamma_{n-m}}
 \right)
	 \rmd s \\
	&\leq 4 \epsilon' C^{n-2} \left(\frac{1}{n}
	+ \frac{1}{n-m}\rme^{-\frac{m}{4}t }N^{\gamma_{m}}
+ \frac{1}{m}\rme^{-\frac{n-m}{4}t }N^{\gamma_{n-m}}
+
	\frac{t}{4} \rme^{-\frac{n}{4}t}N^{-c} N^{\gamma_n}
\right)
	 \,.
	\end{align*}
Here, each of the time-dependent terms can be bounded as above: for each of the terms, we split at $t\ge  4c \frac{k-1}{k}\log(N)$, where $k=m$ for the second term, $k=n-m$ for the third, and $k=n-1$ for the fast term.  Then, $\rme^{-\frac{k}{4} t} N^{\gamma_{k}}\le 1$ for $t\ge 4c \frac{k-1}{k}\log(N)$,
and, otherwise, $\rme^{\frac{t}{4}}\le N^c$.  Since $x\rme^{-x}\le \rme^{-1}$ for $x\ge 0$, We obtain an upper bound
\[
  4 \epsilon C^{n} \left(\frac{1}{n}
	+ \frac{1}{n-m}
	+ \frac{1}{m}
	+ \rme^{-1} \right) \left(1+
 \rme^{-\frac{n}{4}t} N^{\gamma_n}
\right)
\le
  8 \epsilon C^{n} \left(1+
 \rme^{-\frac{n}{4}t} N^{\gamma_n}
\right)\,.
\]

Summarizing, the induction assumption implies that
\[
 \abs{h^\alpha_n(t)} \leq \rme^{-\frac{n}{4} t}
  B^n N^{\gamma_n}
  +  8 \epsilon C^{n} \left(1+
 \rme^{-\frac{n}{4}t} N^{\gamma_n}
\right)
\,.
\]
Therefore, for example, if $C\ge 2 B$ and $8\vep \le \frac{1}{2}$,
we have
$\abs{h^\alpha_n(t)} \leq  C^{n} \left(1+
 \rme^{-\frac{n}{4}t} N^{\gamma_n}
\right)$.  This proves that the choices $C=2B>1$ and $N_0=1+(8n_*^2)^{\frac{1}{1-\alpha}}$ suffice for the statement in the Proposition.
	\end{proofof}

\subsection{Convergence to equilibrium for non-repeated energy cumulants}\label{sec:nonrepequil}

For the estimate for convergence to stationarity in the non-repeated case,
we first recall from the previous section that if  functions \(h_n\) are defined for \(n \in [N]\) as
\begin{align}
h_n(t) \coloneqq (-1)^{n-1}\frac{(N-1)^{n-1}}{(n-1)!}\kappanr{t}{n}\,,
\end{align}
then these rescaled cumulants satisfy the following time-evolution equation
\begin{align}
	\dv{t}h_n(t) = -D_{n,N} h_n(t) + \frac{n}{2} \sum_{m=1}^{n-1}h_m(t)h_{n-m}(t), \quad 2\le n \le N\,.
	\label{eq:non-repeated-maximally-chaotic-rescaling}
	\end{align}
This follows from \eqref{eq:non-repeated-evolution'} and corresponds to the limiting case of $\alpha\to 1$ for $h_n^\alpha$ defined in the previous section.

Let \((\bar{h}_n)_{n\in [N]} = (\bar{h}^N_n)_{n\in [N]}\) be defined by the non-repeated cumulants in the stationary, uniform measure.
Since this yields a stationary initial state, the right hand side of \eqref{eq:non-repeated-maximally-chaotic-rescaling} must be zero.
	\begin{lemma}
	The stationary sequence \((\bar{h}_n)_{n\in [N]}\) has the following upper bound for all \(n\)
	\begin{align}
	\abs{\bar{h}_n} \leq 8^{n-1}.
	\label{ineq:stationary-non-repeated-bound}
	\end{align}
	\end{lemma}
	\begin{proof}
	Since \eqref{eq:non-repeated-maximally-chaotic-rescaling} evaluates to zero, we have a recursion relation
	\begin{align}\label{eq:nonrepstationaryiteration}
  	\bar{h}_n = \frac{n}{2D_{n,N}}\sum_{m=1}^{n-1}\bar{h}_m\bar{h}_{n-m}\,.
  	\end{align}
	Denote $b_n := \abs{\bar{h}_n}$, and we obtain that $b_1=1$ and for $n\ge 2$
  	\begin{align}
  	b_n \leq 2\sum_{m=1}^{n-1}b_m b_{n-m}.
  	\label{ineq:non-repeated-recursion-upper-bound}
  	\end{align}
  	
  	The sequence \((a_n)_{n=1}^\infty\) with \(a_n = \frac{2^n}{4(2n-1)} \frac{(2n)!}{(n!)^2}\) is the unique sequence that satisfies
  	\begin{align}
  	\begin{cases}
  	a_1 &= 1\,, \\
  	a_{n} &= 2 \sum_{m=1}^{n-1} a_{m}a_{n-m}, \quad n \geq 2 \,.
  	\end{cases}
  	\label{eq:recursion_an}
  	\end{align}
	This can be solved by a simple rescaling argument, by first noting that the sequence \(a_n\) solves \eqref{eq:recursion_an} if and only if the sequence \((c_n)_{n=1}^\infty\), with \(c_n \coloneqq \frac{a_n}{2^{n-1}}\), solves 
	\begin{align}
  	\begin{cases}
  	c_1 &= 1\,, \\
  	c_{n} &= \sum_{m=1}^{n-1} c_{m}c_{n-m}, \quad n \geq 2 \,.
  	\end{cases}
  	\label{eq:recursion_cn}
  	\end{align}
  	The sequence \((c_n)_{n=1}^\infty\) solves \eqref{eq:recursion_cn} if and only if the sequence \((\tilde{c}_n)_{n=0}^\infty\) with \(\tilde{c}_n = c_{n+1}\) solves
  	\begin{align}
  	\begin{cases}
  	\tilde{c}_0 &= 1\,, \\
  	\tilde{c}_{n} &= \sum_{m=1}^{n} c_{m-1}c_{n-m}, \quad n \geq 1 \,.
  	\end{cases}
  	\label{eq:recursion_tilde_cn}
  	\end{align}
  	This recurrence relation defines the \emph{Catalan numbers}, see for instance \cite{singmaster_some_1979}, and now \eqref{eq:recursion_an} is uniquely solved by the sequence with \(a_n = 2^{n-1}\tilde{c}_{n-1}\) for \(n \geq 1\).
  	
  	But by \eqref{ineq:non-repeated-recursion-upper-bound}, the relation $0\le b_n\le a_n$ is true for $n=1$ and  inductively propagates for all $n\le N$.
  	Therefore, for each \(n \in [N]\), we have proven a bound
  	\begin{align}
  	|\bar{h}_n |\leq a_n \leq \frac{2^n}{4(2n-1)}\frac{2^{2n}(n!)^2}{2(n!)^2} = \frac{8^{n-1}}{(2n-1)}\,.
  	\end{align}
  	This concludes the proof of the Lemma.
	\end{proof}
	
	We now have an idea of how the stationary non-repeated cumulants are bounded. Now we can linearize the system. Using the difference variables \(\xi_n(t) \coloneqq h_n(t) - \bar{h}_n\), we can write \eqref{eq:non-repeated-maximally-chaotic-rescaling} as
	\begin{align}
	\dv{t}h_n(t) &= -D_{n,N} (\xi_n(t)+ \bar{h}_n) + \frac{n}{2}\sum_{m=1}^{n-1}\left(\xi_m(t) + \bar{h}_n\right)\left(\xi_{n-m}(t) + \bar{h}_{n-m}\right) \nonumber \\
	&= -D_{n,N} \xi_n(t) + \frac{n}{2}\sum_{m=2}^{n-2}\xi_m(t)\xi_{n-m}(t)+ n\sum_{m=2}^{n-1}\xi_m(t)\bar{h}_{n-m},
	\end{align}
	where three of the four end points in the two sums vanish because \(\xi_1(t) \equiv 0\). Consequently,
	\begin{align}\label{eq:nrxinevol}
	\dv{t}\xi_n(t) &= -D_{n,N} \xi_n(t) + \frac{n}{2}\sum_{m=2}^{n-2}\xi_m(t)\xi_{n-m}(t)+ n\sum_{m=2}^{n-1}\xi_m(t)\bar{h}_{n-m},
	\end{align}

	\begin{proposition}
	Suppose $A\ge 0$ is given such that initially
	\begin{align}
	\abs{\xi_n(0)} \leq A^{n-1} (n-1)!\,.
	\end{align}
	There exists a pure constant \(C'\ge 1\) such that
	\begin{align}\label{eq:goaldiffnr}
	\abs{\xi_n(t)} \leq A^{n-1} (C')^{n-1} (n-1)! \rme^{-\frac{t}{2}}
	\end{align}
	for all \(t \geq 0\) and any  \(n \in [N]\).
	For example, $C'=17$ is possible here.
	\label{proposition:convergence-to-equilibrium-non-repeated-cumulants}
	\end{proposition}
	
	\begin{remark}
		From this result, it follows that
		if $A\ge 16$ is such that
		\[
		 |\kappanr{0}{n}| \le \frac{1}{2} A^{n-1}
		 (n-1)! (N-1)^{-(n-1)} \,,
		\]
        then $\abs{\xi_n(0)} \leq A^{n-1} (n-1)!$, and hence
		\begin{align}
		\abs{\kappanr{t}{n}-\bar{\kappa}[e_{1_n}]}
		=   (n-1)! (N-1)^{-(n-1)} \abs{\xi_n(t)}
		\leq
		 \left(\frac{C' A}{N-1}\right)^{n-1} ((n-1)!)^2 \rme^{-\frac{t}{2}}\, .
		\end{align}
		The generic bound in Proposition \ref{prop:generic-bounds-for-initial-data},
		implies that the above assumption holds for any symmetric measure using $A=16 (N-1)$.
    Thus even for highly non-chaotic states, non-repeated cumulants convergence to the stationary state on time-scales which are slightly worse than linear in the order of cumulants, independently of the number of particles $N\ge 2$.
    Of course, if relative convergence is required, i.e., $\kappanr{t}{n}/\bar{\kappa}[e_{1_n}]$, the necessary time-scale will again be order of $\ln N$ since the limit value is small.
	\end{remark}
		
	\begin{proof}
	We begin with the case $n=2$ for which \eqref{eq:nrxinevol} simplifies to
	\[
	 \dv{t}\xi_2(t) = -D_{2,N} \xi_2(t)
	\]
    where, by \eqref{eq:constant_CnN},
    $D_{2,N}\ge \frac{1}{2}$.   Therefore,
    \[
     \xi_2(t) = \rme^{-D_{2,N} t} \xi_2(0)\,,
    \]
     and it satisfies
     \[
      |\xi_2(t)|\le A \rme^{-\frac{t}{2}}\,.
     \]
    Therefore, \eqref{eq:goaldiffnr} holds for any $C'\ge 1$ if $n=2$.

	Assume then that $n\ge 3$ and $C'\ge 1$ is such that
      \eqref{eq:goaldiffnr} for for all values less than $n$.
    We exponentiate the first term into a Duhamel formula and estimate $D_{2,N}\ge \frac{n}{4}$.
    This shows that
    \begin{align*}
    & |\xi_n(t)| \\
    & \quad \le \rme^{-\frac{n}{4} t } |\xi_n(0)|
    + \int_0^t \rme^{-\frac{n}{4} (t-s) }   \frac{n}{2}\sum_{m=2}^{n-2}|\xi_m(s)| |\xi_{n-m}(s)|\rmd s+ n\sum_{m=2}^{n-1}
     \int_0^t \rme^{-\frac{n}{4} (t-s) }
    |\xi_m(s)|\, |\bar{h}_{n-m}|\rmd s\,.
    \end{align*}
    The first term is bounded by assumption by
    \[
      A^{n-1} (n-1) ! \rme^{-\frac{3}{4} t}\,.
    \]
    The second term is zero for $n=3$.  For $n\ge 4$, it may be estimated using the induction assumption by
    \begin{align*}
    &
      \int_0^t \rme^{-\frac{n}{4} (t-s) }   \frac{n}{2}\sum_{m=2}^{n-2}|\xi_m(s)| |\xi_{n-m}(s)|\rmd s
     \\ & \quad
     \le
      \int_0^t \rme^{-\frac{n}{4} (t-s) } \rme^{-s}\rmd s\, A^{n-2}  (C')^{n-2}
      \frac{n}{2}\sum_{m=2}^{n-2}
       (m-1)! (n-m-1)!\,.
    \end{align*}
Here, $ \int_0^t \rme^{-\frac{n}{4} (t-s)-s } \le
\int_0^t \rme^{-\frac{n}{4} (t-s)-\frac{1}{2}s }
\le \frac{4}{n-2}\rme^{-\frac{1}{2} t}$.
On the other hand,
\[
 (m-1)! (n-m-1)! = (n-2)! \binom{n-2}{m-1}^{-1}
 \le (n-2)! \,.
\]
Therefore,
the second term has an upper bound
\[
A^{n-2}  (C')^{n-2}
      \frac{4 n}{2(n-2)} (n-2)! (n-3)
      \rme^{-\frac{1}{2} t}
      \le 4 A^{n-1}\rme^{-\frac{1}{2} t}(n-1)! (C')^{n-2} \,.
\]
where in the last step used the assumption that for $n\ge 4$.  The final bound is also a bound for the zero which occurs when $n=3$.

Finally, the third term can be estimated by using
Lemma \ref{ineq:stationary-non-repeated-bound}.
We obtain
 \begin{align*}
    & n\sum_{m=2}^{n-1}
     \int_0^t \rme^{-\frac{n}{4} (t-s) }
    |\xi_m(s)|\, |\bar{h}_{n-m}|\rmd s
  \\ & \quad
  \le
  n\sum_{m=2}^{n-1} A^{m-1} (C')^{m-1}
  (m-1)!8^{n-m-1}
     \int_0^t \rme^{-\frac{n}{4} (t-s)-\frac{1}{2} s }\rmd s
  \\ & \quad
  \le \frac{4 n}{n-2}\rme^{-\frac{1}{2} t}
  \sum_{m=2}^{n-1} A^{m-1}(C')^{m-1}
  (m-1)!8^{n-m-1}
  \\ & \quad
  \le \frac{4 n}{n-2}\rme^{-\frac{1}{2} t}
   (A C')^{n-2} (n-2)!
  \sum_{m=2}^{n-1} (C')^{-(n-m-1)} 8^{n-m-1}
     \,.
    \end{align*}
Now if $C'\ge 16$, we have
$  \sum_{m=2}^{n-1} (C')^{-(n-m-1)} 8^{n-m-1} \le 2$.
Then, the last term may be bound by
\[
\rme^{-\frac{1}{2} t}
   (A C')^{n-2} (n-1)! \frac{8 n}{(n-2)(n-1)}\le
   12 \rme^{-\frac{1}{2} t}
   (A C')^{n-2} (n-1)! \,,
\]
where in the last bound we have used $n\ge 3$.

We find that if $C'\ge 16$, the induction assumption implies the following bound for any $n\ge 3$
 \begin{align*}
    & |\xi_n(t)| \le
     (A C')^{n-1} (n-1) ! \rme^{-\frac{t}{2}}
    \left( (C')^{-2} +
    \frac{4}{C'} + \frac{12}{C'}
    \right)
\,.
 \end{align*}
Choosing for example $C'=17$ yields a number less than one for the sum in parenthesis.  Hence, in this case the induction assumption propagates and we have concluded the proof of the theorem.
	\end{proof}
	
	\section{Repeated energy cumulants}
	\label{sec:Repeated}

	\subsection{Generation of chaos for repeated energy cumulants}
	
	The time-evolution of the repeated energy cumulants proceeds as follows. First, we analyze the linear part of the evolution. This splits into three parts. The first part was already established in the previous section, where we controlled the completely non-repeated energy cumulants. After this, we will write the rest of the linear part as a sum of the main term and a perturbation term, which is at most \(o_N(1)\) in a suitable norm, taken here to be \(\norm{\cdot}_{\alpha,n,N}\) for all suitably large \(N\). Exponential decay for the semigroup generated by the main term implies exponential decay for the full linear part, provided that we take \(N\) large enough in comparison to \(n\), which is the previously fixed maximal order of cumulants. After controlling the linear part, we will prove bounds that propagate over the non-linear term.
	
	In order to propagate the smallness of the cumulants over the nonlinear term, we have to prove that it does not produce unwanted correlations which would make the cumulants large. We will therefore need the following lemma, which quantifies the combined ``size of chaos'' of the nonlinear term.

	To this end, if $R$ is a set of distinct indices for a sequence \((a,i_a)\), $a\in R$, of particle labels $i_a$, we can identify the sequence with a coloring map $c$ defined by $c(a)=i_a$, $a\in R$.  The coloring map then also assigns to each position index $a$ in the sequence the corresponding random  variable \(e_{c(a)} = e_{i_a}\).  We denote $e_A := \prod_{i\in A} e_{c(i)}$ if $A\subset R$.  Then \(\abs{c(A)} = \len(s_A)\), where \(s_A\) is the partition classifier corresponding to the subsequence given by $A\subset R$ and \(c\).

	\begin{lemma}
	Let \(R\) be a finite set and \(c \colon R \to \N\) be a coloring of this set, i.e., an assignment of a natural number for each of the elements in \(R\). Suppose that \(J \subset R\) is nonempty and denote \(m:=|c(J) \cap c(R \setminus J)|\). Let \(\{A_\ell\}_{\ell=1}^k\) be a partition of \(R\) for which none of the sets \(A_\ell\) is internal to \(R \setminus J\); in other words, we require \(A_\ell \cap J \neq \emptyset\) for all $\ell$. Then
	\[\sum_{\ell=1}^k \abs{c(A_\ell)} \geq \abs{c(R\setminus J)} + k-m\,.\]
	\label{lemma:nonlinearity-coloring}
	\end{lemma}
	\begin{proof}
	Let $I:= R\setminus J$.  By assumption,
	\(\{A_\ell\}_{\ell=1}^k\) forms a partition of $R$ such that \(A_\ell \cap J \neq \emptyset\) for all $\ell$.  Let us denote the colors which are not covered by those given by $J$ by $D_1 := c(R)\setminus c(J)= c(I)\setminus c(J)$ and its preimage by $R_1 := c^{-1}(D_1)$.  Clearly, $R_1\cap J=\emptyset$.  This implies that for each $\ell$ we
	must have $c(A_\ell)\setminus c(A_\ell\cap R_1)\ne \emptyset$ and, therefore,
	$\abs{c(A_\ell)}\ge 1 + |c(A_\ell\cap R_1)|$.
	Since $m=|c(I)|-|D_1|$, we find that
	\begin{align}
	\sum_{\ell=1}^k \abs{c(A_\ell)} \geq
	 \sum_{\ell=1}^k 1 +
	 \sum_{\ell=1}^k |c(A_\ell\cap R_1)| \geq
	 k + \abs{\bigcup_{\ell=1}^k c(A_\ell\cap R_1)}
	 = k + \abs{c(R_1)} = |c(I)|+k-m\,.
	\end{align}
	\end{proof}

	\begin{remark}
	Lemma \ref{lemma:nonlinearity-coloring} is useful in quantifying the total amount of chaos arising in the products of joint cumulants.
	Indeed, suppose that $I$ and $J$ are some given
	particle label sequences which share exactly one particle label between them. Let $R$ denote an index set for the sequence composed by first collecting the labels from the sequence $I$ and then the labels from the sequence $J$ and let $c$ denote the corresponding coloring, defined as above.  With a slight abuse of notation, we identify $I$ and $J$ with the corresponding subsets of $R$. Then $I=R\setminus J$ and
	\(|c(J) \cap c(R \setminus J)|=1\) so the conditions of the lemma are satisfied with  $m=1$.

    We
	recall the truncated moments to cumulants formula
	\begin{align}
	\mean{\wick{e_{I}}e_{J}}_t = \sum_{\pi \in \mathcal{P}(R)}\prod_{\ell=1}^{\abs{\pi}}\left(\cf{\pi_\ell \cap J \neq \emptyset}\kappa_t[e_{\pi_\ell}]\right)\,.
	\end{align}
	Now \(\abs{c(R)}\) is the length of the partition classifier that matches \(R\), while the lengths of the partition classifiers corresponding to each of \(\pi_\ell\) match \(\abs{c(\pi_\ell)}\). Since this length corresponds to some kind of smallness of \(\kappa_t[e_{\pi_\ell}]\), the above lemma together with the moments to cumulants formula allows us to control the size of the non-linear term.
	Indeed, the lemma says that
	for any partition $\pi$ which might yield a non-zero contribution to the sum, we have
	\(\sum_{\ell=1}^{|\pi|}( \len(s_{\pi_\ell})-1) \geq \len(s_{I})-1\).
	\end{remark}
	
	We are now ready to prove the main result.
	
	\begin{proofof}{Theorem \ref{thm:generation-of-alpha-chaotic-bounds}}
	\paragraph{Step 1: Time-evolution of the energy cumulants}
	Let \(n \in \N_1\) a fixed order of cumulants. As noted, the full evolution of the energy cumulant hierarchy at order \(n\) consists of a source term, which can be controlled inductively provided that we have a bound for the lower level cumulants, and a linear term, which couples cumulants of order \(n\) to other cumulants of order \(n\). Since the initial measure \(F_0^N\) is assumed to be permutation invariant, the same holds for \(F_t^N\). Therefore, we can index the cumulants using \emph{partition classifiers}, as the particle labels do not matter.

	If \(s = (s_1,\dots, s_n)\) is a partition classifier, then \eqref{eq:time-evolution-of-energy-cumulants} gives us
	\begin{align}
	\dv{t}\kappa_t[e_s] &= T_t^{\text{break}}(s) + T_t^{\text{fuse}}(s)+ T_{t}^{\text{ex}}(s),
	\end{align}
	where we first set $L=L(s):=\len(s)\le n$ and then define
	\begin{align}\label{eq:Ttbreak}
	T_t^{\text{break}}(s)
	&= \frac{1}{N-1}\sum_{i=1}^L \sum_{j=L+1}^{N} \sum_{\ell=1}^{s_i} \binom{s_i}{\ell} \int_{-\pi}^{\pi}\frac{\rmd \theta}{2\pi} \mean{\wick{e_{s-(\ell \times i)}} P_\theta(v_i,v_j)^\ell} \nonumber \\
	&+\frac{1}{N-1}\sum_{j=1}^L \sum_{i=L+1}^{N} \sum_{\ell=1}^{s_j} \binom{s_j}{\ell} \int_{-\pi}^{\pi}\frac{\rmd \theta}{2\pi} \mean{\wick{e_{s-(\ell \times j)}} Q_\theta(v_j,v_i)^\ell}\,,
	\end{align}
		\begin{align}\label{eq:Ttfuse}
	T_t^{\text{fuse}}(s)
	&= \frac{1}{N-1}\sum_{i=1}^L \sum_{j=1}^{L} \cf{j\neq i} \sum_{\ell=1}^{s_i} \binom{s_i}{\ell} \int_{-\pi}^{\pi}\frac{\rmd \theta}{2\pi} \mean{\wick{e_{s-(\ell \times i)}}P_\theta(v_i,v_j)^\ell} \nonumber \\
	&+\frac{1}{N-1} \sum_{j=1}^{L}\sum_{i=1}^L \cf{j\neq i} \sum_{\ell=1}^{s_j} \binom{s_j}{\ell} \int_{-\pi}^{\pi}\frac{\rmd \theta}{2\pi} \mean{\wick{e_{s-(\ell \times j)}} Q_\theta(v_j,v_i)^\ell}\,,
	\end{align}
	\begin{align}\label{eq:Ttex}
	T_{t}^{\text{ex}}(s)
	&= \frac{1}{N-1}\sum_{i,j=1}^{L}\cf{i\neq j}\sum_{\ell=1}^{s_i} \sum_{\ell' = 1}^{s_{j}} \binom{s_i}{\ell}\binom{s_j}{\ell'}
\int_{-\pi}^{\pi} \frac{\rmd \theta}{2\pi}(-1)^{\ell'}\mean{\wick{e_{s-(\ell \times i) - (\ell' \times j)}} P_{\theta}(v_i,v_j)^{\ell+\ell'}}\,.
	\end{align}
	Note that for any $s\in \mathscr{C}_n$ we have
	$s_i\ge 1$ if $i\le L(s)$, and $s_i=0$ if $L(s)<i\le n$.  The names of these terms stem from the type of partition classifiers that contribute to the linear operator arising from each term: for the first term, the contribution comes from classifiers where one of the components has been ``broken'' into two pieces, for the second term, the contribution comes from ``fused'' components, and for the third term, the contribution comes for classifiers where some components have ``exchanged'' particles with each other.

    To compute the contributions from each term on the evolution of the cumulants, we recall the explicit form of $P_\theta$ and $Q_\theta$ given in \eqref{eq:energy-collision-polynomials-def}.  For any $m\in \N$, we can employ the multinomial theorem and conclude
\[
  P_\theta(v_i,v_j)^m
  = \sum_{a,b,c=0}^m \cf{a+b+c=m}\frac{m!}{a!b!c!} (-1)^a (\sin(\theta))^{2 a} e_i^a
  2^b (\cos(\theta))^b (\sin(\theta))^b v_i^b v_j^b  (\sin(\theta))^{2 c} e_j ^c\,.
\]
Therefore,
\begin{align}\label{eq:expectPm}
&  \int_{-\pi}^{\pi}\frac{\rmd \theta}{2\pi}P_\theta(v_i,v_j)^m
  = \sum_{k,h=0}^m \cf{k+h\le m,\, h\text{ even}}\frac{m!}{k!h!(m-k-h)!} (-1)^k  2^h
e_i^{k+\frac{h}{2}}
  e_j^{m-k-\frac{h}{2}}
  I_{m-\frac{h}{2},\frac{h}{2}}\,,
\end{align}
where we have defined
\begin{align}
	I_{a,b}:=\int_{-\pi}^{\pi}\frac{\rmd \theta}{2\pi}  (\sin(\theta)^2)^{a} (\cos(\theta)^2)^b\, \,, \qquad a,b\in \N_0\,.
	\end{align}
	Clearly, $0\le I_{a,b}\le 1$, and
as we show in Appendix \ref{sec:trigonometric-integrals}, these integrals can also be evaluated explicitly in terms of factorials.
Similar computation for $Q_\theta(v_j,v_i)$ yields

\begin{align}\label{eq:expectQm}
&  \int_{-\pi}^{\pi}\frac{\rmd \theta}{2\pi}Q_\theta(v_j,v_i)^m
  = \sum_{k,h=0}^m \cf{k+h\le m,\, h\text{ even}}\frac{m!}{k!h!(m-k-h)!} (-1)^{k}  2^h
e_j^{k+\frac{h}{2}}
  e_i^{m-k-\frac{h}{2}}
  I_{m-\frac{h}{2},\frac{h}{2}}\,.
\end{align}

Consider then  $a:=m-k-\frac{h}{2}$ in the above sums.  It is an integer which satisfies $0\le\frac{m-k}{2}\le a\le m-k\le m$, and the sums yield linear combinations of the random variables $e_i^{m-a}e_j^a$ in \eqref{eq:expectPm}
and $e_j^{m-a}e_i^a$ in \eqref{eq:expectQm}.
Therefore,  $T_t^{\text{break}}(s)$ and $T_t^{\text{fuse}}(s)$ only contain expectations of the form
\[
 \mean{\wick{e_{s-(\ell \times i)}}e_i^{\ell-a}e_j^a}\,, \qquad
 \mean{\wick{e_{s-(\ell \times j)}}e_j^{\ell-a}e_i^a}\,.
\]
The second expectation is obtained by swapping $i\leftrightarrow j$ in the first case, and it thus suffices to inspect the first term arising from $P_\theta(v_i,v_j)^\ell$.

Using the truncated moments-to-cumulants formula \eqref{eq:truncated-moments-to-cumulants} and assuming
$0\le a\le \ell$, $1\le \ell\le s_i$, we have
\begin{align}\label{eq:mainbreakexpect}
 & \mean{\wick{e_{s-(\ell \times i)}}e_i^{\ell-a}e_j^a} = \kappa[e_{s-(a\times i)+(a\times j)}]
 + (\text{nonlinear term})\,,
\end{align}
where the nonlinear term contains only cumulants of order strictly less that $n$ and it is a sum of terms each of which is a product of at least two but at most $\ell$ such cumulants.
For $T_t^{\text{break}}(s)$, we have $j>L$, so $s_j=0$.  Hence,
if $a=0$ or $a=s_i$, the partition classifier of $s-(a\times i)+(a\times j)$ is $s$ and, if $0<a<s_i$, the partition classifier of
$s-(a\times i)+(a\times j)$ belongs to
\(\texttt{break}_i(s)\) and thus has a length $L+1$.  In fact, as shown below, this is the only way the length of the partition classifier can increase.

In contrast,
for $T_t^{\text{fuse}}(s)$ we have $1\le j\le L$, $j\ne i$, and $s_j\ge 1$.  Thus only if $a=0$ is
the partition classifier of $s-(a\times i)+(a\times j)$ equal to $s$.  If $0<a< s_i$,
the partition classifier $r'$ of $s-(a\times i)+(a\times j)$ is not $s$ but it belongs to $\mathscr{C}'_n$ and has $\len(r')=L$.
Finally, $a=s_i$ is only possible if $\ell=s_i$.
Then again the partition classifier $r'$ of $s-(a\times i)+(a\times j)$ belongs to $\mathscr{C}'_n$
but it has $\len(r')=L-1$.

For the exchange term, we have $1\le i,j\le L$, with $i\ne j$, and thus $s_i,s_j\ge 1$.
We
use $m=\ell+\ell'$ which satisfies $2\le m \le s_i+s_j$.  Then
$T_{t}^{\text{ex}}(s)$ is a linear combination of
expectations
\begin{align}\label{eq:mainexexpect}
\mean{\wick{e_{s-(\ell \times i) - (\ell' \times j)}} e_i^{m-a}e_j^a}
= \kappa[e_{s-(\ell \times i) - (\ell' \times j)
+((m-a) \times i) + (a \times j)}]
 + (\text{nonlinear term})\,,
\end{align}
where $0\le a\le m=\ell+\ell'$ and the nonlinear term is a polynomial of lower order cumulants with highest possible power now given by $m$.
If follows that if $a<\ell'$, then
$s-(\ell \times i) - (\ell' \times j)
+((m-a) \times i) + (a \times j)=s+((\ell'-a) \times i) - ((\ell'-a) \times j)$ and its partition classifier corresponds to taking part of $s_j$ and adding it to $s_i$.  Similarly, if $a>\ell'$, then
$s-(\ell \times i) - (\ell' \times j)
+((m-a) \times i) + (a \times j)=s-((a-\ell') \times i) + ((a-\ell') \times j)$ which does the opposite.  The length of the so obtained partition classifier is $L$, apart from the extreme cases of $a=0$, $\ell'= s_j$, and $a=m$, $\ell = s_i$, when the length is $L-1$. Finally, if $a=\ell'$, then
$s-(\ell \times i) - (\ell' \times j)
+((m-a) \times i) + (a \times j)=s$.

	\paragraph*{Step 2: Control of the linear semigroup}

	As shown above,
	using the truncated moments-to-cumulants formula \eqref{eq:truncated-moments-to-cumulants}, we can write each of the expectations appearing in \(T^{\text{break}}\), \(T^{\text{fuse}}\) and \(T^{\text{ex}}\) in terms of products of cumulants. Let us denote by \(L_t^{\text{break}}(s), L_t^{\text{fuse}}(s), L_t^{\text{ex}}(s)\) those terms that contain a product involving only one cumulant (which has to be a cumulant of order \(n\)).
	Summaring from above, the first of these is explicitly given by
	\begin{align}
	& L_t^{\text{break}}(s) =
	\frac{1}{N-1}\sum_{i=1}^L \sum_{j=L+1}^{N} \sum_{\ell=1}^{s_i} \binom{s_i}{\ell}
	 \nonumber \\ & \quad\qquad\times
	 \sum_{k,h=0}^\ell
	 \cf{k+h\le \ell,\, h\text{ even}}
	 \frac{\ell!}{k!h!(\ell-k-h)!} (-1)^k  2^h
	 I_{\ell-\frac{h}{2},\frac{h}{2}} \kappa[e_{s-(a\times i)+(a\times j)}]|_{a=\ell-k-\frac{h}{2}} \nonumber \\
	& \qquad + (\text{same term but swap }i\leftrightarrow j) \nonumber \\ & \quad =
	2 \frac{N-L}{N-1}\sum_{i=1}^L\sum_{a=0}^{s_i}
	\kappa[e_{s-(a\times i)+(a\times (n+1))}]
	\nonumber \\ & \qquad\times
	  \sum_{\ell=1}^{s_i} \binom{s_i}{\ell}
	 \sum_{k,h=0}^\ell
	 \cf{k+h\le \ell,\, h\text{ even}}
	 \cf{a=\ell-k-\frac{h}{2}}
	 \frac{\ell!}{k!h!(\ell-k-h)!} (-1)^k  2^h I_{\ell-\frac{h}{2},\frac{h}{2}}\,,
	\label{eq:Lbreak-1}
	\end{align}
 where we have used symmetry to obtain the second simplified formula.
 Similarly, we have
	\begin{align}
	L_t^{\text{fuse}}(s) &=
	\frac{2}{N-1}\sum_{i=1}^L \sum_{j=1}^{L} \cf{j\neq i} \sum_{a=0}^{s_i}
	\kappa[e_{s-(a\times i)+(a\times j)}]
	\nonumber \\ & \qquad\times
	  \sum_{\ell=1}^{s_i} \binom{s_i}{\ell}
	 \sum_{k,h=0}^\ell
	 \cf{k+h\le \ell,\, h\text{ even}}
	 \cf{a=\ell-k-\frac{h}{2}}
	 \frac{\ell!}{k!h!(\ell-k-h)!} (-1)^k  2^h I_{\ell-\frac{h}{2},\frac{h}{2}}\,,
	\label{eq:Lfuse-1}
	\end{align}
	and
	\begin{align}
	L_t^{\text{ex}}(s) &=
	\frac{1}{N-1}\sum_{i,j=1}^{L}\cf{i\neq j}\sum_{\ell=1}^{s_i} \sum_{\ell' = 1}^{s_{j}} \binom{s_i}{\ell}\binom{s_j}{\ell'} (-1)^{\ell'}  \sum_{a=0}^{\ell+\ell'}
	\kappa[e_{s+((\ell'-a) \times i) - ((\ell'-a) \times j)}]\nonumber \\ & \qquad\times
	 \sum_{k,h=0}^{\ell+\ell'}
	 \cf{k+h\le \ell+\ell',\, h\text{ even}}
	 \cf{a=\ell+\ell'-k-\frac{h}{2}}
	 \frac{(\ell+\ell')!}{k!h!(\ell+\ell'-k-h)!} (-1)^k  2^h I_{\ell+\ell'-\frac{h}{2},\frac{h}{2}}\,.
	\label{eq:Lex-1}
	\end{align}	
	
	These can be written in a more suggestive way as 
	\begin{align}\label{eq:LtoCtilde}
	L^{\text{break}}_t(s) &= \sum_{r\in \mathscr{C}_n} \tilde{C}^{\text{break}}_{s,r} \kappa_t[e_r] \nonumber \\
	L^{\text{fuse}}_t(s) &= \sum_{r\in \mathscr{C}_n} \tilde{C}^{\text{fuse}}_{s,r} \kappa_t[e_r] \nonumber \\
	L^{\text{ex}}_t(s) &= \sum_{r\in \mathscr{C}_n} \tilde{C}^{\text{ex}}_{s,r} \kappa_t[e_r],
	\end{align}
	where
	\begin{align}
	& \tilde{C}^{\text{break}}_{s,r} =
	2 \frac{N-L}{N-1}\sum_{i=1}^L\sum_{a=0}^{s_i}
	\cf{r \sim {s-(a\times i)+(a\times (n+1))}}
	\nonumber \\ & \qquad\times
	  \sum_{\ell=1}^{s_i} \binom{s_i}{\ell}
	 \sum_{k,h=0}^\ell
	 \cf{k+h\le \ell,\, h\text{ even}}
	 \cf{a=\ell-k-\frac{h}{2}}
	 \frac{\ell!}{k!h!(\ell-k-h)!} (-1)^k  2^h I_{\ell-\frac{h}{2},\frac{h}{2}}\,,
	\label{eq:tilde_c_sr_break}
	\end{align}
  and similarly for \(L_t^\text{fuse}(s)\) and \(L_t^{\text{ex}}(s)\).

	The break term further simplifies to
	\begin{align}
	\tilde{C}_{s,r}^{\text{break}} &=
	2 \sum_{i=1}^{L(s)}\sum_{a=0}^{s_i}
	\cf{r \sim{s-(a\times i)+(a\times (n+1))}}
	\nonumber \\ & \qquad\times
	  \sum_{\ell=1}^{s_i} \binom{s_i}{\ell}
	 \sum_{k,h=0}^\ell
	 \cf{k+h\le \ell,\, h\text{ even}}
	 \cf{a=\ell-k-\frac{h}{2}}
	 \frac{\ell!}{k!h!(\ell-k-h)!} (-1)^k  2^h I_{\ell-\frac{h}{2},\frac{h}{2}} \nonumber \\
	& \quad -
	2 \frac{L(s)-1}{N-1}\sum_{i=1}^{L(s)}\sum_{a=0}^{s_i}
	\cf{r \sim {s-(a\times i)+(a\times (n+1))}}
	\nonumber \\ & \qquad\times
	  \sum_{\ell=1}^{s_i} \binom{s_i}{\ell}
	 \sum_{k,h=0}^\ell
	 \cf{k+h\le \ell,\, h\text{ even}}
	 \cf{a=\ell-k-\frac{h}{2}}
	 \frac{\ell!}{k!h!(\ell-k-h)!} (-1)^k  2^h I_{\ell-\frac{h}{2},\frac{h}{2}}
	 \nonumber \\
	&=: \tilde{C}_{s,r}^{\text{break};1} + \tilde{C}_{s,r}^{\text{break};2}\,.
	\end{align}
	 In terms of their behavior in the number of particles \(N\), the coefficients \(\tilde{C}^{\text{break};2}_{s,r}\), \(\tilde{C}^{\text{fuse}}_{s,r}\) and \(\tilde{C}^{\text{ex}}_{s,r}\) will be \(O(1/N-1)\).

	Recall that \(\mathscr{C}_n' = \mathscr{C}_n \setminus \{1_n\}\). This is the collection of all other partition classifiers, apart from the completely non-repeated classifier \(1_n\). In Proposition \ref{prop:non-repeated-stronger-claim-for-finite-order}, we  have already established bounds for non-repeated classifiers, so we can treat them as sources.
  We can conclude from the discussion in Step 1 that a non-zero term with \(r=1_n\) can occur
 in the sums in \eqref{eq:LtoCtilde}  only for $\tilde{C}^{\text{break}}_{s,r}$, and even then it must have
 $s=(2,1,1,\ldots,1,0)\in \mathscr{C}'_n$ and $i=1$.  Hence, in all other cases the sums can be reduced to go over $r\in \mathscr{C}'_n$.

	 Let \(M = M_{n,N} \in \mathcal{M}_{\mathscr{C}_{n}', \mathscr{C}_{n}'}\) be the \(\abs{\mathscr{C}_{n}'}\)-dimensional square matrix that describes how the time-evolution of the cumulant vector \((\kappa_t[e_s])_{s\in \mathscr{C}_{n}'}\) indexed by cumulants with a repeat is affected by other cumulants with a repeat of the same order. We can decompose this matrix further into two parts
	\begin{align}
	M = M_{n} + R_{n,N}.
	\label{eq:cumulant-matrix-decomposition}
	\end{align}
	Here \(M_n\) is the matrix with elements \((M_n)_{s,r} = \tilde{C}^{\text{break;1}}_{s,r}\), for \(s, r \in \mathscr{C}'_n\). On the other hand, \(R_{n,N}\) is a perturbation as an operator on the space \((X_{n,N}^\alpha, \norm{\cdot}_\alpha)\).
	
	We want to prove that the operator norm of \(\rme^{tM_{n,N}}\) decays in time, when we consider it as an operator \(X_{n,N}^\alpha \to X_{n,N}^\alpha\), provided that we have picked the number of particles \(N\) to be large enough. The following lemma shows that the main part of the semigroup decays exponentially in the norm \(\norm{\cdot}_\alpha\). Note that we have excluded \(n=1\) from the statements below since \(\mathscr{C}_1' = \emptyset\).

	\begin{lemma}
	For every \(\alpha \in (0,1)\) and \(n^* \in \N\), there exists \(N_0 = N_0(n^*, \alpha) \in \N\), such that
	\begin{align}
		\norm{\rme^{tM_{n}}}_\alpha = \norm{\rme^{tM_n}}_{\alpha, n,N} \leq 10\rme^{-\frac{1}{2}t},
	\end{align}
	whenever \(n \in [2,n^*]\) and \(N \geq N_0\).
	\label{lemma:limit-semigroup-convergence-in-alpha}
	\end{lemma}

	\begin{proofof}{Lemma \ref{lemma:limit-semigroup-convergence-in-alpha}}
	Let \(N_0(n^*, \alpha)\) be large enough so that \(	\frac{80 n^* 5^{n^*}}{N_0^{\alpha}} \leq 9\). Let \(\hat{\kappa}_0 = \hat{\kappa}_0^n = (\kappa_0[e_r])_{r \in \mathscr{C}_n'}\) be the vector of initial repeated energy cumulants.
	\begin{align}
	\hat{\kappa}_t(r) = (\rme^{tM_n}\hat{\kappa}_0)_{r}.
	\end{align}

	Define layer sets \(\mathscr{C}'_{n, \ell}\), \(\ell = 2,\dots, n\) of the space of partition classifiers  iteratively as follows.
	\begin{align}
		\mathscr{C}'_{n,2} &\coloneqq \{(2,1\dots, 1)\} \nonumber \\
		\mathscr{C}'_{n,\ell} &\coloneqq \{r \in \mathscr{C}'_n \colon \text{ there exists } r' \in \mathscr{C}'_{n,\ell-1} \text{ such that } r \succ r'\}\,, \qquad 3\le \ell \le n\,.
	\end{align}
	Note that if \(r \in \mathscr{C}'_{n,\ell}\), then \(\len(r) = n-(\ell-1)\).
	
	We can solve \(\hat{\kappa}_t(r)\) by iteratively solving the following differential equations:
	if \(r \in \mathscr{C}'_{n,\ell+1}\), then
	\begin{align}
	\dv{t}\hat{\kappa}_t(r) = (M_{n})_{r,r}\hat{\kappa}_t(r) + \sum_{r' \in \mathscr{C}'_{\ell}} (M_{n})_{r,r'}\cf{r \succ r'} \hat{\kappa}_t(r')\, .
	\end{align}
	Here
	\begin{align}
	 (M_{n})_{r,r} &=
	  \sum_{i=1}^{L(r)}
	  \sum_{m=1}^{r_i} \binom{r_i}{m}
	 \sum_{k,h=0}^m
	 \cf{k+h\le m,\, h\text{ even}}
	 \left(\cf{0=m-k-\frac{h}{2}}+\cf{r_i=m-k-\frac{h}{2}}\right)
	\nonumber \\ & \quad\times
	 \frac{m!}{k!h!(m-k-h)!} (-1)^k  2^{h+1} I_{m-\frac{h}{2},\frac{h}{2}}
	 \nonumber \\
	 &=\sum_{i=1}^{L(r)}  \sum_{m=1}^{r_i}  \binom{r_i}{m} (-1)^{m} 2 I_{m, 0}  + \sum_{i=1}^n 2 I_{r_i,0}\,.
	\end{align}
Directly from the definition, it follows that
\[
  \sum_{m=0}^{r_i}  \binom{r_i}{m} (-1)^{m} I_{m, 0} = I_{0,r_i}\,.
\]
On the other hand, using the explicit formulae given in Appendix \ref{sec:trigonometric-integrals} we have
\[
 I_{0,m} = I_{m,0} = \frac{(2 m-1)!!}{(2m)!!}\,, \qquad m\in \N\,,
\]
which is a decreasing function of $m$ and has the particular values $I_{1,0}= \frac{1}{2}$ $I_{2,0}= \frac{3}{8}$, and $I_{3,0}= \frac{5}{16}$.  Therefore,
\[
  (M_{n})_{r,r} = 2 \sum_{i=1}^{L(r)} \left(2 I_{r_i,0}-1\right) \le -\frac{1}{2}\sum_{i=1}^{L(r)} \cf{r_i=2} - \frac{3}{4}\sum_{i=1}^{L(r)} \cf{r_i\ge 3}\,.
\]
 Thus, for every \(r \in \mathscr{C}_n'\), we have \((M_n)_{r,r} \leq -\frac{1}{2}\). Moreover, if \(r \neq (2,1,\dots, 1)\), then \((M_n)_{r,r} \leq  -\frac{3}{4}\).

	By using these differential equations as a starting point, we want to prove inductively in $\ell$ that
	\begin{align}
	N^{\alpha(\len(r)-1)}\abs{\hat{\kappa}_{t}(r)} \leq 10\rme^{-\frac{1}{2}t}\norm{\hat{\kappa}_0}_\alpha.
	\end{align}
	From this, the claim that the operator norm of the semigroup is bounded as
	 \begin{align}
	 \norm{\rme^{tM_n}}_\alpha \leq 10 \rme^{-\frac{1}{2}t}
	 \end{align}
	 follows immediately.
		
	For the base case \(r \in \mathscr{C}'_{n,2}\), we have only one partition classifier, \(r = (2,1,\dots, 1)\). In this case, clearly
	\begin{align}
	\hat{\kappa}_t(r) = \rme^{-\frac{1}{2}t}\hat{\kappa}_0(r)\,.
	\end{align}
  For the induction step, let \(r \in \mathscr{C}'_{n,\ell+1}\), $\ell\ge 2$.   The Duhamel formula implies that
	\begin{align}
	\abs{\hat{\kappa}_t(r)} \leq \rme^{(M_{n})_{r,r} t}\abs{\hat{\kappa}_0(r)} + \sum_{r' \in \mathscr{C}'_{n,\ell}} |(M_n)_{r,r'}| \int_{0}^{t} \rme^{(M_n)_{r,r}(t-s)} \abs{\hat{\kappa}_s(r')} \rmd s\,.
	\end{align}
  Applying the induction assumption and noting that $r\ne (2,1,1,\ldots,1,0)$, this implies
  \begin{align}
	\abs{\hat{\kappa}_t(r)}N^{\alpha(\len(r)-1)} &\leq \rme^{-\frac{3}{4}t}
	\norm{\hat{\kappa}_0}_\alpha
	 + \sum_{r' \in \mathscr{C}'_{n,\ell}} |(M_n)_{r,r'}| N^{\alpha(\len(r)-\len(r'))}
	 10 \norm{\hat{\kappa}_0}_\alpha
	 \int_{0}^{t} \rme^{-\frac{3}{4}(t-s)}
	 \rme^{-\frac{1}{2}s}\rmd s
	\nonumber \\
	&\leq \rme^{-\frac{1}{2}t}
	\norm{\hat{\kappa}_0}_\alpha \biggl(1
	 + 40 N^{-\alpha}
	 \sum_{r' \in \mathscr{C}'_{n,\ell}} |(M_n)_{r,r'}|\biggr)
	\,.
	\end{align}

	We estimate the remaining sum over the  absolute values of the matrix elements \((M_n)_{r,r'}\) by using the earlier explicit formulae.  This yields
	\begin{align}
	& \sum_{r' \in \mathscr{C}'_{n,\ell}} |(M_n)_{r,r'}|
	\le
	\sum_{i=1}^{L(r)}\sum_{a=0}^{r_i}
	\sum_{r' \in \mathscr{C}'_{n,\ell}}
	\cf{r' \sim{r-(a\times i)+(a\times (n+1))}}
	\nonumber \\ & \qquad\times
	  \sum_{m=1}^{r_i} \binom{r_i}{m}
	 \sum_{k,h=0}^m
	 \cf{k+h\le m,\, h\text{ even}}
	 \cf{a=m-k-\frac{h}{2}}
	 \frac{m!}{k!h!(m-k-h)!}  2^{h+1} |I_{m-\frac{h}{2},\frac{h}{2}}|
	 \nonumber \\
	&
	\quad
	\le
	\sum_{i=1}^{L(r)}
	  \sum_{m=1}^{r_i} \binom{r_i}{m}
	 \sum_{k,h=0}^m
	 \cf{k+h\le m}
	 \frac{m!}{k!h!(m-k-h)!}  2^{h+1}
	 \nonumber \\ & 	\quad
	\le
	2 \sum_{i=1}^{L(r)}
	  \sum_{m=1}^{r_i} \binom{r_i}{m}
	 4^m \le
	2 \sum_{i=1}^{L(r)} 5^{r_i} \le 2 n 5^n
	\le 2 n^* 5^{n^*}
	 \, .
	\label{eq:break_1_estimates}
	\end{align}
	Therefore, for any \(n \leq n^*\) and \(N \geq N_0(n^*, \alpha)\),
	\begin{align}
	\abs{\hat{\kappa}_t(r)}N^{\alpha(\len(r)-1)} &\leq \rme^{-\frac{1}{2}t}
	\norm{\hat{\kappa}_0}_\alpha \left(1
	 + 80 N^{-\alpha} n^* 5^{n^*}\right)
	\leq 10 \rme^{-\frac{1}{2}t}\norm{\hat{\kappa}_0}_\alpha
	\,.
	\end{align}
	This completes the induction step, and thus also the proof of the lemma.
	\end{proofof}
	
	\begin{lemma}
	For every \(\alpha \in (0,1)\), \(n^* \in \N\), and $\delta>0$, there exists \(N_0 = N_0(n^*,\alpha,\delta)\) such that
	\begin{align}
	\norm{R_{n,N}}_{\alpha} \leq \frac{\delta}{10},
	\end{align}
	for all \(n \in [2, n^*]\) and all \(N \geq N_0\).
	\label{lemma:perturbation-operator-norm-in-alpha}
	\end{lemma}	

	\begin{proofof}{Lemma \ref{lemma:perturbation-operator-norm-in-alpha}}
	Similar to the previous proof, the matrix elements \[(R_{n,N})_{r,r'} = \tilde{C}^{\text{break};2}_{r,r'} + \tilde{C}^{\text{fuse}}_{r,r'} + \tilde{C}^{\text{ex}}_{r,r'}\] satisfy
	\begin{align}
	\sum_{r' \in \mathscr{C}'_{n,\ell}}\abs{(R_{n,N})_{r,r'}} \leq \frac{n^2 C^n}{N-1}
	\,
	\end{align}
	Therefore,
	\begin{align}
	\norm{R_{n,N}}_\alpha & \leq \sup_{r} \sum_{r' \in \mathscr{C}_n'} \abs{(R_{n,N})_{r,r'}}N^{\alpha(\len(r)-\len(r'))} \nonumber \\
	& \leq n^{2}C^{n}\max(N^{-\alpha}, N^{-1}, N^{\alpha-1}) \nonumber \\
	&\leq (n^*)^2 C^{n^*} \max(N^{-\alpha}, N^{\alpha-1}).
	\end{align}
	
	Therefore, if both \(N^{-\alpha}\) and \(N^{\alpha - 1}\) are small enough, the result follows. In essence, the number of particles will have to be large enough to satisfy
	\begin{align}
	\log(N) \geq C'n^*
	\end{align}
	for a constant that is otherwise universal but depends on \(\alpha\) and $\delta$.
	\end{proofof}	
		
	We will also need the following lemma, whose proof can be found in \citep[pp. 495--496]{kato_perturbation_1995}.
	\begin{lemma}
	Let \(A,B \colon X \to X\) be bounded operators on a Banach space and $t>0$. If $M>0$ and $\beta\in \R$ are such that
	\begin{align}
	\norm{\rme^{sA}} \leq M\rme^{\beta s}, \quad \text{ for all } s \in (0,t),
	\end{align}
	then
	\begin{align}
	\norm{\rme^{t(A+B)}} \leq M \rme^{(\beta + M\norm{B})t}.
	\end{align}
	\end{lemma}

	Together, these three lemmas prove that for any given \(\delta > 0\), the full linear semigroup arising in the energy cumulant evolution will satisfy
	\begin{align}
	\norm{\rme^{tM_{n,N}}}_\alpha \leq 10\rme^{-(\frac{1}{2}-\delta)t},
	\label{ineq:linear-semigroup-delta}
	\end{align}
	for all \(n \in [2,n^*]\) and all \(N \geq N_0(n^*, \alpha, \delta)\).

	\paragraph*{Step 3: The full evolution}
	
	Having established the exponential decay for the linear part as well as for the completely non-repeated cumulants, it remains to prove that we can iteratively propagate the bounds over the non-linearity. This will be the main focus of this section.
	Fix \(n^* \in \N\) and \(\delta \in (0,1/8]\), and let \(N_0(n^*,\alpha,\delta)\) be large enough so that the statements in Lemmas \ref{lemma:limit-semigroup-convergence-in-alpha} and \ref{lemma:perturbation-operator-norm-in-alpha} are satisfied. In particular, \eqref{ineq:linear-semigroup-delta} holds. We now prove the claim by induction on the order of the cumulant \(n\). The base case \(n=2\) is clearly true.
	
	Assuming that the claim holds for all orders \(k <n\), let's consider cumulants of order \(n\). Using the notation 
	\begin{align}
	\kappa_t &= (\kappa_t^{n,N}[e_r])_{r \in \mathscr{C}'_n} \nonumber \\
	\mathcal{N}_{<,n}[\kappa_s] &= (\mathcal{N}_{<,n}[\kappa^{n,N}_s](e_r))_{r\in \mathscr{C}'_n}
	\end{align}
  	and 
  	\begin{align}
  	g_{n}(s)=(\cf{r = (2,1,\dots, 1)} \beta_{n,N} \kappa^{n,N}_t[e_{1_n}])_{r \in \mathscr{C}'_n},
  	\end{align}
  	with
  	\begin{align}
  	\beta_{n,N} &=
  	2 \frac{N+1-n}{N-1}\sum_{\ell=1}^{2} \binom{2}{\ell}
	 \sum_{k,h=0}^\ell
	 \cf{k+h\le \ell,\, h\text{ even}}
	 \cf{1=\ell-k-\frac{h}{2}}
	 \frac{\ell!}{k!h!(\ell-k-h)!} (-1)^k  2^h I_{\ell-\frac{h}{2},\frac{h}{2}}
	 \\ &
	  = 2 \frac{N+1-n}{N-1}
	  \left(2 I_{1,0}-2 I_{2,0}+4 I_{1,1}\right)
	  = \frac{3}{2} \frac{N+1-n}{N-1}\,.
  	\end{align}
  	This satisfies, for all allowed \(n, N\),
  	\begin{align}
  	\abs{\beta_{n,N}}\leq \frac{3}{2}\,.
  	\end{align}

  	Using the known semigroup bounds and the assumption about the initial data in the Duhamel formula implies that the solution satisfies the following estimate:
	\begin{align}
	\norm{\kappa_t}_\alpha \leq 10\rme^{-(\frac{1}{2}-\delta)t} B^n (n-1)! N^{\gamma_n} + 10 \int_{0}^{t} \rme^{-(\frac{1}{2}-\delta)(t-s)} \norm{\mathcal{N}_{<,n}[\kappa_s]}_\alpha \rmd s + 10\int_{0}^t  \rme^{-(\frac{1}{2}-\delta)(t-s)} \norm{g_n(s)}_\alpha \rmd s
	\end{align}
	By Proposition \ref{prop:non-repeated-stronger-claim-for-finite-order}, the last term satisfies, with $C_1=C(B)$,
	\begin{align}
	10 \int_{0}^{t} \rme^{-(\frac{1}{2}-\delta)(t-s)} \norm{g_n(s)}_\alpha \rmd s &\leq \frac{15}{(N-1)^{\alpha}}
	C_1^n (n-1)! \int_{0}^t  \rme^{-(\frac{1}{2}-\delta)(t-s)} (\rme^{-\frac{n}{4}s}N^{\gamma_n} + 1)
%
	\nonumber \\
	&\leq \frac{15}{(N-1)^{\alpha}}C_1^n(n-1)! \left(\rme^{-(\frac{1}{2}-\delta)t}N^{\gamma_n} + 4\right)\,.
	\end{align}

	Finally, we need to treat the non-linear term. The components of the vector \((\mathcal{N}_{<,n}[\kappa_s](r))_{r\in \mathscr{C}_n'}\) are given by
	\begin{align}
	\mathcal{N}_{<,n}[\kappa_s](r) = \mathcal{N}^{\text{break}}_{<,n}[\kappa_s](r) + \mathcal{N}^{\text{fuse}}_{<,n}[\kappa_s](r) + \mathcal{N}^{\text{ex}}_{<,n}[\kappa_s](r)\,.
	\end{align}
	Here, we use the notation \(\mathcal{N}_{<,n}^{\circ} \in X_{\alpha,n}\) to stand for the nonlinear parts, so that
	\begin{align}
\mathcal{N}_{<,n}^{\circ}[\kappa_s](r) = T^{\circ}_s(r) - (L^{\circ}\kappa_s)_r, \quad \circ \in \{\text{break}, \text{fuse}, \text{ex}\}\,.
	\end{align}
	
	Let us estimate the norm of these vectors individually. Note that
	by lemma \ref{lemma:nonlinearity-coloring} and
	\eqref{eq:mainbreakexpect}, the nonlinear contribution from the expectations appearing in \(\mathcal{N}_{<,n}^{\text{break}}[\kappa_s](r)\) are bounded by
	\[\frac{1}{(N-1)^{\alpha (\len(r)-1)}}\sum_{\pi \in \mathcal{P}([n])} \cf{\abs{\pi} \geq 2}\prod_{A \in \pi} \norm{\kappa^{\abs{A},N}_s}_{\alpha, \abs{A}}\,.\]
	Similarly, the expectations appearing in \(\mathcal{N}_{<,n}^{\text{fuse}}[\kappa_s](r)\) and \(\mathcal{N}_{<,n}^{\text{ex}}[\kappa_s](r)\), using \eqref{eq:mainexexpect}, are bounded by
	\begin{align}
	\frac{1}{(N-1)^{\alpha(\len(r)-2)}} \sum_{\pi \in \mathcal{P}([n])} \cf{\abs{\pi} \geq 2} \prod_{A \in \pi} \norm{\kappa_s^{\abs{A},N}}_{\alpha, \abs{A}}\,,
	\end{align}
	which has an additional factor of
	$(N-1)^\alpha$.  The extra factor will be compensated by the prefactor \(\frac{1}{N-1}\) coming from the weights in these two terms.
	For applications of lemma \ref{lemma:nonlinearity-coloring}, we note that if e.g.\ $\ell=s_i$, in which case the length of the first classifier is reduced by one, then $i$ is no longer shared between the two terms, so also $m$ is reduced by one.

	By the induction assumption, here
	\begin{align}
	\norm{\kappa_s^{\abs{A},N}}_{\alpha, \abs{A}} \leq C^{\abs{A}^2} \left(N^{\gamma_{\abs{A}}} \rme^{-\frac{1}{4}s} + 1\right) |A|!\,.
	\end{align}
	Consider then $\pi=(A_\ell)_{\ell=1}^k \in \mathcal{P}([n])$ and assume $k\ge 2$.
	Let $m_\ell=|A_\ell|$ denote the sizes of the clusters in the partition which satisfy $m_\ell\ge 1$ and $\sum_{\ell=1}^k m_\ell=n$.
	Thus $n^2-\sum_{\ell=1}^k m_\ell^2 =
	\sum_{\ell=1}^k m_\ell \sum_{\ell'\ne \ell} m_{\ell'}\ge \sum_{\ell=1}^k m_\ell =n$
	and $\sum_{\ell=1}^k \gamma_{m_\ell} = c(n-k)\le\gamma_n$.
	Hence, we obtain
	\begin{align}
	& \sum_{\pi \in \mathcal{P}([n])} \cf{\abs{\pi} \geq 2} \prod_{A \in \pi} \norm{\kappa_s^{\abs{A},N}}_{\alpha, \abs{A}}
	\nonumber \\ & \quad
	\le C^{n^2}C^{-n}  \left(N^{\gamma_{n}} \rme^{-\frac{1}{4}s} + 1\right)
	\sum_{\pi \in \mathcal{P}([n])} \cf{\abs{\pi} \geq 2} 2^{|\pi|}\prod_{A \in \pi} |A|!
	\nonumber \\ & \quad
	\le C^{n^2}C^{-n}  \left(N^{\gamma_{n}} \rme^{-\frac{1}{4}s} + 1\right) 2^n \rme^n n!\,,
	\end{align}
    where in the last step we have used the combinatorial estimate whose proof can be found e.g.\ from the proof of Lemma 7.3 in \cite{lukkarinen_weakly_2011}.  Adding the earlier estimates for the size and number of terms in each of the three cases, we thus obtain
	\begin{align}
	&10  \int_{0}^t \rme^{-\left(\frac{1}{2}-\delta\right)(t-s)} \norm{\mathcal{N}_{<,n}[\kappa_s]}_\alpha \rmd s
	 \nonumber \\
	&\leq 10 n^2 C^{-n}C^{n^2} (C')^n n! \int_{0}^t \rme^{-\left(\frac{1}{2}-\delta\right)(t-s)}
	 \left(N^{\gamma_{n}} \rme^{-\frac{1}{4}s} + 1\right)\rmd s
	 \nonumber \\
	&\leq 80 n^2 C^{-n}C^{n^2} (C')^n n!
	 \left(N^{\gamma_{n}} \rme^{-\frac{1}{4}t} + 1\right)\,,
	\end{align}
	where $C'$ is some pure constant.  Hence, assuming that $C$ is large enough, for example if $C\ge (C'+C_1(B)+B) 200 (n^*)^2$, we can sum up the derived estimates for the three terms and conclude the induction step.
	\end{proofof}	
	
	\begin{remark}
	The theorem might be explained in words as saying that the finite order cumulants exhibit generation of chaos as soon as the system has enough interacting particles. 
	Since for all \(N \geq N_0\), \(n \in [2,n^*]\), the norm
	\begin{align}
	\norm{\kappa_t}_{\alpha,n,N} \leq n! C^{n^2}\left(N^{\gamma_n} \rme^{-\frac{t}{4}} + 1\right),
	\end{align}
	it follows that
	\begin{align}
	\abs{\kappa_t(e_r)} \leq \frac{C(n,B)}{(N-1)^{\alpha(\len(r)-1)}}
	\end{align}
	as soon as \(t \geq 4\gamma_n \log(N)\). 
	\end{remark}

 \subsection{Convergence to stationary values for repeated energy cumulants}

{\em Proof of Theorem \ref{thm:convergence-to-stationary-non-chaotic}:}

\medskip

Given \(n, N\), let \((\bar{\kappa}^{n,N}[e_r])_{r \in \mathscr{C}'_n}\) correspond to the cumulants under the uniform measure on the sphere. This is a stationary solution to the cumulant hierarchy and thus it satisfies
\begin{align}\label{eq:kappabarprop}
(M_n + R_{n,N})\bar{\kappa} + \mathcal{N}_{<,n}[\bar{\kappa}] + \bar{s}_n = 0, \quad r \in \mathscr{C}'_n\,,
\end{align}
where the source term coming from non-repeated sequences has already been estimated in Sec.\ \ref{sec:nonrepequil}.
For fixed \(n\), these can be explicitly constructed, starting from \(n = 2\).

Let \(q^n_t(r) = q_t^{n,N}(r) \coloneqq \kappa^{n,N}_t[e_r]- \bar{\kappa}^{n,N}[e_r]\), and similarly for a sequence of particle labels \(A\), we let \(q^n_t(A) \coloneqq \kappa^{n,N}[e(t)_A] - \bar{\kappa}^{n,N}[e_A]\). Provided that there are enough particles in the system, we will see that the equilibrium cumulants are bounded as \(\norm{\bar{\kappa}^{n,N}[e_r]}_{\alpha} \leq n! C^{n^2}\). Indeed, we have already proven that
\begin{align}
\limsup_{t\to \infty} \norm{\kappa^{n,N}_t[e_r]}_\alpha \leq n! C^{n^2}.
\end{align}
In what follows, the aim is to prove that \(\norm{q_t^n}_\alpha \to 0\). To do this, we need to control the linear part of the system and then we need bounds like \(\norm{\bar{\kappa}^{m,N}[e_r]}_\alpha \leq m! C^{m^2}\) for all \(m < n\). Once we establish \(\norm{q_t^n}_\alpha \to 0 \), it then follows that
\begin{align}
\norm{\bar{\kappa}^{n,N}}_\alpha  \leq \norm{q_t^n}_\alpha + \norm{\kappa_t^{n,N}}_\alpha, \quad t \geq 0, 
\end{align} whereby \(\norm{\bar{\kappa}^{n,N}[e_r]}_\alpha \leq n! C^{n^2}\) as well. Therefore, we can prove the desired bounds for the stationary cumulants alongside the induction steps.

Let us now move to the more detailed analysis.
The goal is to prove that that there is a constant \(C\) and a minimal number of particles \(N_0\) such that for all \(N \geq N_0\), the estimate
\begin{align}
\norm{q^n_t}_{\alpha} \leq n! C^{n^2}N^{c(n-1)}\rme^{-\frac{t}{4}}
\end{align}
holds.

We recall that
\begin{align}
\dv{t}\kappa_t[e_r] = T_t^{\text{break}}(r) + T_t^{\text{fuse}}(r) + T_{t}^{\text{ex}}(r)\,,
\end{align}
where $T_t^{\text{break}}$, $T_t^{\text{fuse}}$, and $T_{t}^{\text{ex}}$, are defined in \eqref{eq:Ttbreak}, \eqref{eq:Ttfuse},
and \eqref{eq:Ttex}, respectively.

As shown earlier, each of these terms is a finite order polynomial of the cumulants.  More precisely, they are a linear combination of products of cumulants where the products are determined by some partition $\pi$ of the appropriate index set $R$.  Explicitly, for each partition $\pi$ the product is given by
\begin{align}
& \prod_{A \in \pi} \kappa_t[e_A] = \prod_{A \in \pi} (\bar{\kappa}+q_t)[e_A]
= \prod_{A \in \pi} \bar{\kappa}[e_A]
+ \sum_{\emptyset\ne\pi'\subset \pi}
\prod_{A \in \pi'} q_t(A)
\prod_{A \in \pi\setminus \pi'} \bar{\kappa}[e_A]\,.
\end{align}

The terms which contain no $q_t$ must sum to zero, by \eqref{eq:kappabarprop}, and the terms for which $|\pi|=1$ will result in the action of the earlier discussed linear operator on $q_t$.
If $|\pi|>1$, every $A\in \pi$ is lower order, so we can use the induction assumption to estimate their magnitude:
\begin{align}
\abs{\prod_{A \in \pi'} q_t(A)
\prod_{A \in \pi\setminus \pi'} \bar{\kappa}[e_A]} \leq \prod_{A \in \pi\setminus \pi'}\left(|A|! C^{\abs{A}^2} N^{-\alpha(\len(A) - 1)} \right)\prod_{A' \in \pi'}N^{-\alpha(\len(A')-1)}\norm{q_t}_{\alpha, \abs{A'}}
\end{align}
Consequently, since $|\pi'|\ge 1$,
\begin{align}
& N^{\alpha(\len(r)-1)}\abs{\prod_{A \in \pi'} q_t(A)
\prod_{A \in \pi\setminus \pi'} \bar{\kappa}[e_A]}
\nonumber \\ & \quad
\leq N^{\alpha(\abs{\pi} + \len(r)-1 -\sum_{A\in \pi}\len(A))}\rme^{-\frac{t}{4}} N^{c(n-1)}
\prod_{A \in \pi}\left(|A|! C^{\abs{A}^2}  \right)
\, .
\end{align}
Here \(\abs{\pi} + \len(r)-1-\sum_{A\in \pi}\len(A) \leq 1\) by Lemma \ref{lemma:nonlinearity-coloring}, and it is equal to $1$ only for terms which have an explicit prefactor $1/(N-1)$, as explained in the previous section.
The remaining sum over $\pi'$ can be estimated trivially by $2^{|\pi|}$.
Therefore, we can proceed as earlier, and conclude the induction step
along the same lines as in the proof of Theorem \ref{thm:generation-of-alpha-chaotic-bounds}.

	\section{Accuracy of the Boltzmann--Kac hierarchy}\label{sec:KacBEaccuracy}
	
As detailed in the Introduction, a key motivation for this study of the cumulant hierarchy of the stochastic Kac model is the question of applicability and accuracy of kinetic theory models, in general.
As derived already by Kac in his original work, for the present model and using our choice of units, the evolution of the probability density of a (randomly chosen) particle, \(f_t(v)\), could be approximated by solutions of the following
Boltzmann--Kac equation:
\begin{align}
\partial_t f_t(v)= 2\int_{-\pi}^{\pi}  \int_{\R}\left(f_t(v')f_t(w')-f_t(v)f_t(w)\right) \rmd w \frac{\rmd \theta}{2\pi}.
\end{align}
The velocities in the gain term are given
\begin{align}
v' &= v'(v,w;\theta) = v \cos(\theta) + w\sin(\theta) \\
w' &= w'(v,w;\theta) = v\sin(\theta) - w \cos(\theta),
\end{align} or equivalently in the matrix form:
\begin{align}
\bar{v}' = \begin{pmatrix}
v' \\
w'
\end{pmatrix} &= \begin{pmatrix}
\cos(\theta) &\sin(\theta) \\
-\sin(\theta) &\cos(\theta)
\end{pmatrix}\begin{pmatrix}
v \\
w
\end{pmatrix} = A(\theta)\bar{v}.
\end{align}
We recall that the matrix \(A(\theta)\) is orthogonal and equals the two-dimensional rotation present in $R_{i,j}(\theta)$.
Since the evolution equation is symmetric under swap of $v'$ and $w'$, one can replace above $A(\theta)$ by its transpose, as is also often used in the definition of the collision operator.


While the above differential equation describes the time-evolution of the density \(f_t\), it also allows studying evolution of general Radon probability measures $\mu_t$ on $\R$, via the dual description using $\inner{\phi}{\mu_t} := \int_\R\mu_t(\rmd v) \phi(v)$ for which
\begin{align}\label{eq:weakKacBE}
\dv{t} \inner{\phi}{\mu_t} = -2 \inner{\phi}{\mu_t} + 2\inner{\phi}{\mu_t \circ \mu_t}, \quad \phi \in C_c(\R).
\end{align}
The measure \(\mu_t \circ \mu_t\) is determined by mixing the product measure $\mu_t \otimes \mu_t$ and taking a marginal:
\begin{align}
\inner{\phi}{\mu_t \circ \mu_t} &= \int_{-\pi}^{\pi}\int_{\R^2}\phi((A(\theta)^{-1}\bar{v})_1) (\mu_t \otimes \mu_t) (\rmd \bar{v}) \frac{\rmd \theta}{2\pi} \nonumber \\
&= \int_{-\pi}^{\pi} \int_{\R} \int_{\R} \phi(\cos(\theta) v - \sin(\theta) w) \mu_t(\rmd v)\mu_t(\rmd w) \frac{\rmd \theta}{2\pi}.
\end{align}
This is clearly a positive linear functional of $\phi$ and, hence, $\mu_t \circ \mu_t$ is uniquely determined as a Radon measure by Riesz--Markov--Kakutani representation theorem.  One can also check that if $\mu_t(\R)=1$ and $\phi=1$, the right hand side vanishes, so it is consistent to require that a solution indeed would remain a probability measure.  Finally, assuming that the second moment is finite and choosing $\phi(v)=v^2$ in the above shows that also the mean energy is conserved: $\E_{\mu_t}[e_1] =\E_{\mu_0}[e_1]$.

We will here only consider such energy and probability conserving solutions to the equation (\ref{eq:weakKacBE}).  Although not worked out in full detail for the present Kac collision operator, it is mentioned in the beginning of Section 4 in
\cite{toscani_probability_1999} that such an evolution problem is well-posed: there exists a unique global weak solution for any initial data given by a non-negative probability measure with finite variance, and these solutions conserve normalization and energy.  For convenience, and to allow direct comparison of the techniques used above for the cumulant hierarchy, we assume also that $\mu_t$ is a solution to the above evolution equation which has exponential moments for the energy $v^2$, uniformly bounded as in  \eqref{eq:momgenbound}.  In fact, this is not a restriction, as we will prove in Lemma \ref{th:Kacexpmomlemma} at the end of this Section.

For solutions of this type, we can apply the argument used in Sec.\ \ref{sect:Cumulants-in-the-Kac}, and conclude that
\begin{align}
\dv{t}\log\E_{\mu_t}[\rme^{\xi \cdot v^2}] = 2\int_{-\pi}^{\pi} \frac{\rmd \theta}{2\pi} \E_{\mu_t\otimes \mu_t}[\rme^{\xi (-\sin(\theta)^2 e_1 + 2\cos(\theta)\sin(\theta)v_1 v_2 + \sin(\theta)^2 e_2)}G^{\mu_t\otimes \mu_t}_{w}(\xi; e_1)] -2\,.
\end{align}
Here we have used the fact that
\begin{align}
\E_{\mu_t \otimes \mu_t}[g(v)] = \E_{\mu_t}[g(v)]
\end{align}
for all continuous non-negative functions \(g\) to obtain the identity 
\(G_w^{\mu_t}(\xi; e_1) = G_w^{\mu_t \otimes \mu_t}(\xi; e_1)\).
Correspondingly, the cumulants of this hierarchy satisfy
\begin{align}\label{eq:Ftildecumulanthierarchy}
\dv{t}\kappa_t^n(e_1) = 2 \sum_{\ell=1}^n \binom{n}{\ell} \int_{-\pi}^{\pi} \frac{\rmd \theta}{2\pi}  \E_{\mu_t \otimes \mu_t}[P_\theta(v_1,v_2)^\ell \wick{e_1^{n-\ell}}],
\end{align}
where, as before,
\begin{align}
P_\theta(v_1,v_2) = -\sin(\theta)^2 e_1 + 2\cos(\theta)\sin(\theta) v_1v_2 + \sin(\theta)^2 e_2.
\end{align}
We recall that then, explicitly,
\begin{align}
&
 \int_{-\pi}^{\pi} \frac{\rmd \theta}{2\pi}
P_\theta(v_1,v_2)^\ell = \sum_{\ell_1=0}^{\ell} \sum_{\ell_2 = 0}^{\ell-\ell_1} \cf{\ell_2 \in 2\N} \binom{\ell}{\ell_1}\binom{\ell-\ell_1}{\ell_2}2^{\ell_2} (-1)^{\ell_1}
I_{\ell_2/2, \ell-\ell_2/2}\,
e_1^{\ell_1 + \ell_2/2}e_2^{\ell-\ell_1-\ell_2/2} \nonumber \\
& \quad = \sum_{a=0}^\ell C(\ell,a) e_1^{a}e_2^{\ell-a}
\,,
\end{align}
where for $0\le a\le \ell$, we have set
\[
C(\ell,a) := \sum_{\ell_1 = \max(0,2a-\ell)}^{a} (-1)^{\ell_1} 2^{2a-2\ell_1-2\ell}\frac{(2\ell+2\ell_1-2 a)!}{(\ell + \ell_1 - 2a)! \ell_1!(a-\ell_1)!(\ell-a+\ell_1)!}\,.
\]
One immediate consequence of the fact that we are taking cumulants with respect to the product measure \(\mu_t \otimes \mu_t\) is that in writing the expectations
\begin{align}
\E_{\mu_t \otimes \mu_t}[e_1^{\ell_1 + \ell_2/2}e_2^{\ell-\ell_1-\ell_2/2}\wick{e_1^{n-\ell}}]
\end{align}
in terms of the cumulants, all partitions \(\pi\) for which at least one of the clusters contains both \(e_1\) and \(e_2\) is not only small but even vanishes.  Hence, the rate function depends only on the one-particle cumulants.

Let us point out that the above rate function is very similar
to what we have obtained for completely repeated cumulants, i.e., the cumulants of the one-particle marginal distribution, in the Kac model.
Namely, for the completely repeated cumulant \(s = (n,0,\dots, 0)\), we have
	\begin{align}
	\dv{t}\kappa_t[e_s] &= T_t^{\text{break}}(s) \\
	&= \frac{N-1}{N-1}\sum_{\ell=1}^{n}\binom{n}{\ell}\int_{-\pi}^{\pi} \frac{\rmd \theta}{2\pi} \E_{F_t^N}[(P_\theta(v_1,v_2)^\ell)\wick{e_1^{n-\ell}}] \\
	& \quad + \frac{N-1}{N-1} \sum_{\ell=1}^{n}\binom{n}{\ell} \int_{-\pi}^{\pi} \frac{\rmd \theta}{2\pi} \E_{F_t^N}[(Q_\theta(v_1,v_2)^\ell)\wick{e_1^{n-\ell}}] \\
	&= 2 \sum_{\ell=1}^n \binom{n}{\ell} \int_{-\pi}^{\pi} \frac{\rmd \theta}{2\pi} \E_{F_t^N}[P_\theta(v_1,v_2)^\ell \wick{e_1^{n-\ell}}].
	\end{align}
where we have recalled $\int_{-\pi}^{\pi} \frac{\rmd \theta}{2\pi}
  Q_\theta(v_1,v_2)^\ell =  \int_{-\pi}^{\pi} \frac{\rmd \theta}{2\pi}
 P_\theta(v_1,v_2)^\ell $.
Since \(F_t^N\) is not factorized, the linear term includes cumulants of the form \(\kappa_t^n(e_1,\cdots, e_1,e_2,\cdots e_2)\). These will, however, be of lower order, and the rest of the linear part yields an operator which is identical to the one in the Boltzmann case.
We also recall that the prefactor of \(\kappa_t[e_s]\) on the right hand side is equal to $-a_n$ where
	\begin{align}
	a_n := 2\left(1-2 I_{n,0}\right), \quad I_{n,0} = \int_{-\pi}^{\pi} \frac{\rmd \theta}{2\pi} \sin(\theta)^{2n}= 2^{-2n} \binom{2n}{n}\,.
	\end{align}
	If $n\ge 2$, we have $I_{n,0} \le \frac{3}{8}$, and thus then $ a_n\ge \frac{1}{2}$.

	\begin{proofof}{Theorem \ref{thm:kinetic-accuracy}}	
	Having fixed the number of particles \(N\) and a symmetric initial data \(F_0^N\) we obtain the measures
	\(F_t^N\), $t\ge 0$, from the Kac process.  We first wait a time $t_0\ge 0$ and then use the first marginal of $F_{t_0}^N$ to define an initial measure $\mu_0$ for the Boltzmann--Kac equation for which $T=t-t_0\ge 0$ serves that as the time parameter.
	This is a Borel measure with exponential moments and thus, as explained above and using the results of \cite{toscani_probability_1999} and Lemma \ref{th:Kacexpmomlemma},
	there exists a unique global weak solution  $(\mu_T)_{T\ge 0}$ to the Boltzmann--Kac equation and its cumulants satisfy
	the evolution hierarchy in (\ref{eq:Ftildecumulanthierarchy}).
	We use these measures to define the symmetric product measures $\tilde{F}^N_T := \otimes_{i=1}^N \mu_T$ on $\R^N$,
	and denote its joint energy cumulants by
    \((\tilde{\kappa}^{n,N}_T[e_s])_{s \in \mathscr{C}_n, n\in \N}\).
	If $s$ is not fully repeating, we have $\tilde{\kappa}^{n,N}_T[e_s]=0$, and for a fully
	repeating cumulant of $n$ elements we have $\tilde{\kappa}^{n,N}_T[e_s]=\kappa_T^n(e_1)$ since all one-particle marginals of
	$\tilde{F}^N_T$ are equal to $\mu_T$.  The rate function of the fully repeating cumulants is thus given by the right hand side of (\ref{eq:Ftildecumulanthierarchy}).  We can even write it in a more familiar closed form for $\tilde{F}^N_T$ after we note that
    for any two-particle observable $g$ we have
    \[
     \E_{\mu_t \otimes \mu_t}[g(e_1,e_2)]=\E_{\tilde{F}^N_T}[g(e_1,e_2)]= \frac{1}{N(N-1)}\sum_{i,j=1}^{N} \cf{i\ne j}\E_{\tilde{F}^N_T}[g(e_i,e_j)]
     \,,
    \]
    which for polynomial observables $g$ may be expanded in terms of cumulants of $\tilde{F}^N_T$.

	To facilitate the comparison, let the joint cumulants of $F_{t_0+T}$ be denoted by
    \((\kappa^{n,N}_T[e_s])_{s \in \mathscr{C}_n, n\in \N}\), and
    the goal is to control
    \[
     \delta^{n,N}_T[e_s] \coloneqq
     \kappa^{n,N}_{t_0+T}[e_s] - \tilde{\kappa}^{n,N}_T[e_s]\,.
    \]

    If $s$ is not fully repeating, we have $\tilde{\kappa}^{n,N}_T[e_s]=0$, and thus an answer follows directly from the generation of chaos bounds in Theorem \ref{thm:generation-of-alpha-chaotic-bounds} assuming that the initial data $F^N_0$ satisfies Assumption \ref{assumption:chaos-bounds} for some constant $B\ge 1$,
    with fixed $c,\alpha,n^*$,
    and the number of particles is large enough, $N\ge N_0$.
    Namely, then there is a constant \(C\), depending only on \(B\), such that for all
	\(n \leq n^*\),
		\begin{align}
		     |\delta^{n,N}_T[e_s]|\le
		\norm{\kappa_{t_0+T}^{n,N}}_\alpha (N-1)^{-\alpha (\len(s)-1)}
		\leq (N-1)^{-\alpha} C^{n^2} n! \left(N^{c(n-1)}\rme^{-\frac{1}{4}(t_0+T)} + 1\right).
	\end{align}
	Now $t_0$ is assumed to be large enough so that $N^{c(n-1)}\rme^{-\frac{1}{4}t_0}\le 1$ for any $n\le n^*$.
	Then, using the generation of chaos bounds for $t=T+t_0\ge t_0$, we find
    \begin{align}
		     |\delta^{n,N}_T[e_s]|
		\leq 2 (N-1)^{-\alpha} C^{n^2} n! ,
	\end{align}
	for all $n\le n^*$ and $T\ge 0$, proving the claim in the Theorem for repeating $s$.

	It remains to control the case with $\len(s)=1$, i.e., the completely repeated cumulants.  Now the earlier estimates do not provide decay in $N$, instead we merely have a bound which is uniform in $T$ and $N$:
	\begin{align}
		     |\kappa^{n,N}_{t_0+T}[e_{s}]|
		     \leq C^{n^2} n! \left(N^{c(n-1)}\rme^{-\frac{1}{4}(t_0+T)} + 1\right) \le 2 C^{n^2} n! \,.
	\end{align}
    Our goal, and induction assumption, is to prove that nevertheless also for these completely repeating cumulants
    \begin{align}\label{eq:deltangoal}
		     |\delta^{n,N}_T[e_s]|
		\leq 2 (N-1)^{-\alpha} C^{n^2} n! \,.
	\end{align}
	The bound is trivially true in two cases: first, if $T=0$, then by construction the marginal of $F^N_0$ is $\mu_0$, and thus the initial cumulants coincide and thus $\delta^{n,N}_0[e_s]=0$ if $s$ is completely repeating.  In addition, both evolutions conserve mean energy, and thus
	$\delta^{n,N}_T[e_s]=0$ for all $T$ if $n=1$.  Hence, the induction assumption holds at $n=1$.

	Since the first marginal of $\tilde{F}^N_T$ is $\mu_T$, we have then $\tilde{\kappa}^{n,N}_T[e_s]=\kappa_T^n(e_1)$,
	and thus for $s=(n,0,\dots, 0)$,
	\[
	 \delta_{n,N}(T) \coloneqq \delta^{n,N}_T[e_s]
	 = \kappa^{n,N}_{t_0+T}[e_{s}]- \kappa_T^n(e_1)\,.
	\]
	At $T=0$, by construction, $\delta_{n,N}(T)=0$. For $T>0$,
    we can rely on the previous discussion, and conclude that these differences satisfy an evolution equation
    \[
    	\dv{T}\delta_{n,N}(T) =
    	2 \sum_{\ell=1}^n \binom{n}{\ell}
    	\sum_{a=0}^\ell C(\ell,a)
    	\left( \E_{F_{t_0+T}^N}[ \wick{e_1^{n-\ell}}e_1^{a}e_2^{\ell-a}]
    	-  \E_{\tilde{F}_{T}^N}[\wick{e_1^{n-\ell}}e_1^{a}e_2^{\ell-a} ]\right)\,.
    \]
    Given $n,a,\ell$, let $\mathcal{P}'(n,a,\ell)$ denote those partititions of the sequence $(e_1,\ldots,e_1,e_2,\ldots,e_2)$ of
    $(n-\ell+a) + (\ell-a)=n$ elements in which there is no cluster formed out of the first $n-\ell$ elements.  This removes all partitions forbidden by the above Wick regularization from the moments-to-cumulants formula, and we have
    \[
     \E_{F_{t_0+T}^N}[ \wick{e_1^{n-\ell}}e_1^{a}e_2^{\ell-a}]
    	-  \E_{\tilde{F}_{T}^N}[\wick{e_1^{n-\ell}}e_1^{a}e_2^{\ell-a}]
	= \sum_{\pi\in \mathcal{P}'(n,a,\ell)}
	\left(\prod_{A\in \pi} \kappa^{|A|,N}_{t_0+T}[e_A] -\prod_{A\in \pi} \cf{\len(A)=1} \kappa_T^{|A|}(e_1)  \right)\,.
    \]

	If $\pi\in \mathcal{P}'(n,a,\ell)$ is such that it contains some $A\in \pi$ with $\len(A)>1$, i.e., a not completely repeating cluster,
	then the second product yields zero but the corresponding term in the first product has an ``extra'' decay factor of $(N-1)^{-\alpha (\len(A)-1)}\le (N-1)^{-\alpha}$.

	If $\pi\in \mathcal{P}'(n,a,\ell)$ is such that $\len(A)=1$ for all $A\in \pi$, we telescope the terms in the second product by substituting
	there $\kappa_T^{|A|}(e_1)
	 = \kappa^{|A|,N}_{t_0+T}-  \delta_{|A|,N}(T) $.  This yields
	 \[
	 \prod_{A\in \pi} \kappa^{|A|,N}_{t_0+T} -\prod_{A\in \pi}  \kappa_T^{|A|}(e_1) = \sum_{\emptyset\ne \pi'\subset \pi} (-1)^{|\pi'|+1}
	 \prod_{A\in \pi\setminus \pi'} \kappa^{|A|,N}_{t_0+T} \prod_{A\in \pi'}  \delta_{|A|,N}(T)
	 \,.
	 \]
	 If $|\pi|=1$, there is only one term in the sum, with $\pi'=\pi$,
	 and in this case the right hand side equals $\delta_{n,N}(T)$.
	 However, since this requires that $\len(A)=1$ for $\pi=\{A\}$, we conclude that this term is present if and only if either $\ell=a$ and $a>0$,
	 or $\ell=n$ and $a=0$.
	 If $|\pi|>1$, the partition $\pi$ has at least two clusters all of which have a size strictly less than $n$ and at least one of which will produce a factor $\delta_{|A|,N}(T)$ with $|A|<n$.  Hence, these terms may be estimated using the induction assumption and the above mentioned generation of chaos bounds.  We can thus conclude that in this case
	 \[
	 \prod_{A\in \pi} \kappa^{|A|,N}_{t_0+T} -\prod_{A\in \pi}  \kappa_T^{|A|}(e_1) = \delta_{n,N}(T) \left(\cf{\ell=a,a>0}+\cf{\ell=n,a=0}\right)+ O((N-1)^{-\alpha})\,.
	 \]

    The terms with $a=0$ and $\ell=n$ also have $\ell_1=0$, and thus
    the these terms contribute a linear term
    \[
     2 I_{n,0}  \delta_{n,N}(T)
    \]
    in the rate function.
    The terms with $a=\ell>0$ have $C(\ell,a)=(-1)^\ell I_{\ell,0}$ and thus they contribute together a term
    \begin{align*}
      2 \sum_{\ell=1}^n \binom{n}{\ell}(-1)^\ell I_{\ell,0} \delta_{n,N}(T)\,.
	  \end{align*}
    Using the earlier proven result $\sum_{\ell=1}^n  \binom{n}{\ell}(-1)^\ell I_{\ell,0} = I_{n,0}-1$, we thus arrive at the equation
    \begin{align}
	\dv{T}\delta_{n,N}(T) =
	2(2I_{n,0}-1)\delta_{n,N}(T) + f_{n,N}(T)\,,
	\end{align}
	where $f_{n,N}(T) = O((N-1)^{-\alpha})$.  Thus if $n\ge 2$,
	\begin{align}
	\delta_{n,N}(t) &= \rme^{2(2I_{n,0}-1)t}\delta_{n,N}(0)+ \int_{0}^t \rme^{2(2I_{n,0}-1)(t-s)} f_{n,N}(s) \rmd s \\
	&= \int_{0}^t \rme^{2(2I_{n,0}-1)(t-s)} f_{n,N}(s) \rmd s\,,
	\end{align}
   with $2(2I_{n,0}-1)\le -\frac{1}{2}$.  After checking that the combinatorial factors related to $n$ work out as before, we can then conclude the induction step and, hence, the bound in (\ref{eq:deltangoal}) is valid for all $n\le n^*$, as claimed.
   \end{proofof}
   
   \begin{lemma}\label{th:Kacexpmomlemma}
   Let \(f_0\) be obtained from the first marginal of a symmetric probability measure on \(\Nsphere\). If \(f_t\) solves the Boltzmann--Kac equation with the initial data \(f_0\), then \(f_t\) has exponential moments and these satisfy
   $|\E_{f_t}[\rme^{\xi v^2} -1]|\le \frac{1}{2}$ for  any small enough $\xi\in \C$ and uniformly in $t\ge 0$.
   \end{lemma}
   \begin{proof}
   Since \(f_0\) is obtained from a symmetric probability measure on \(\Nsphere\), it satisfies
   \begin{align}
   \E_{f_0}[e(v)^{n}] \leq B^{n-1}\,, \qquad \E_{f_0}[\rme^{\xi e(v)}]\le \rme^{\xi B}\,,
   \end{align}
   at least for $B=N$ and any $\xi\ge 0$.
 To conclude the evolution equation for the moments, we can use a regularization argument: namely, for $M>0$, we may consider solutions to the weak equation
  \begin{align}\label{eq:regularizedBKac}
  & \dv{t} \inner{\phi}{\mu_t}
  \nonumber \\ &\quad
  =
   2\int_{-\pi}^{\pi} \int_{\R} \int_{\R} \left(\phi(\cos(\theta) v - \sin(\theta) w)-
   \cos^2(\theta) \phi(v) - \sin ^2(\theta) \phi(w)\right) \Psi(v^2+w^2- M)
    \mu_t(\rmd v)\mu_t(\rmd w) \frac{\rmd \theta}{2\pi}\,.
\end{align}
where $\Psi(x)$ is a smooth function approximating the step function $\cf{x\le 0}$; in particular, $\Psi(x)\in [0,1]$,
$\Psi(x)=1$, if $x\le -\frac{1}{2}$, and $\Psi(x)=0$ if $x\ge 0$.
It is clear that normalization and mean energy are conserved by such solutions, using $\phi(v)=1$ and $\phi(v)=v^2$.
Since for any $v^2+w^2\le M$ and $\theta$ we have $(\cos(\theta) v - \sin(\theta) w)^2\le v^2+w^2\le M$, solutions to this equation preserve finiteness of exponential energy moments; we note that for $\phi(v)=\rme^{\xi v^2}$, $\xi>0$, the absolute value of the integrand is bounded by $\rme^{\xi M}$.  As $M\to \infty$, the regularized solutions converge weakly to the solution of the original Boltzmann--Kac equation which we know to be unique.

In particular, we can now conlude the following evolution equation for energy moments of the Boltzmann--Kac solution:
\begin{align}
   & \dv{t}\E_{f_t}[e^n] = 2(-1 + \int_{-\pi}^{\pi} (\cos(\theta)^{2n} + \sin(\theta)^{2n})\frac{\rmd \theta}{2\pi})\E_{f_t}[e_n]
   \nonumber \\& \qquad
   + 2\sum_{k=1}^{n-1}\binom{2n}{2k} \int_{-\pi}^{\pi} \cos(\theta)^{2k} \sin(\theta)^{2n-2k} \frac{\rmd \theta}{2\pi} \E_{f_t}[e^{k}]\E_{f_t}[e^{n-k}]\,.
   \end{align}
 To prove that the energy moments are finite for all $t\ge 0$ and satisfy the above evolution equation, we can use the conservation of energy for $n=1$ and then, using as an induction assumption finiteness of moments up to order $n-1$, take the limit $M\to \infty$ in the Duhamel form which follows by exponenting all order $n$ terms in (\ref{eq:regularizedBKac}); an analogous argument will be used later to control their magnitude more precisely.

 We aim to propagate the upper bound \(\E_{f_t}[e^n] \leq C^{n-1}(n-1)!\) where the constant \(C\) depends on \(B\). This will prove propagation of exponential moments. To this end, let \(Q_n(t) = \frac{\E_{f_t}[e^n]}{C^{n-1}(n-1)!}\). We want to show that for all \(t\), we have \(\abs{Q_{n}(t)} = Q_n(t) \leq 1\) once \(C\) is picked suitably.
   
   For the base case $n=1$ this is clear, no matter what \(C\). Assuming that the assumption holds for all \(1\le k < n\), we now consider the case \(n\). Using \eqref{eq:explicitIab}, we obtain
   \begin{align}
   \dv{t}Q_n(t) &= -a_n Q_n(t) + \sum_{k=1}^{n-1} \frac{(2n)!}{(2k)!(2n-2k)!}\frac{(k-1)!(n-k-1)!}{(n-1)!} \frac{C^{k-1}C^{n-k-1}}{C^{n-1}} I_{k,n-k} Q_{k}(t)Q_{n-k}(t) \nonumber \\
   &= -a_n Q_{n}(t) + \frac{2}{C}\frac{(2n)!}{2^{2n}n!(n-1)!} \sum_{k=1}^{n-1} \frac{1}{k(n-k)} Q_k(t)Q_{n-k}(t).
   \end{align}
   Thus,
   \begin{align}
   Q_n(t) &= \rme^{-a_n t}Q_n(0) + \frac{2}{C} 2^{-2n}\binom{2n}{n} \sum_{k=1}^{n-1} \frac{n}{k(n-k)} \int_{0}^t \rme^{-a_n (t-s)} Q_k(s)Q_{n-k}(s) \rmd s\,,
   \end{align}
   and by positivity, induction assumption, and the known Erd\H{o}s upper bound $\binom{2n}{n}< 2^{2n}\frac{1}{\sqrt{\pi n}}$, we obtain the following upper bound for \(n \geq 2\)
   \begin{align}
   Q_n(t) &\leq \rme^{-a_n t}\frac{B^{n-1}}{C^{n-1} (n-1)!} + \frac{2}{C} \frac{1}{\sqrt{n}} \sum_{k=1}^{n-1} \frac{n}{k(n-k)} \int_{0}^t \rme^{-a_n (t-s)}  \rmd s\,.
   \end{align}
   Here, the factor $\sum_{k=1}^{n-1} \frac{n}{k(n-k)} $ forms a Riemann sum, and is bounded by $2 \int_{1/n}^{\frac{1}{2}}\rmd x \frac{1}{x(1-x)}\le 4 \ln \frac{n}{2}$.  Since also $\int_{0}^t \rme^{-a_n (t-s)}  \rmd s \le \frac{1}{a_n}\le 2$, for $n\ge 2$ and $t\ge 0$, we may conclude a bound
   \begin{align}
   Q_n(t) &\leq (B/C)^{n-1} + \frac{1}{C} \frac{2^4}{\sqrt{n}} \ln \frac{n}{2} \,.
   \end{align}
   If $C\ge \max_{n\ge 2}\left(2B,\frac{2^5}{\sqrt{n}} \ln \frac{n}{2}\right)  < \infty$, we thus have $Q_n(t) \le 1$ for all $t$, which concludes the induction argument.  We point out that for large enough $N$ it suffices to pick \(C = 2N\). The moments of the time-evolved solution \(f_t\) then satisfy
   \begin{align}
   \E_{f_t}[e^n] \leq (2N)^{n-1}(n-1)! \,,
   \end{align}
   whereby we have exponential moments globally in time.  For example, with any $0\le \xi< \frac{1}{2N}$, it follows that $\E_{f_t}[\rme^{\xi v^2}]\le 1+ \frac{1}{2N}\ln (1-2N\xi)^{-1}$ for all $t\ge 0$.  Since $N\ge 2$, we have
   $0\le \E_{f_t}[\rme^{\xi v^2} -1]\le \frac{1}{2N}\ln (1-2N\xi)^{-1}\le \frac{1}{2}$ for any
	$\xi\le \frac{1}{4N}$, implying $|\E_{f_t}[\rme^{\xi v^2} -1]|\le \E_{f_t}[\rme^{|\xi| v^2} -1]\le \frac{1}{2}$, if
	$\xi\in \C$ with
	$|\xi|\le  \frac{1}{4N}$.
   \end{proof}

\appendix

	\section{Joint moments, cumulants, and Wick polynomials}\label{sec:cumulantsandWick}

	Let \((S, \mathcal{B}, \mu)\) be a probability space and \(Y \colon S \to \R\) a random variable. For simplicity, let us assume that it has moments of all order, even exponential moments\footnote{Meaning that there exists \(0<r<+\infty \) such that \(\E[\rme^{\lambda \abs{Y}}]\) exists for all \(\lambda \in (-r,r)\).}.  The cumulants of \(Y\) describe the higher order variance structure of \(Y\) with respect to the randomness \(\mu\). If \(\kappa_n[Y]\) stands for the \(n\)th order cumulant of \(Y\) with respect to \(\mu\), then the first two cumulants are given by
	\begin{align}
	\kappa_1[Y] &= \E_\mu[Y] \\
	\kappa_2[Y] &= \E_{\mu}[Y^2] - \E_{\mu}[Y]^2.
	\end{align}
	These are the mean and the variance of the random variable \(Y\).

	The higher order cumulants will have more complicated expressions. In general, the \(n\)th cumulant \(\kappa_n[Y]\) can be written in terms of order \(\leq n \) moments of \(Y\):
	\begin{align}
	\kappa_n[Y] \coloneqq \sum_{\pi \in \mathcal{P}([n])} (-1)^{\abs{\pi}-1}(\abs{\pi}-1)! \prod_{A \in \pi} \E[Y^{\abs{A}}].
	\end{align}
	Here \(\pi \in \mathcal{P}([n])\) means that \(\pi = \{A_1,A_2,\dots, A_\ell\}\) is a partition of the set \([n] = \{1,2,\dots, n\}\). For instance, in computing the variance, there are only two possible partitions of \(\{1,2\}\), namely \(\{\{1,2\}\}\) and \(\{\{1\}, \{2\}\}\).

	Provided that the random variable \(Y\) has exponential moments, the cumulants can equivalently be defined in terms of the Taylor coefficients of the locally (near zero) analytic function \(\xi \mapsto \log(\E[\rme^{\xi Y}])\). Here we have taken the principal branch for the complex logarithm. The cumulants can now be obtained as the Taylor coefficients of the generating function:
	\begin{align}
	\kappa_n[Y] = \partial_{\xi}^n \log(\E[\rme^{\xi Y}])\vert_{\xi = 0}.
	\end{align}

	\begin{remark}
	A well-known fact about cumulants is that the only random variables with only finitely many non-zero cumulants are Gaussians. Moreover, each Gaussian is specified by its only non-zero cumulants: its mean and variance. This suggests that the level of Gaussianity of a random variable can be gauged by studying how small its cumulants of order \(\geq 3\) are (for further results along this direction, see \cite{doring_method_2022}). If, for a sequence of random variables, the corresponding first and second cumulants converge to a pair of fixed constants \(\mu, \sigma^2\) while the rest of the cumulants go to zero, then the random variables converge in distribution to the Gaussian \(\mathcal{N}(\mu,\sigma^2)\).
	\end{remark}

	The notion of cumulants can be generalized from random variables to an indexed family of random variables. This gives rise to the notion of \emph{joint cumulants}. Assuming that we have a finite collection of random variables \(\{Y_{i}\}_{i=1}^n\) and that their joint probability distribution has exponential moments, i.e. that there exists some constant \(\delta >0 \) such that
	\begin{align}
	\E\left[\rme^{\delta \sum_{i=1}^n \abs{Y_i}}\right]<+\infty,
	\end{align}
	the joint cumulants can be defined through the Taylor coefficients of the joint cumulant generating function, which is defined, in a suitably small ball around \(0 \in \R^n\), as
	\begin{align}
	g_c(\xi) \coloneqq \log \E[\rme^{\xi \cdot Y}], \quad \xi \in B(0, r_n).
	\end{align}
	Here \(Y\) is the random vector \(Y = (Y_1,\dots, Y_n)\).

	The joint cumulants can then be obtained from the multidimensional joint cumulant generating function using the formula
	\begin{align}
	\kappa[Y_I] \coloneqq \left(\partial^{I}_{\xi} g_c(\xi)\right)\vert_{\xi =0}.
	\end{align}
	Here \(I \colon [m] \to [n]\) is a sequence that picks out \(m\) random variables from the collection \(\{Y_i\}_{i=1}^n\), with repeated random variables allowed. The differential operator is defined as \(\partial^{I}_{\xi} \coloneqq \prod_{i=1}^{n}\partial_{\xi_i}^{I_{i}}\), and finally \(Y_I\) stands for the sequence of random variables \(Y_I = (Y_{I_i})_{i=1}^{m}\). The assumption of joint exponential moments guarantees that everything works, as it implies that the function \(\xi \mapsto g_c(\xi)\) is analytic in a small neighborhood of \(0\).

	Under the same joint moment assumptions, the Wick polynomial generating function of the collection \(\{Y_{i}\}_{i=1}^n\) may be defined as
	\begin{align}
	G(Y;\xi) \coloneqq \rme^{\xi \cdot Y-g_c(\xi)}.
	\end{align}
	Wick polynomials related to the random variables \(\{Y_1,\dots, Y_n\}\) are polynomial expressions of the random variables, with the coefficients of the individual monomials given by the underlying randomness.

	Recall that the following truncated moments-to-cumulants formula holds (see \cite{lukkarinen_wick_2016} for the proof and further details)
	\begin{align}
	\E[\wick{Y_J} Y^I] = \sum_{\pi \in \mathcal{P}(J \cup I)} \prod_{A \in \pi} \cf{A \cap I \neq \emptyset}\kappa(Y_A).
	\label{eq:truncated-moments-to-cumulants}
	\end{align}
	In particular, if $I$ has only one element, the partition can only contain the entire set $J\cup I$.  Therefore, $\E[\wick{Y_J} Y_i]=
	\kappa[Y_{J\cup \set{i}}]$.
	The notation \(\mathcal{P}(I \cup J)\) stands for the collection of partitions of \(I \cup J\). Sets \(A \in \pi\) in the partition of \(J \cup I\) failing to meet the condition \(A \cap I \neq \emptyset\) will here be called \emph{Wick-internal} sets. Note that the difference between \eqref{eq:truncated-moments-to-cumulants} and the full moments-to-cumulants formula is that in the former all partitions involving a Wick internal set are removed from the right hand side. As we see in the text, this can have a regularizing effect.
	
	\section{Trigonometric integrals}\label{sec:trigonometric-integrals}
	Let
	\begin{align}
	I_{a,b} \coloneqq \int_{-\pi}^{\pi} \frac{\rmd \theta}{2\pi} (\sin(\theta)^2)^{a}(\cos(\theta)^2)^b\,, \qquad a,b\ge 0.
	\end{align}
	Using $\sin(\theta)^2=1-\cos(\theta)^2$ and integration by parts, we can check that
	these values satisfy the following recursion relations
	\begin{align}
	I_{a,b} &= \frac{2a-1}{2a+2b}I_{a-1,b}, \quad &a \geq 1, b \geq 0\\
	I_{0,b} &= \frac{2b-1}{2b}I_{0, b-1}, \quad &b\geq 1
	\end{align}
	together with the boundary condition
	\begin{align}
	I_{0,0} = 1.
	\end{align}
	It follows that for any $a,b\in \N_0$, with the conventions $(-1)!! \coloneqq 1$ and $0!! \coloneqq 1$,
	\begin{align}\label{eq:explicitIab}
	I_{a,b} &= \frac{(2a-1)!!(2b-1)!!}{(2a+2b)!!} = \frac{(2a)!(2b)!}{2^a2^b a! b! 2^{a+b}(a+b)!} = \frac{(2a)!(2b)!}{2^{2(a+b)}a!b!(a+b)!} \nonumber \\
	&= \frac{1}{\pi}\frac{\Gamma(a+\frac{1}{2})\Gamma(b+\frac{1}{2})}{\Gamma(a+b+1)}.
	\end{align}

\bibliographystyle{plainnat}

\bibliography{Kac-model}

	\end{document}